%% file: paper.tex
\documentclass[11pt]{article}
\input{preamble}

\title{Understanding the Moments of Tabulation Hashing via Chaoses} 

\author{Jakob Bæk Tejs Houen\footnote{Research supported by Investigator Grant 16582, Basic Algorithms Research Copenhagen (BARC), from the VILLUM Foundation.}\ }

\author{Mikkel Thorup$^*$}

\affil{University of Copenhagen,\\
    \tt{\{jakn,mthorup\}@di.ku.dk}
}
\date{}

\begin{document}

\maketitle

\input{abstract}

\newpage
\tableofcontents
\newpage

\section{Introduction}
\input{introduction/introduction}
\input{introduction/results}
\input{introduction/technical-overview}
\input{introduction/related-work}

\section{Preliminaries}
\input{introduction/preliminaries}

\section{Moment Inequalities}
\input{auxiliary/intro}
\input{auxiliary/poisson-moments}
\input{auxiliary/general-moments}
\input{auxiliary/tails-to-moments}
\input{auxiliary/martingale-decoupling}

\section{Strong Concentration for Tabulation Hashing}
\input{concentration/intro}
\input{concentration/general-facts}
\input{concentration/simple-tabulation}
\input{concentration/mixed-tabulation}

\section{Lower Bounds}
\input{concentration/lower-bound}

\section{Adding a Query Element}
\input{concentration/query-element}

\section{Acknowledgement}Research supported by Investigator Grant 16582, Basic Algorithms Research Copenhagen (BARC), from the VILLUM Foundation.

\bibliography{paper}

\appendix
\input{appendix/mixed.tex}

\end{document}

%% file: preamble.tex
\usepackage{geometry}
\usepackage[utf8]{inputenc}
\usepackage{stringstrings}

\usepackage{fullpage}

\usepackage{amsmath}
\usepackage{amsfonts}
\usepackage{amssymb}
\usepackage{amsthm}
\usepackage{bm}
\usepackage{bbm}
\usepackage{dsfont}

\usepackage{tikz} 
\usetikzlibrary{patterns}
\usepackage{subcaption}

\usepackage[ruled,noresetcount]{algorithm2e}

\usepackage{hyperref}
\usepackage{cleveref}

\usepackage{mathtools}
\usepackage{multirow}

\usepackage{environ}

\usepackage{ifthen}

\usepackage{pdfpages}

\usepackage{authblk}

\newtheorem{theorem}{Theorem}

\newtheorem{corollary}[theorem]{Corollary}
\newtheorem{lemma}[theorem]{Lemma}

\theoremstyle{definition}
\newtheorem{definition}[theorem]{Definition}

\newtheorem{remark}[theorem]{Remark}

\DeclareMathOperator*{\E}{E}

\DeclareMathOperator*{\Var}{Var}

\newcommand\req[1]{(\ref{#1})}

\newcommand{\N}{\mathbb N}
\newcommand{\R}{\mathbb R}

\newcommand{\Z}{\mathbb Z}

\newcommand{\eps}{\varepsilon}

\newcommand{\abs}[1]{\left\lvert {#1} \right\rvert}

\newcommand{\norm}[2]{\left\lVert {#1} \right\rVert_{{#2}}}
\newcommand{\pnorm}[2]{\left\lVert {#1} \right\rVert_{{#2}}}
\newcommand{\pnormcond}[3]{\left\rVert {#1} \; \middle\vert \; {#3} \right\lVert_{{#2} }}
\newcommand{\ep}[2][]{\E_{{#1}}\!\left[ {#2} \right]}
\newcommand{\epcond}[3][]{\E_{{#1}}\!\left[ {#2} \, \middle| \, {#3} \right]}
\newcommand{\prb}[2][]{\Pr_{#1}\!\left[ #2 \right]}
\newcommand{\prbcond}[3][]{\Pr_{{#1}}\!\left[ {#2} \, \middle| \, {#3} \right]}

\newcommand{\varcond}[2]{\Var\!\left[ {#1} \, \middle| \, {#2} \right]}
\newcommand{\indicator}[1]{\left[ {#1} \right]}

\newcommand{\set}[1]{\left\{ {#1} \right\}}
\newcommand{\setbuilder}[2]{\left\{ {#1} \, \middle| \, {#2} \right\}}

\newcommand{\ceil}[1]{\left\lceil {#1} \right\rceil}
\newcommand{\floor}[1]{\left\lfloor {#1} \right\rfloor}

\newcommand{\xor}{\oplus}
\newcommand{\bigxor}{\bigoplus}

\newcommand{\powerset}[1]{\mathcal{P}\left( {#1} \right)}

\DeclarePairedDelimiterX{\infdivx}[2]{(}{)}{{#1} \,\delimsize\|\, {#2}}

\NewEnviron{splitEquation}{\begin{equation}\begin{split} \BODY \end{split}\end{equation}}
\NewEnviron{splitEquation*}{\begin{equation*}\begin{split} \BODY \end{split}\end{equation*}}

\makeatletter
\providecommand*{\cupdot}{%
  \mathbin{%
    \mathpalette\@cupdot{}%
  }%
}
\newcommand*{\@cupdot}[2]{%
  \ooalign{%
    $\m@th#1\cup$\cr
    \hidewidth$\m@th#1\cdot$\hidewidth
  }%
}
\providecommand*{\bigcupdot}{%
  \mathop{%
    \mathpalette\@bigcupdot{}%
  }%
}
\newcommand*{\@bigcupdot}[2]{%
  \ooalign{%
    $\m@th#1\bigcup$\cr
    \hidewidth$\m@th#1\cdot$\hidewidth
  }%
}
\makeatother

\makeatletter
\newcommand{\bigmodeproduct}[1]{
  \mathop{
    \mathchoice{\vcenter{\hbox{{\huge $\m@th\mkern-2mu\times\mkern-2mu$}$_{#1}$}}}          
               {\vcenter{\hbox{{\LARGE $\m@th\mkern-2mu\times\mkern-2mu$}$_{#1}$}}}         
               {\vcenter{\hbox{{$\m@th\mkern-2mu\times\mkern-2mu$}$_{#1}$}}}                
               {\vcenter{\hbox{{\footnotesize $\m@th\mkern-2mu\times\mkern-2mu$}$_{#1}$}}}  
  }\displaylimits
}
\makeatother

\newcommand{\restateable}[4][theorem]{
  \ifthenelse{\equal{#1}{theorem}}{
    \begin{theorem}\label{thm:#2}
      {#3}
    \end{theorem}

    \convertchar[e]{#2}{-}{ }
    \capitalizewords[e]{\thestring}
    \noblanks[e]{\thestring}
    \retokenize[q]{\thestring}

    \expandafter\newcommand\csname restateThm\thestring\endcsname{
      {
        \def\thetheorem{\ref{thm:#2}}
        \begin{theorem}
          {#4}
        \end{theorem}
        \addtocounter{theorem}{-1}
      }
    }%
  }{\ifthenelse{\equal{#1}{lemma}}{
    \begin{lemma}\label{lem:#2}
      {#3}
    \end{lemma}

    \convertchar[e]{#2}{-}{ }
    \capitalizewords[e]{\thestring}
    \noblanks[e]{\thestring}
    \retokenize[q]{\thestring}

    \expandafter\newcommand\csname restateLem\thestring\endcsname{
      {
        \def\thetheorem{\ref{lem:#2}}
        \begin{lemma}
          {#4}
        \end{lemma}
        \addtocounter{theorem}{-1}
      }
    }%
  }{}}
}

%% file: abstract.tex
\begin{abstract}

    Simple tabulation hashing dates back to Zobrist in 1970 and is defined as follows:
    Each key is viewed as $c$ characters from some alphabet $\Sigma$, we have $c$ fully random hash functions $h_0, \ldots, h_{c - 1} \colon \Sigma \to \set{0, \ldots, 2^{l} - 1}$, and a key $x = (x_0, \ldots, x_{c - 1})$ is hashed to $h(x) = h_0(x_0) \xor \ldots \xor h_{c - 1}(x_{c - 1})$ where $\xor$ is the bitwise XOR operation.    
    The previous results on tabulation hashing by P{\v a}tra{\c s}cu and Thorup~[J.ACM'11] and by Aamand et al.~[STOC'20] focused on proving Chernoff-style tail bounds on hash-based sums, e.g., the number keys hashing to a given value, for simple tabulation hashing, but their bounds do not cover the entire tail. 
    Thus their results cannot bound moments. 
    The paper Dahlgaard et al.~[FOCS'15] provides a bound on the moments of certain hash-based sums, but their bound only holds for constant moments, and we need logarithmic moments.

    Chaoses are random variables of the form $\sum a_{i_0, \ldots, i_{c - 1}} X_{i_0} \cdot \ldots \cdot X_{i_{c - 1}}$ where $X_i$ are independent random variables.
    Chaoses are a well-studied concept from probability theory, and tight analysis has been proven in several instances, e.g., when the independent random variables are standard Gaussian variables and when the independent random variables have logarithmically convex tails.
    We notice that hash-based sums of simple tabulation hashing can be seen as a sum of chaoses that are not independent.
    This motivates us to use techniques from the theory of chaoses to analyze hash-based sums of simple tabulation hashing.

    In this paper, we obtain bounds for all the moments of hash-based sums for simple tabulation hashing which are tight up to constants depending only on $c$.
    In contrast with the previous attempts, our approach will mostly be analytical and does not employ intricate combinatorial arguments.
    The improved analysis of simple tabulation hashing allows us to obtain bounds for the moments of hash-based sums for the mixed tabulation hashing introduced by Dahlgaard et al.~[FOCS'15]. With simple tabulation hashing, there are certain inputs for which
    the concentration is much worse than with fully random hashing. However, with mixed tabulation, we get logarithmic moment bounds
    that are only a constant factor
    worse than those with fully random hashing for any possible input.
    This is a strong addition to other powerful
    probabilistic properties of mixed tabulation hashing
    proved by Dahlgaard et al.
\end{abstract}

%% file: introduction/introduction.tex

Hashing is a ubiquitous tool of randomized algorithms which dates all the way back to the 1950s~\cite{dumey1956}.
A hash function is a random function, $h \colon U \to R$, that assigns a random hash value, $h(x) \in R$, to every key, $x \in U$.
When designing algorithms and data structures, it is often assumed that one has access to a uniformly random hash function that can be evaluated in constant time.
Even though this assumption is very useful and convenient, it is unfortunately also unrealistic.
It is thus a natural goal to find practical and efficient constructions of hash functions that provably have guarantees akin to those of uniformly random hashing.

If we want implementable algorithms with provable performance similar to
that proven assuming uniformly random hashing, then we have to find
practical and efficient constructions of hash functions with
guarantees akin to those of uniformly random hashing. An example of
this is simple tabulation hashing introduced by Zobrist in
1970~\cite{zobrist70}. The scheme is efficient and easy to implement,
and P\v{a}tra\c{s}cu and Thorup~\cite{Patrascu2012} proved that it
could replace uniformly random hashing in many algorithmic
contexts. The versatility of simple tabulation does not stem from a
single probabilistic power like $k$-independence (it is only
3-independent), but from an array of powers that have different
usages in different applications. Having one hash function with multiple powers
has many advantages. One is that we can use the same hash function
implementation for many purposes. Another is that hash functions are
often an inner-loop bottleneck, and then it is an advantage if the
same hash value can be used for multiple purposes. Also, if we have
proved that a simple hash function has some very different
probabilistic properties, then, morally, we would
expect it to possess many other properties to be
uncovered as it has happened over the years for simple tabulation (see, e.g.,
\cite{AamandKT18power-of-choice,AamandT19nonempty}). Finally, when we
hash a key, we may not even know what property is needed, e.g., with
weighted keys, we may need one property to deal with a few heavy keys,
and another property to deal with the many
light keys, but when we hash the key, we may not know if it is heavy
or light.

One of the central powers proved for simple tabulation in~\cite{Patrascu2012}
is that it has strong concentration bounds for hash-based sums (will be defined
shortly in Section \ref{sec:hash-based-sums}). The concentration holds only 
for quite limited expected values, yet this suffices for important
applications in classic hash tables. Recently, Aamand et al. \cite{aamand2020}
introduced tabulation-permutation, which is only about twice as slow
as simple tabulation, and which offers general concentration bounds
that hold for all hash-based sums regardless of the expected size.
An issue with tabulation-permutation is that it is not clear if it
possesses the other strong powers of simple tabulation.

A different way to go is to construct increasingly strong schemes, each
inheriting all the nice properties of its predecessors. In
this direction, \cite{PatrascuT13twisted} introduced twisted tabulation
strengtening simple tabulation, and \cite{Dahlgaard2015} introduced
mixed tabulation strengthening twisted tabulation. Each new scheme was
introduced to get some powers not available with the predecessor. In
particular, mixed tabulation has some selective full-randomness that
is needed for aggregating statistics over hash-based $k$-partitions.
These applications also needed concentration bounds for hash-based
sums, but \cite{Dahlgaard2015} only provided some specialized suboptimal
concentration bounds.

In this paper, we do provide strong concentration bounds for mixed
tabulation hashing which can then be used in tandem with all the other
strong properties of simple, twisted, and mixed tabulation.  In fact
our bounds are more general than the strong concentration bounds
proved in \cite{aamand2020} for tabulation-permutation. More
precisely, the concentration bounds in \cite{aamand2020} are
Chernoff-style tail bounds that hold with high probability, while what
we do is to show moment bounds that imply such tail bounds as 
special cases. Indeed the key to our results for mixed tabulation is a
much stronger understanding of the moments of simple
tabulation.

Below we proceed to describe our new mathematical understanding, including
the relevance of chaoses. We will contextualize this with other work later
in Section \ref{sec:context}.

\subsection{Moment bounds for hash-based sums}\label{sec:hash-based-sums}

In this paper, we will focus on analyzing hash-based sums.
More precisely, we consider a fixed \emph{value function}, $v \colon U \times R \to \R$, and define the random variable $X_x = v(x, h(x))$ for every key $x \in U$.
We are then interested in proving concentration bounds for the sum $X = \sum_{x \in U} X_x = \sum_{x \in U} v(x, h(x))$.
It should be noted that the randomness of $X$ derives from the hash function $h$, thus the results will depend on the strength of $h$.

This is quite a general problem, and at first glance, it might not be obvious why this is a natural construction to consider, but it does generalize a variety of well-studied constructions:
\begin{enumerate}
    \item Let $S \subseteq U$ be a set of balls and assign a weight, $w_x \in \R$, for every ball, $x \in S$. 
    The goal is to distribute the balls, $S$, into a set of bins $R = [m]$.\footnote{For a positive integer $m \in \N$ we define $[m] = \set{0, \ldots, m - 1}$.}
    For a bin, $y \in [m]$, we define the value function $v_y \colon U \times [m] \to \R$ by $v_y(x, j) = w_x \indicator{j = y} \indicator{x \in S}$,
    then $X = \sum_{x \in U} v_y(x, h(x)) = \sum_{x \in S} w_x \indicator{h(x) = y}$ will be the weight of the balls hashing to bin $y$.\footnote{For a statement $P$ we let $[P]$ be 1 if $P$ is true and $0$ otherwise.}
    \item Instead of concentrating on a single bin, we might be interested in the total weight of the balls hashing
    below some threshold $l$. This is useful for sampling, for if $h(x)$ is uniform in $[m]$, then $\Pr[h(x)<l]=l/m$.
    We then define the value function $v \colon U \times [m] \to \R$ by $v(x, j) = w_x \indicator{j < l} \indicator{x \in S}$,
    then $X = \sum_{x \in U} v(x, h(x)) = \sum_{x \in S} w_x \indicator{h(x)< l}$ will be precisely the total weight of the balls
    hashing below $l$.
\end{enumerate}
The first case appears when one tries to allocate resources, and the second case arises in streaming algorithms, see, e.g.,~\cite{aamand2020repititions}. In any case, $X$ ought to be concentrated around the mean $\mu = \ep{X}$.
If $h$ is a uniformly random hash function then this will be the case under mild assumptions about $v$ but it cannot otherwise be assumed a priori to be the case.

There are two natural ways to quantify the concentration of $X$, either we bound the tail of $X$, i.e., we bound $\prb{\abs{X - \mu} \ge t}$ for all $t \ge 0$, or we bound the central moments of $X$, i.e., we bound the $p$-th moment $\ep{\abs{X - \mu}^p}$ for all $p \ge 2$.
If we have a bound on the tail that is exponentially decreasing, we can bound the central moments of $X$ for all $p \ge 2$.
Unfortunately, some of the prior works~\cite{aamand2020,dietzfelbinger2009,thorup2013} prove bounds on the tail that are exponentially decreasing but also has an additive term of the form $n^{-\gamma}$ where $\gamma = O(1)$.
It will then only be possible to give strong bounds for the central moments of $X$ for $p = O(1)$.
This is not necessarily a fault of the hash function but a defect of the analysis.
In contrast, if we prove strong bounds for the central moments of $X$ for $p = O(\log n)$ then we can use Markov's inequality to prove a bound the tail that is exponentially decreasing but with an additive term of the form $n^{-\gamma}$ where $\gamma = O(1)$.
Thus in some sense, it is more robust to bound the moments compared to bounding the tail.

We can use the classic $k$-independent hashing framework of Wegman and Carter~\cite{wegman1981} as an easy way to obtain a hash function that has bounds on the central moments as a uniformly random hash function.
A random hash function, $h \colon U \to R$, is $k$-independent if $(h(x_0), \ldots, h(x_{k - 1}))$ is uniformly distributed in $R^k$ for any $k$ distinct keys $x_0, \ldots, x_{k - 1} \in U$.
The $p$-th central moment $\ep{\left(X - \mu\right)^p}$ of $X$ for a $k$-independent hash function $h$ is the same as the $p$-th central moment of $X$ for a fully random hash function when $p$ is an even integer less than $k$.

\subsection{Tabulation Hashing}

Simple tabulation hashing dates back to 1970 and was first introduced by Zobrist for optimizing chess computers~\cite{zobrist70}.
In simple tabulation hashing, we view the universe, $U$, to be of the form $U = \Sigma^c$ for some alphabet, $\Sigma$, and a positive integer $c$.
Let $T \colon \set{0, \ldots, c - 1} \times \Sigma \to [2^l]$ be a uniformly random table, i.e., each value is chosen independently and uniformly at random from the set $[2^l]$.
A simple tabulation hash function, $h \colon \Sigma^c \to [2^l]$, is then defined by
\begin{align*}
    h(\alpha_0, \ldots, \alpha_{c - 1})
        = \bigxor_{i = 0}^{c - 1} T(i, \alpha_i)
    \; ,
\end{align*}
where $\xor$ is the bitwise XOR-operation, i.e., addition when $[2^l]$ is identified with the Abelian group $(\Z / 2\Z)^l$.
We say that $h$ is a simple tabulation hash function with $c$ characters.
With 8- or 16-bit characters, the random table $T$ fits in cache, and then simple tabulation is
very fast, e.g., in experiments, \cite{Patrascu2012} found it to be as fast as two to three multiplications.

The moments of simple tabulation hashing have been studied in multiple
papers. Braverman et al.~\cite{BravermanCLMO10} showed that for a
fixed bin the 4th central moment is close to that achieved by truly
random hashing. Dahlgaard et al.~\cite{Dahlgaard2017} generalized this
to any constant moment $p$. Their proof works for any $p$ but with a
doubly exponential dependence on $p$, so their bound is only useful
for $p = O(1)$.  In this paper, we obtain bounds for all the moments
of hash-based sums for simple tabulation hashing which are tight up to
constants depending only on $c$.

Previous work has just treated $c$ as a constant, hidden in $O$-notation. However, $c$ does provide a fundamental trade-off between evaluation
time with $c$ lookups and the space $cU^{1/c}$. We therefore find it relevant to
elucidate how our moment bounds depend on $c$ even though we typically
choose $c=4$.

Mixed tabulation hashing was introduced by Dahlgaard et al.~\cite{Dahlgaard2015}.
As in simple tabulation hashing, we view the universe, $U$, to be of the form $U = \Sigma^c$ for some alphabet, $\Sigma$, and a positive integer $c$.
We further assume that the alphabet, $\Sigma$, has the form $\Sigma = [2^k]$.
Let $h_1 \colon \Sigma^c \to [2^l]$, $h_2 \colon \Sigma^c \to \Sigma^d$, and $h_3 \colon \Sigma^d \to [2^l]$ be independent simple tabulation hash functions.
A mixed tabulation hash function, $h \colon \Sigma^c \to [2^l]$, is then defined by
\begin{align*}
    h(x) = h_1(x) \xor h_3(h_2(x))
    \; .
\end{align*}
As in simple tabulation hashing, $\xor$ is the bitwise XOR-operation.
We call $h$ a mixed tabulation hash function with $c$ characters and $d$ derived characters.
We note that $h_1$ and $h_2$ can be combined in a single
simple tabulation hash function $\Sigma^c \to [2^l] \times \Sigma^d$, and
then $h$ is implemented with only $c+d$ lookups.

 With simple tabulation hashing, there are certain inputs for which
    the concentration is much worse than with fully random hashing. However, with mixed tabulation, even if
    we have just $d=1$ derived character, we get logarithmic moment bounds
    that, for $c = O(1)$, are only a constant factor
    worse than those with fully-random hashing for any input assuming that hash range at most polynomial in
    the key universe.

Getting within a constant factor is very convenient within algorithm
    analysis, where we typically only aim for $O$-bounds that are tight
    within a constant factor.

\subsection{Relation between Simple Tabulation and Chaoses}\label{sec:relation-chaos-simple}
A \emph{chaos} of order $c$ is a random variable of the form
\[
    \sum_{0 \le i_0 < \ldots < i_{c - 1} < n} a_{i_0, \ldots, i_{c - 1}} \prod_{j \in [c]} X_{i_j} \; ,
\]
where $(X_i)_{i \in [n]}$ are independent random variables and $(a_{i_0, \ldots, i_{c -1}})_{0 \le i_0 < \ldots < i_{c - 1} < n}$ is a multiindexed array of real numbers.
And a \emph{decoupled chaos} of order $c$ is a random variable of the form
\[
    \sum_{i_0, \ldots, i_{c - 1} \in [n]} a_{i_0, \ldots, i_{c - 1}} \prod_{j \in [c]} X_{i_j}^{(j)} \; ,
\]
where $(X_i^{(j)})_{i \in [n], j \in [c]}$ are independent random variables and $(a_{i_0, \ldots, i_{c -1}})_{i_0, \ldots, i_{c - 1} \in [n]}$ is a multiindexed array of real numbers.
Chaoses have been studied in different settings, e.g., when the variables are standard Gaussian variables~\cite{latala2006,lehec2011}, when the variables have logarithmically concave tails~\cite{Adamczak2012}, and when the variables have logarithmically convex tails~\cite{kolesko2015}.

From the definition of a chaos and simple tabulation hashing it might not be immediately clear that there is connection between the two.
But we can rewrite the expression for hash-based sums of simple tabulation hashing as follows
\begin{align*}
    \sum_{x \in \Sigma^c} v(x, h(x))
        &= \sum_{\alpha_0, \ldots, \alpha_{c - 1} \in \Sigma}
            v((\alpha_0, \ldots, \alpha_{c - 1}), h(\alpha_0, \ldots, \alpha_{c - 1}))
        \\&= \sum_{j_0, \ldots, j_{c - 1} \in [m]} \sum_{\alpha_0, \ldots, \alpha_{c - 1} \in \Sigma}
            v\left((\alpha_0, \ldots, \alpha_{c - 1}), \bigxor_{i \in [c]} j_i \right) \prod_{i \in [c]} \indicator{T(i, \alpha_i) = j_i}
    \; .
\end{align*}
We then notice that $\sum_{\alpha_0, \ldots, \alpha_{c - 1} \in \Sigma} v\left((\alpha_0, \ldots, \alpha_{c - 1}), \bigxor_{i \in [c]} j_i \right) \prod_{i \in [c]} \indicator{T(i, \alpha_i) = j_i}$ is a decoupled chaos of order $c$ for any $(j_i)_{i \in [c]}$,
thus hash-based sums of simple tabulation hashing can be seen as a sum of chaoses.
Now since the random variables, $(\indicator{T(i, \alpha_i) = j})_{j \in [m]}$, are not independent then the chaoses are not independent either which complicates the analysis.
Nonetheless, this realization inspires us to use techniques from the study of chaoses to analyze the moments of tabulation hashing,
in particular, our approach will be analytical in contrast with the combinatorial approach of the previous papers. 
We will expand further on the techniques in \Cref{sec:techniques}.

%% file: introduction/results.tex

\subsection{Our Results}

When proving and stating bounds for the $p$-th moment of a random variable it is often more convenient and more instructive to do it in terms of the $p$-norm of the random variable.
The $p$-norm of a random variable is the $p$-th root of the $p$-th moment of the random variable and is formally defined as follows:
\begin{definition}[$p$-norm]
    Let $p \ge 1$ and $X$ be a random variable with $\ep{\abs{X}^p} < \infty$.
    We then define the $p$-norm of $X$ by $\pnorm{X}{p} = \ep{\abs{X}^p}^{1/p}$.
\end{definition}

Our main contributions of this paper are analyses of the moments of hash-based sums of simple tabulation hashing and mixed tabulation hashing.
To do this we first had to analyze the moments of hash-based sums of fully random hashing which as far as we are aware have not been analyzed tightly before.

\subsubsection{The Moments of Fully Random Hashing}

Previously, the focus has been on proving Chernoff-like bounds by using the moment generating function but a natural, different approach would be to use moments instead. 
Both the Chernoff bounds~\cite{chernoff1952} and the more general Bennett's inequality~\cite{Bennett1962} bound the tail using the Poisson distribution.
More precisely, let $v \colon U \times [m] \to \R$ be a value function that satisfies that $\sum_{j \in [m]} v(x, j) = 0$ and define the following two parameters $M_v$ and $\sigma_v^2$ which will be important throughout the paper as follows: 
\begin{align}
    M_v &= \max_{x \in U, j \in [m]} \abs{v(x, j)} \; , \\
    \sigma_v^2 &= \frac{\sum_{x \in U, j \in [m]} v(x, j)^2}{m} \; .
\end{align}
Bennett's inequality specialized to our setting then says that for a fully random hash function $h$
\begin{align}\begin{split}\label{eq:bennett}
    \prb{\abs{\sum_{x \in U} v(x, h(x))} \ge t}
        &\le 2\exp\!\left(-\tfrac{\sigma_v^2}{M_v^2} \mathcal{C}\!\left(\tfrac{t M_v}{\sigma_v^2} \right) \right)
        \\&\le \begin{cases}
            2\exp\!\left(-\tfrac{t^2}{3\sigma_v^2} \right) &\text{if $t \le \frac{\sigma_v^2}{M_v}$} \\
            2\exp\!\left(-\tfrac{t}{2 M_v} \log\!\left(1 + \tfrac{t M_v}{\sigma_v^2} \right) \right) &\text{if $t > \frac{\sigma_v^2}{M_v}$}
        \end{cases}
    \; ,
\end{split}\end{align}
where $\mathcal{C}(x) = (x + 1)\log(x + 1) - x$.\footnote{Here and throughout the paper $\log(x)$ will refer to the natural logarithm.}


This inspires us to try to bound the $p$-norms of $X_v$ with the $p$-norms of the Poisson distribution.
To do this we will introduce the function $\Psi_p(M, \sigma^2)$ which is quite technical but we will prove that $\Psi_p(1, \lambda)$ is equal up to a constant factor to the central $p$-norm of a Poisson distributed variable with mean $\lambda$.
One should think of $\Psi_p(M, \sigma^2)$ as a $p$-norm version of $\tfrac{\sigma_v^2}{M_v^2} \mathcal{C}\!\left(\tfrac{t M_v}{\sigma_v^2} \right)$ which appears in Bennett's inequality.
\begin{definition}
    For $p \ge 2$ we define the function $\Psi_p \colon \R_+ \times \R_{+} \to \R_{+}$ as follows
    \begin{align*}
        \Psi_p(M, \sigma^2) = \begin{cases}
            \left(\frac{\sigma^2}{p M^2}\right)^{1/p} M &\text{if $p < \log \frac{p M^2}{\sigma^2}$} \\
            \tfrac{1}{2}\sqrt{p}\sigma &\text{if $p < e^{2} \frac{\sigma^2}{M^2}$} \\
            \frac{p}{e \log \frac{p M^2}{\sigma^2}} M &\text{if $\max\!\set{\log \frac{p M^2}{\sigma^2}, e^{2} \frac{\sigma^2}{M^2}} \le p$}
        \end{cases}
        \; .
    \end{align*}
\end{definition}
\begin{remark}

    When $p$ is \emph{small} then case 1 and 2 apply while for \emph{large} $p$ case 3 applies.
    If $2 < e^2\frac{\sigma^2}{M^2}$ then we always have that $p > \log \frac{p M^2}{\sigma^2}$ for $2 \le p$, hence only case 2 and 3 apply.
    Similarly, if $e^2\frac{\sigma^2}{M^2} \le 2$ then $p \ge e^2\frac{\sigma^2}{M^2}$ for all $2 \le p$, hence only case 1 and 3 apply. 
    This shows that the cases disjoint and cover all parameter configurations.
\end{remark}

The definition $\Psi_p(M, \sigma^2)$ might appear strange but it does in fact capture the central $p$-norms of Poisson distributed random variables.
This is stated more formally in the following lemma.
\restateable[lemma]{psi-relation-poisson}{
    There exist universal constants $K_1$ and $K_2$ satisfying that for a Poisson distributed random variable, $X$, with $\lambda = \ep{X}$
    \[
        K_2 \Psi_p(1, \lambda) \le \pnorm{X - \lambda}{p} \le K_1 \Psi_p(1, \lambda)
        \; ,
    \]
    for all $p \ge 2$.
}{
    There exist universal constants $K_1$ and $K_2$ satisfying that for a Poisson distributed random variable, $X$, with $\lambda = \ep{X}$
    \[
        K_2 \Psi_p(1, \lambda) \le \pnorm{X - \lambda}{p} \le K_1 \Psi_p(1, \lambda)
        \; ,
    \]
    for all $p \ge 2$.
}

Bennett's inequality shows that we can bound the tail of $\sum_{x \in U} v(x, h(x))$ and \Cref{lem:psi-relation-poisson} shows that $\Psi_p(M, \sigma^2)$ captures the central $p$-norms of the Poisson distribution.
It is therefore not so surprising that we are to bound the $p$-norms of $\sum_{x \in U} v(x, h(x))$ using $\Psi_p(M, \sigma^2)$.
\restateable[theorem]{sampling-moments}{
    Let $h \colon U \to [m]$ be a uniformly random function, let $v \colon U \times [m] \to \R$ be a fixed value function, and
    assume that $\sum_{j \in [m]} v(x, j) = 0$ for all keys $x \in U$.
    Define the random variable $X_v = \sum_{x \in U} v(x, h(x))$.
    Then for all $p \ge 2$
    \begin{align*}
        \pnorm{X_v}{p}
            \le L \Psi_p\left(M_v, \sigma_v^2 \right)
        \; ,
    \end{align*}
    where $L \le 16 e$ is a universal constant.
}{
    Let $h \colon U \to [m]$ be a uniformly random function, let $v \colon U \times [m] \to \R$ be a fixed value function, and
    assume that $\sum_{j \in [m]} v(x, j) = 0$ for all keys $x \in U$.
    Define the random variable $X_v = \sum_{x \in U} v(x, h(x))$.
    Then for all $p \ge 2$
    \begin{align*}
        \pnorm{X_v}{p}
            \le L \Psi_p\left(M_v, \sigma_v^2 \right)
        \; ,
    \end{align*}
    where $L \le 16 e$ is a universal constant.
}

To get a further intuition for $\Psi_p(M, \sigma^2)$ is is instructive to apply Markov's inequality and compare the tail bound to Bennett's inequality.
More precisely, assume that $\pnorm{Y - \ep{Y}}{p} \le L \Psi_p(M, \sigma^2)$ for a constant $L$ and for all $p \ge 2$.
Then we can use Markov's inequality to get the following tail bound for all $t > 0$
\begin{align}\begin{split}\label{eq:psi-tail}
    \prb{\abs{Y - \ep{Y}} \ge t}
        &\le \left( \frac{\pnorm{Y - \ep{Y}}{p}}{t} \right)^p 
        \\&\le \begin{cases}
            \frac{L^2 \sigma^2}{2 t^2} &\text{if $t \le L \max\!\set{M, \tfrac{e \sigma}{\sqrt{2}}}$} \\
            \exp\!\left(- \tfrac{4 t^2}{e^2 L^2 \sigma^2} \right) &\text{if $L \tfrac{e \sigma}{\sqrt{2}} \le t \le L \tfrac{e^2 \sigma^2}{2 M}$} \\
            \exp\!\left(- \tfrac{t}{L M} \log\!\left( \tfrac{2 t M}{L \sigma^2} \right) \right) &\text{if $L \max\!\set{\tfrac{e^2 \sigma^2}{2 M}, M} \le t$} \\
        \end{cases}
    \; .
\end{split}\end{align}
In order to obtain these bounds $p$ is chosen as follows: If $t \le \max\!\set{M, \tfrac{e \sigma}{\sqrt{2}}}$ then $p = 2$ and otherwise $p$ is chosen such that $\pnorm{Y - \ep{Y}}{p} \le e^{-1} t$.
More precisely, we have that
\begin{align*}
    p = \begin{cases}
        2 &\text{if $t \le L \max\!\set{M, \tfrac{e \sigma}{\sqrt{2}}}$} \\
        \tfrac{4 t^2}{e^2 L^2 \sigma^2} &\text{if $L \tfrac{e \sigma}{\sqrt{2}} \le t \le L \tfrac{e^2 \sigma^2}{2 M}$} \\
        \tfrac{t}{L M}\log\left(\tfrac{2 t M}{L \sigma^2} \right) &\text{if $L \max\!\set{\tfrac{e^2 \sigma^2}{2 M}, M} \le t$} \\
    \end{cases}
    \; .
\end{align*}
We see that \cref{eq:psi-tail} gives the same tail bound as Bennett's inequality, \cref{eq:bennett}, up to a constant in the exponent.

We also prove a matching lower bound to \Cref{thm:sampling-moments} which shows that $\Psi_p(M, \sigma^2)$ is the correct function to consider.
\restateable[theorem]{sampling-moments-lower-bound}{
    Let $h \colon U \to [m]$ be a uniformly random function, then there exists a value function, $v \colon U \times [m] \to \R$, where $\sum_{j \in [m]} v(x, j) = 0$ for all keys $x \in U$, such that the random variable $X_v = \sum_{x \in U} v(x, h(x))$ satisfies that for all $p \le L_1 \abs{U} \log(m)$
    \begin{align*}
        \pnorm{\sum_{x \in U} v(x, h(x)) }{p}
            \ge L_2 \Psi_p\left(M_v, \sigma_v^2 \right)
        \; ,
    \end{align*}
    where $L_1$ and $L_2$ are a universal constant.
}{
    Let $h \colon U \to [m]$ be a uniformly random function, then there exists a value function, $v \colon U \times [m] \to \R$, where $\sum_{j \in [m]} v(x, j) = 0$ for all keys $x \in U$, such that the random variable $X_v = \sum_{x \in U} v(x, h(x))$ satisfies that for all $p \le L_1 \abs{U} \log(m)$
    \begin{align*}
        \pnorm{\sum_{x \in U} v(x, h(x)) }{p}
            \ge L_2 \Psi_p\left(M_v, \sigma_v^2 \right)
        \; ,
    \end{align*}
    where $L_1$ and $L_2$ are a universal constant.
}

\subsubsection{The Moments of Tabulation Hashing}
We analyze the $p$-norms of hash-based sums for simple tabulation hashing, and our analysis is the first that provides useful bounds for non-constant moments.
Furthermore, it is also the first analysis of simple tabulation hashing that does not assume that $c$ is constant.
We obtain an essentially tight understanding of this problem and show that simple tabulation hashing only works well when the range is large.
This was also noted by Aamand et al.~\cite{aamand2020} and they solve this deficiency of simple tabulation hashing by introducing a new hashing scheme, tabulation-permutation hashing.
We show that it is also possible to break the bad instances of simple tabulation hashing by using mixed tabulation hashing.

We introduce a bit of notation to make the theorems cleaner.
We will view a value function $v \colon \Sigma^c \times [m] \to \R$ as a vector, more precisely, we let
\[
    \norm{v}{q} = \left(\sum_{x \in \Sigma^c} \sum_{j \in [m]} \abs{v(x, j)}^q \right)^{1/q}
\]
for all $q \in [1, \infty]$.
For every key $x \in \Sigma^c$ we define $v[x]$ to be the sub-vector $v$ restricted to $x$, more precisely, we let
\[
    \norm{v[x]}{q} = \left(\sum_{j \in [m]} \abs{v(x, j)}^q \right)^{1/q}
\]
for all $q \in [1, \infty]$.

\paragraph{Simple Tabulation Hashing.}

Our main result for simple tabulation hashing is a version of \Cref{thm:sampling-moments}.

\restateable[theorem]{simple-tab-moments}{
    Let $h \colon \Sigma^c \to [m]$ be a simple tabulation hash function, $v \colon \Sigma^{c} \times [m] \to \R$ a value function, and assume that $\sum_{j \in [m]} v(x, j) = 0$ for all keys $x \in \Sigma^c$.
    Define the random variable $V^{\text{\normalfont simple}}_v = \sum_{x \in \Sigma^c} v(x, h(x))$.
    Then for all $p \ge 2$
    \begin{align*} 
        \pnorm{V^{\text{\normalfont simple}}_v}{p}
            \le L_1 \Psi_p\left(K_c \gamma_p^{c - 1} M_v, K_c \gamma_p^{c - 1} \sigma_v^2  \right)
        \; ,
    \end{align*}
    where $K_c = \left(L_2 c \right)^{c - 1}$, $L_1$ and $L_2$ are universal constants, and 
    \begin{align*}
        \gamma_p = \frac{\max\!\set{\log(m) + \log\!\left(\tfrac{\sum_{x \in \Sigma^c} \norm{v[x]}{2}^2}{\max_{x \in \Sigma^c} \norm{v[x]}{2}^2}\right)/c, p}}
            {\log\!\left( e^2 m \left( \max_{x \in \Sigma^c} \frac{\norm{v[x]}{1}^2}{\norm{v[x]}{2}^2} \right)^{-1} \right)}
    \end{align*}
}{
    Let $h \colon \Sigma^c \to [m]$ be a simple tabulation hash function, $v \colon \Sigma^{c} \times [m] \to \R$ a value function, and assume that $\sum_{j \in [m]} v(x, j) = 0$ for all keys $x \in \Sigma^c$.
    Define the random variable $V^{\text{\normalfont simple}}_v = \sum_{x \in \Sigma^c} v(x, h(x))$.
    Then for all $p \ge 2$
    \begin{align*} 
        \pnorm{V^{\text{\normalfont simple}}_v}{p}
            \le L_1 \Psi_p\left(K_c \gamma_p^{c - 1} M_v, K_c \gamma_p^{c - 1} \sigma_v^2  \right)
        \; ,
    \end{align*}
    where $K_c = \left(L_2 c \right)^{c - 1}$, $L_1$ and $L_2$ are universal constants, and 
    \begin{align*}
        \gamma_p = \frac{\max\!\set{\log(m) + \log\!\left(\tfrac{\sum_{x \in \Sigma^c} \norm{v[x]}{2}^2}{\max_{x \in \Sigma^c} \norm{v[x]}{2}^2}\right)/c, p}}
            {\log\!\left( e^2 m \left( \max_{x \in \Sigma^c} \frac{\norm{v[x]}{1}^2}{\norm{v[x]}{2}^2} \right)^{-1} \right)}
    \end{align*}
}

\newcommand{\theoremSimpleMoments}[1][]{
    \begin{theorem}\label{#1}
        Let $h \colon \Sigma^c \to [m]$ be a simple tabulation hash function, $v \colon \Sigma^{c} \times [m] \to \R$ a value function, and assume that $\sum_{j \in [m]} v(x, j) = 0$ for all keys $x \in \Sigma^c$.
        Define the random variable $V^{\text{\normalfont simple}}_v = \sum_{x \in \Sigma^c} v(x, h(x))$.
        Then for all $p \ge 2$
        \begin{align*} 
            \pnorm{V^{\text{\normalfont simple}}_v}{p}
                \le \Psi_p\left(K_c \gamma_p^{c - 1} M_v, K_c \gamma_p^{c - 1} \sigma_v^2  \right)
            \; ,
        \end{align*}
        where $K_c = L_1 \left(L_2 c \right)^{c - 1}$, $L_1$ and $L_2$ are universal constants, and 
        \begin{align*}
            \gamma_p = \frac{\max\!\set{\log(m) + \log\!\left(\tfrac{\sum_{x \in \Sigma^c}\norm{v[x]}{2}^2}{\max_{x \in \Sigma^c} \norm{v[x]}{2}^2}\right)/c, p}}{\log\!\left(\min_{x \in \Sigma^c} \frac{e^2 m \norm{v[x]}{2}^2}{\norm{v[x]}{1}^2}\right)}
        \end{align*}
    \end{theorem}
}


It is instructive to compare this result to \Cref{thm:sampling-moments} for fully random hashing.
Ignoring the constant $K_c$, the result for simple tabulation hashing corresponds to the result for fully random hashing if we group keys into groups of size $\gamma_p^{c - 1}$.

The definition of $\gamma_p$ is somewhat complicated because of the generality of the theorem, but we will try to explain the intuition behind it.
The expression $\max_{x \in \Sigma^c} \frac{\norm{v[x]}{1}^2}{\norm{v[x]}{2}^2}$ measures how spread out the mass of the value function is.
It was also noted in the previous analysis by Aamand et al.~\cite{aamand2020} that this measure is naturally occurring.
In fact, their result needs that $\max_{x \in \Sigma^c} \frac{\norm{v[x]}{1}^2}{\norm{v[x]}{2}^2} \le m^{1/4}$.
If we consider the example from the introduction of hashing below a threshold $l \le m$ where each key, $x \in \Sigma^c$, has weight $w_x$, then the value function, $v$, will be $v(x, j) = w_x\left(\indicator{j < l} - \tfrac{l}{m}\right)$ for $x \in \Sigma^c, j \in [m]$, and we then get that
\begin{align*}
    \max_{x \in \Sigma^c} \frac{\norm{v[x]}{1}^2}{\norm{v[x]}{2}^2}
        = 4 l\left(1 - \frac{l}{m}\right)
        \le 4 l
        \; .
\end{align*}
This correctly measures that the mass of the value function is mostly concentrated to the $l$ positions of $[m]$.

The expression $\frac{\sum_{x \in \Sigma^c} \norm{v[x]}{2}^2}{\max_{x \in \Sigma^c} \norm{v[x]}{2}^2}$ is a measure for how many keys that have significant weight.
This also showed up in the previous analyses of simple tabulation hashing~\cite{aamand2020,Patrascu2012}.
If we again consider the example from before, we get that
\begin{align*}
    \frac{\sum_{x \in \Sigma^c} \norm{v[x]}{2}^2}{\max_{x \in \Sigma^c} \norm{v[x]}{2}^2}
        = \frac{\sum_{x \in \Sigma^c} w_x^2}{\max_{x \in \Sigma^c} w_x^2}
    \; .
\end{align*}

We can summarize the example in the following corollary.

\begin{corollary} 
    Let $h \colon \Sigma^c \to [m]$ be a simple tabulation hash function, assign a weight, $w_x \in \R$, to every key, $x \in \Sigma^c$, and consider a threshold $l \le m$.
    Define the random variable $V^{\text{\normalfont simple}}_v = \sum_{x \in \Sigma^c} w_x\left(\indicator{h(x) < l} - \tfrac{l}{m} \right)$.
    Then for all $p \ge 2$
    \begin{align*}
        \pnorm{V^{\text{\normalfont simple}}_v}{p}
            \le \Psi_p\left(K_c \gamma_p^{c - 1} \max_{x \in \Sigma^c} \abs{w_x}, K_c \gamma_p^{c - 1}\left(\sum_{x \in \Sigma^c} w_x^2\right)\frac{l}{m}\left(1 - \frac{l}{m}\right) \right)
        \; ,
    \end{align*}
    where $K_c = L_1 \left(L_2 c \right)^{c - 1}$, $L_1$ and $L_2$ are universal constants, and 
    \begin{align*}
        \gamma_p = \frac{\max\!\set{\log(m) + \log\!\left(\tfrac{\sum_{x \in \Sigma^c} w_x^2}{\max_{x \in \Sigma^c} w_x^2}\right)/c, p}}{\log\!\left(\tfrac{e^2 m}{4l}\right)}
    \end{align*}
\end{corollary}

A natural question is how close \Cref{thm:simple-tab-moments} is to being tight.
We show that if $\log(m) + \log\!\left(\tfrac{\sum_{x \in \Sigma^c} \norm{v[x]}{2}^2}{\max_{x \in \Sigma^c} \norm{v[x]}{2}^2}\right)/c = O\left( \log\!\left(1 + m\left(\max_{x \in \Sigma^c} \frac{\norm{v[x]}{1}^2}{\norm{v[x]}{2}^2}\right)^{-1} \right) \right)$ then the result is tight up to a universal constant depending only $c$.
Formally, we prove the following lemma.

\restateable[theorem]{simple-tab-moments-lower-bound}{
    Let $h \colon \Sigma^c \to [m]$ be a simple tabulation hash function, and $2 \le p \le L_1 \abs{\Sigma} \log(m)$, then there exists a value function, $v \colon U \times [m] \to \R$, where $\sum_{j \in [m]} v(x, j) = 0$ for all keys $x \in \Sigma^c$, and for which
    \begin{align*} 
        \pnorm{\sum_{x \in \Sigma^c} v(x, h(x))}{p}
            \ge K'_c \Psi_p\left(\gamma_p^{c - 1} M_v,  \gamma_p^{c - 1} \sigma_v^2  \right)
        \; ,
    \end{align*}
    where $K'_c = L_1^{c}$ and $L_1$ is a universal constant, and 
    \begin{align*}
        \gamma_p = \max\!\set{1, \frac{p}{\log\!\left(e^2 m \left(\max_{x \in \Sigma^c} \tfrac{\norm{v[x]}{1}^2}{\norm{v[x]}{2}^2} \right)^{-1} \right)}}
    \end{align*}
}{
    Let $h \colon \Sigma^c \to [m]$ be a simple tabulation hash function, and $2 \le p \le L_1 \abs{\Sigma} \log(m)$, then there exists a value function, $v \colon U \times [m] \to \R$, where $\sum_{j \in [m]} v(x, j) = 0$ for all keys $x \in \Sigma^c$, and for which
    \begin{align*} 
        \pnorm{\sum_{x \in \Sigma^c} v(x, h(x))}{p}
            \ge K'_c \Psi_p\left(\gamma_p^{c - 1} M_v,  \gamma_p^{c - 1} \sigma_v^2  \right)
        \; ,
    \end{align*}
    where $K'_c = L_1^{c}$ and $L_1$ is a universal constant, and 
    \begin{align*}
        \gamma_p = \max\!\set{1, \frac{p}{\log\!\left(e^2 m \left(\max_{x \in \Sigma^c} \tfrac{\norm{v[x]}{1}^2}{\norm{v[x]}{2}^2} \right)^{-1} \right)}}
    \end{align*}
}

\newcommand{\theoremSimpleMomentsLower}[1][]{
    \begin{theorem}\label{#1}
        Let $h \colon \Sigma^c \to [m]$ be a simple tabulation hash function, and $2 \le p \le L_1 \abs{\Sigma} \log(m)$, then there exists a value function, $v \colon U \times [m] \to \R$, where $\sum_{j \in [m]} v(x, j) = 0$ for all keys $x \in \Sigma^c$, and for which
        \begin{align} 
            \pnorm{\sum_{x \in \Sigma^c} v(x, h(x))}{p}
                \ge K'_c \Psi_p\left(\gamma_p^{c - 1} M_v,  \gamma_p^{c - 1} \sigma_v^2  \right)
            \; ,
        \end{align}
        where $K'_c = L_1^{c}$ and $L_1$ is a universal constant, and 
        \begin{align*}
            \gamma_p = \max\!\set{1, \frac{p}{\log\!\left(\min_{x \in \Sigma^c} \tfrac{e m \norm{v[x]}{2}^2}{\norm{v[x]}{1}^2}\right)}}
        \end{align*}
    \end{theorem}
}



\paragraph{Mixed Tabulation Hashing.}
The results of simple tabulation hashing work well when the range is large and when the mass of the value function is on few coordinates.
We show that mixed tabulation hashing works well even if the range is small.

\restateable[theorem]{mixed-tab-moments}{
    Let $h \colon \Sigma^c \to [m]$ be a mixed tabulation function with $d \ge 1$ derived characters, $v \colon \Sigma^{c} \times [m] \to \R$ a value function, and assume that $\sum_{j \in [m]} v(x, j) = 0$ for all keys $x \in \Sigma^c$.
    Define the random variable $V^{\text{\normalfont mixed}}_v = \sum_{x \in \Sigma^c} v(x, h(x))$.
    For all $p \ge 2$ then
    \begin{align}
        \pnorm{V^{\text{\normalfont mixed}}_v}{p}
            \le \Psi_p\left(K_c \gamma_p^{c} M_v,  K_c \gamma_p^{c} \sigma_v^2 \right)
    \end{align}
    where $K_c = L_1 \left(L_2 c \right)^{c}$, $L_1$ and $L_2$ are universal constants, and 
    \begin{align*}
        \gamma_p = \max\!\set{1,  \frac{\log(m)}{\log(\abs{\Sigma})}, \frac{p}{\log(\abs{\Sigma} )} }
        \; .
    \end{align*}
}{
    Let $h \colon \Sigma^c \to [m]$ be a mixed tabulation function with $d \ge 1$ derived characters, $v \colon \Sigma^{c} \times [m] \to \R$ a value function, and assume that $\sum_{j \in [m]} v(x, j) = 0$ for all keys $x \in \Sigma^c$.
    Define the random variable $V^{\text{\normalfont mixed}}_v = \sum_{x \in \Sigma^c} v(x, h(x))$.
    For all $p \ge 2$ then
    \begin{align}\label{eq:mixed-tab-moments}
        \pnorm{V^{\text{\normalfont mixed}}_v}{p}
            \le \Psi_p\left(K_c \gamma_p^{c} M_v,  K_c \gamma_p^{c} \sigma_v^2 \right)
    \end{align}
    where $K_c = L_1 \left(L_2 c \right)^{c}$, $L_1$ and $L_2$ are universal constants, and 
    \begin{align*}
        \gamma_p = \max\!\set{1,  \frac{\log(m)}{\log(\abs{\Sigma})}, \frac{p}{\log(\abs{\Sigma} )} }
        \; .
    \end{align*}
}

Usually, in hashing contexts, we do not map to a much larger domain, i.e., we will usually have that $m \le \abs{U}^{\gamma}$ for some constant $\gamma \ge 1$.
If this is the case then we can obtain the following nice tail bound for mixed tabulation hashing by using Markov's inequality.
\begin{corollary}\label{cor:mixed-tail}
    Let $h \colon \Sigma^c \to [m]$ be a mixed tabulation function with $d \ge 1$ derived characters, $v \colon \Sigma^{c} \times [m] \to \R$ a value function, and assume that $\sum_{j \in [m]} v(x, j) = 0$ for all keys $x \in \Sigma^c$.
    Define the random variable $V^{\text{\normalfont mixed}}_v = \sum_{x \in \Sigma^c} v(x, h(x))$.
    If $m \le \abs{U}^{\gamma}$ for a value $\gamma \ge 1$ then for all $t \ge 0$
    \begin{align*}
        \prb{\abs{V^{\text{\normalfont mixed}}_v} \ge t}
            \le \exp\!\left(-\tfrac{\sigma_v^2}{M_v^2} \mathcal{C}\!\left(\tfrac{t M_v}{ \sigma_v^2} \right) / K_{c, \gamma} \right) + \abs{U}^{-\gamma}
        \; ,
    \end{align*}
    where $\mathcal{C}(x) = (x + 1)\log(x + 1) - x$, $K_{c, \gamma} = L_1 \left(L_2 c^2 \gamma \right)^{c}$, and $L_1$ and $L_2$ are universal constants.
\end{corollary}
\begin{proof}
    The idea is to combine \Cref{thm:mixed-tab-moments} and Markov's inequality.
    We use \Cref{thm:mixed-tab-moments} for $2 \le p \le \gamma \log \abs{U}$ to get that
    \begin{align*}
        \pnorm{V^{\text{\normalfont mixed}}_v}{p} \le \Psi_p\left(K_c \gamma_p^{c} M_v,  K_c \gamma_p^{c} \sigma_v^2 \right)
        \; ,
    \end{align*}
    where we can bound $\gamma_p$ by
    \begin{align*}
        \gamma_p = \max\!\set{1, \frac{\log(m)}{\log(\abs{\Sigma})}, \frac{p}{\log(\abs{\Sigma} )} }
            \le c \gamma
        \; .
    \end{align*}
    So we have that
    \begin{align*}
        \pnorm{V^{\text{\normalfont mixed}}_v}{p} \le \Psi_p\left(\left(L_2 c^2 \gamma \right)^{c} M_v, \left(L_2 c^2 \gamma \right)^{c} \sigma_v^2 \right)
        \; .
    \end{align*}
    Now by the same method as in \cref{eq:psi-tail}, we get the result.
\end{proof}

\paragraph{Adding a query element}
In many cases, we would like to prove that these properties continue to hold even when conditioning on a query element.
An example would be the case where we are interested in the weight of the elements in the bin for which the query element, $q$, hashes to,
i.e., we would like that $\sum_{x \in S} w_x \indicator{h(x) = h(q)}$ is concentrated when conditioning on $q$.
Formally, this corresponds to having the value function $v \colon \Sigma^c \times [m] \times [m]$ defined by $v(x, j, k) = w_x\indicator{x \in S}\indicator{j = k}$ and then proving concentration on 
$\sum_{x \in \Sigma^c \setminus \set{q}} v(x, h(x), h(q))$ when conditioning on $q$.
We show that this holds both for simple tabulation and mixed tabulation.

\restateable[theorem]{simple-tab-query}{
    Let $h \colon \Sigma^c \to [m]$ be a simple tabulation hash function and let $q \in \Sigma^c$ be a designated query element.
    Let $v \colon \Sigma^{c} \times [m] \times [m] \to \R$ a value function, and assume that $\sum_{j \in [m]} v(x, j, k) = 0$ for all keys $x \in U$ and all $k \in [m]$.
    Define the random variable $V^{\text{\normalfont simple}}_{v, q} = \sum_{x \in \Sigma^c \setminus \set{q}} v(x, h(x), h(q))$ and the random variables
    \begin{align*}
        M_{v, q}        &= \max_{x \in \Sigma^c \setminus \set{q}, j \in [m]} \abs{v(x, j, h(q))} \; , \\
        \sigma^2_{v, q} &= \frac{1}{m} \sum_{x \in \Sigma^c \setminus \set{q}} \sum_{j \in [m]} v(x, j, h(q))^2 \; ,
    \end{align*}
    which only depend on the randomness of $h(q)$.
    Then for all $p \ge 2$
    \begin{align*} 
        \epcond{\left(V^{\text{\normalfont simple}}_{v, q}\right)^p}{h(q)}^{1/p}
            \le \Psi_p\left(K_c \gamma_p^{c - 1} M_{v, q}, K_c \gamma_p^{c - 1} \sigma_{v, q}^2  \right)
        \; ,
    \end{align*}
    where $K_c = L_1 \left(L_2 c \right)^{c - 1}$, $L_1$ and $L_2$ are universal constants, and 
    \begin{align*}
        \gamma_p = \frac{\max\!\set{\log(m) + \log\!\left(\tfrac{\sum_{x \in \Sigma^c} \norm{v[x]}{2}^2}{\max_{x \in \Sigma^c} \norm{v[x]}{2}^2}\right)/c, p}}
            {\log\!\left( e^2 m \left( \max_{x \in \Sigma^c} \frac{\norm{v[x]}{1}^2}{\norm{v[x]}{2}^2} \right)^{-1} \right)}
        \; .
    \end{align*}
}{
    Let $h \colon \Sigma^c \to [m]$ be a simple tabulation hash function and let $q \in \Sigma^c$ be a designated query element.
    Let $v \colon \Sigma^{c} \times [m] \times [m] \to \R$ a value function, and assume that $\sum_{j \in [m]} v(x, j, k) = 0$ for all keys $x \in U$ and all $k \in [m]$.
    Define the random variable $V^{\text{\normalfont simple}}_{v, q} = \sum_{x \in \Sigma^c \setminus \set{q}} v(x, h(x), h(q))$ and the random variables
    \begin{align*}
        M_{v, q}        &= \max_{x \in \Sigma^c \setminus \set{q}, j \in [m]} \abs{v(x, j, h(q))} \; , \\
        \sigma^2_{v, q} &= \frac{1}{m} \sum_{x \in \Sigma^c \setminus \set{q}} \sum_{j \in [m]} v(x, j, h(q))^2 \; ,
    \end{align*}
    which only depend on the randomness of $h(q)$.
    Then for all $p \ge 2$
    \begin{align*} 
        \epcond{\left(V^{\text{\normalfont simple}}_{v, q}\right)^p}{h(q)}^{1/p}
            \le \Psi_p\left(K_c \gamma_p^{c - 1} M_{v, q}, K_c \gamma_p^{c - 1} \sigma_{v, q}^2  \right)
        \; ,
    \end{align*}
    where $K_c = L_1 \left(L_2 c \right)^{c - 1}$, $L_1$ and $L_2$ are universal constants, and 
    \begin{align*}
        \gamma_p = \frac{\max\!\set{\log(m) + \log\!\left(\tfrac{\sum_{x \in \Sigma^c} \norm{v[x]}{2}^2}{\max_{x \in \Sigma^c} \norm{v[x]}{2}^2}\right)/c, p}}
            {\log\!\left( e^2 m \left( \max_{x \in \Sigma^c} \frac{\norm{v[x]}{1}^2}{\norm{v[x]}{2}^2} \right)^{-1} \right)}
        \; .
    \end{align*}
}

\restateable[theorem]{mixed-tab-query}{
    Let $h \colon \Sigma^c \to [m]$ be a mixed tabulation hash function and let $q \in \Sigma^c$ be a designated query element.
    Let $v \colon \Sigma^{c} \times [m] \times [m] \to \R$ a value function, and assume that $\sum_{j \in [m]} v(x, j, k) = 0$ for all keys $x \in U$ and all $k \in [m]$.
    Define the random variable $V^{\text{\normalfont simple}}_{v, q} = \sum_{x \in \Sigma^c \setminus \set{q}} v(x, h(x), h(q))$ and the random variables
    \begin{align*}
        M_{v, q}        &= \max_{x \in \Sigma^c \setminus \set{q}, j \in [m]} \abs{v(x, j, h(q))} \; , \\
        \sigma^2_{v, q} &= \frac{1}{m} \sum_{x \in \Sigma^c \setminus \set{q}} \sum_{j \in [m]} v(x, j, h(q))^2 \; ,
    \end{align*}
    which only depend on the randomness of $h(q)$.
    For all $p \ge 2$ then
    \begin{align}
        \epcond{\left(V^{\text{\normalfont simple}}_{v, q}\right)^p}{h(q)}^{1/p}
            \le \Psi_p\left(K_c \gamma_p^{c} M_{v, q},  K_c \gamma_p^{c} \sigma_{v, q}^2 \right)
    \end{align}
    where $K_c = L_1 \left(L_2 c \right)^{c}$, $L_1$ and $L_2$ are universal constants, and 
    \begin{align*}
        \gamma_p = \max\!\set{1,  \frac{\log(m)}{\log(\abs{\Sigma})}, \frac{p}{\log(\abs{\Sigma} )} }
        \; .
    \end{align*}
}{
    Let $h \colon \Sigma^c \to [m]$ be a mixed tabulation hash function and let $q \in \Sigma^c$ be a designated query element.
    Let $v \colon \Sigma^{c} \times [m] \times [m] \to \R$ a value function, and assume that $\sum_{j \in [m]} v(x, j, k) = 0$ for all keys $x \in U$ and all $k \in [m]$.
    Define the random variable $V^{\text{\normalfont simple}}_{v, q} = \sum_{x \in \Sigma^c \setminus \set{q}} v(x, h(x), h(q))$ and the random variables
    \begin{align*}
        M_{v, q}        &= \max_{x \in \Sigma^c \setminus \set{q}, j \in [m]} \abs{v(x, j, h(q))} \; , \\
        \sigma^2_{v, q} &= \frac{1}{m} \sum_{x \in \Sigma^c \setminus \set{q}} \sum_{j \in [m]} v(x, j, h(q))^2 \; ,
    \end{align*}
    which only depend on the randomness of $h(q)$.
    For all $p \ge 2$ then
    \begin{align}
        \epcond{\left(V^{\text{\normalfont simple}}_{v, q}\right)^p}{h(q)}^{1/p}
            \le \Psi_p\left(K_c \gamma_p^{c} M_{v, q},  K_c \gamma_p^{c} \sigma_{v, q}^2 \right)
    \end{align}
    where $K_c = L_1 \left(L_2 c \right)^{c}$, $L_1$ and $L_2$ are universal constants, and 
    \begin{align*}
        \gamma_p = \max\!\set{1,  \frac{\log(m)}{\log(\abs{\Sigma})}, \frac{p}{\log(\abs{\Sigma} )} }
        \; .
    \end{align*}
}

%% file: introduction/technical-overview.tex

\subsection{Technical Overview}\label{sec:techniques}

\subsubsection{Fully Random Hashing}

\paragraph{Sub-Gaussian bounds}
A random variable $X$ is said to be sub-Gaussian with parameter $\sigma$ if $\pnorm{X}{p} \le \sqrt{p} \sigma$ for all $p \ge 2$.
It is a well-known fact that the sum of independent bounded random variables are sub-Gaussian.
In the context of fully random hashing, we have that
\begin{align}\label{eq:p-norm-hoeffding}
    \pnorm{\sum_{x \in U} v(x, h(x))}{p}
        \le \sqrt{4 p} \sqrt{\sum_{x \in U} \norm{v[x]}{\infty}^2 }
    \; .
\end{align}
A natural question is whether this is the best sub-Gaussian bound we can get.
If we are just interested in the contribution to a single bin, i.e., $v(x, j) = w_x([j = 0] - \tfrac{1}{m})$, then we can obtain a better sub-Gaussian bound.
By using the result of Oleszkiewicz~\cite{oleszkiewicz2003nonsymmetric}, we get that
\begin{align}\label{eq:p-norm-bernoulli}
    \pnorm{\sum_{x \in U} v(x, h(x))}{p}
        \le L \sqrt{\frac{p}{\log m}} \sqrt{\sum_{x \in U} w_x^2 }
    \; ,
\end{align}
where $L$ is a universal constant.
This shows that \cref{eq:p-norm-hoeffding} can be improved in certain situations.
We improve on this by proving a generalization of \cref{eq:p-norm-bernoulli}.
We show that
\begin{align}\label{eq:best-sub-gaussian}
    \pnorm{\sum_{x \in U} v(x, h(x))}{p}
        \le L \sqrt{\frac{p}{\log\left(\frac{e^2 m \sum_{x \in U} \norm{v[x]}{\infty}^2}{\sum_{x \in U} \norm{v[x]}{2}^2 }\right)}} \sqrt{\sum_{x \in U} \norm{v[x]}{\infty}^2}
    \; ,
\end{align}
where $L$ is a universal constant.
It is easy to check that if $v(x, j) = w_x([j = 0] - \tfrac{1}{m})$ then it reduces to \cref{eq:p-norm-bernoulli} and that it is stronger than \cref{eq:p-norm-hoeffding}.

\paragraph{Moments for general random variables}
As part of our analysis we develop a couple of lemmas for general random variables which might be of independent interest.
We prove a lemma that provides a simple bound for weighted sums of independent and identically distributed random variables.
\restateable[lemma]{gen-moment-sym}{
    Let $(X_i)_{i \in [n]}$ and $X$ be independent and identically distributed symmetric random variables, and let $(a_i)_{i \in [n]}$ be a sequence of reals.\footnote{A symmetric random variable, $X$, is a random variable that is symmetric around zero, i.e., $\prb{X \ge t} = \prb{-X \ge t}$ for all $t \ge 0$.}
    If $p \ge 2$ is an even integer then
    \begin{align*}
        \pnorm{\sum_{i \in [n]} a_{i} X_i}{p}
            \le K \sup\setbuilder{\frac{p}{s} \left( \frac{\sum_{i \in [n]} a_i^s}{p} \right)^{1/s} \pnorm{X}{s}}{2 \le s \le p}
        \; ,
    \end{align*}
    where $K \le 4 e$ is a universal constant.
}{
    Let $(X_i)_{i \in [n]}$ and $X$ be independent and identically distributed symmetric random variables, and let $(a_i)_{i \in [n]}$ be a sequence of reals.
    If $p \ge 2$ is an even integer then
    \begin{align*}
        \pnorm{\sum_{i \in [n]} a_{i} X_i}{p}
            \le K \sup\setbuilder{\frac{p}{s} \left( \frac{\sum_{i \in [n]} a_i^s}{p} \right)^{1/s} \pnorm{X}{s}}{2 \le s \le p}
        \; ,
    \end{align*}
    where $K \le 4 e$ is a universal constant.
}
If we consider Laplace distributed random variables then it is possible to show that \Cref{lem:gen-moment-sym} is tight up to a universal constant.
Thus a natural question to ask is whether \Cref{lem:gen-moment-sym} is tight, i.e., can we prove a matching lower bound.
But unfortunately, if you consider Gaussian distributed variables then we see that \Cref{lem:gen-moment-sym} is not tight.
It would be nice if there existed a simple modification of \Cref{lem:gen-moment-sym} which had a matching lower bound.

\paragraph{Moments of functions of random variables}
As part of the analysis of tabulation hashing, we will need to analyze random variables of the form $\Psi_p(X, Y)$ where $X$ and $Y$ are random variables.
More precisely, we have to bound $\pnorm{\Psi_p(X, Y)}{p}$.
It is not immediately clear how one would do this but we prove a general lemma that helps us in this regard.
\restateable[lemma]{fn-moment}{
    Let $f \colon \R_{\ge 0}^n \to \R_{\ge 0}$ be a non-negative function which is monotonically increasing in every argument, and assume that there exist positive reals $(\alpha_i)_{i \in [n]}$ and $(t_i)_{i \in [n]}$ such that for all $\lambda \ge 0$
        \begin{align*}
            f(\lambda^{\alpha_0} t_0, \ldots, \lambda^{\alpha_{n - 1}} t_{n - 1})
                \le \lambda f(t_0, \ldots, t_{n - 1})
            \; .
        \end{align*}
        Let $(X_i)_{i \in [n]}$ be non-negative random variables.
        Then for all $p \ge 1$ we have that
        \begin{align*}
            \pnorm{f(X_0, \ldots, X_{n - 1})}{p}
                \le n^{1/p} \max_{i \in [n]} \left( \frac{\pnorm{X_i}{p/\alpha_i}}{t_i} \right)^{1/\alpha_i} f(t_0, \ldots, t_{n - 1})
            \; .
        \end{align*}
}{
    Let $f \colon \R_{\ge 0}^n \to \R_{\ge 0}$ be a non-negative function which is monotonically increasing in every argument, and assume that there exist positive reals $(\alpha_i)_{i \in [n]}$ and $(t_i)_{i \in [n]}$ such that for all $\lambda \ge 0$
        \begin{align*}
            f(\lambda^{\alpha_0} t_0, \ldots, \lambda^{\alpha_{n - 1}} t_{n - 1})
                \le \lambda f(t_0, \ldots, t_{n - 1})
            \; .
        \end{align*}
        Let $(X_i)_{i \in [n]}$ be non-negative random variables.
        Then for all $p \ge 1$ we have that
        \begin{align*}
            \pnorm{f(X_0, \ldots, X_{n - 1})}{p}
                \le n^{1/p} \max_{i \in [n]} \left( \frac{\pnorm{X_i}{p/\alpha_i}}{t_i} \right)^{1/\alpha_i} f(t_0, \ldots, t_{n - 1})
            \; .
        \end{align*}
}

If we can choose $t_i = \pnorm{X_i}{p/\alpha_i}$ for all $i \in [n]$, then we get the nice expression
\begin{align*}
    \pnorm{f(X_0, \ldots, X_{n - 1})}{p}
        \le n^{1/p} f(\pnorm{X_0}{p/\alpha_0}, \ldots, \pnorm{X_{n - 1}}{p/\alpha_{n - 1}})
    \; .
\end{align*}
Now the result is natural to compare to the triangle inequality that says that $\pnorm{X + Y}{p} \le \pnorm{X}{p} + \pnorm{Y}{p}$, which corresponds to considering $f(x, y) = x + y$, and to Cauchy-Schwartz that says that $\pnorm{X Y}{p} \le \pnorm{X}{2p} \pnorm{Y}{2p}$, which corresponds to $f(x, y) = xy$.
These two examples might point to that the $n^{1/p}$ is superfluous, but by considering $f(x_0, \ldots, x_{n - 1}) = \max\!\set{x_0, \ldots, x_{n - 1}}$ and Gaussian distributed variables, it can be shown that \Cref{lem:fn-moment} is tight up to a constant factor.

\subsubsection{Tabulation Hashing}

\paragraph{Symmetrization}
The analyses of chaoses have mainly focused on two types of chaoses: Chaoses generated by non-negative random variables and chaoses generated by symmetric random variables.
It might appear strange that focus has not been on chaoses generated by mean zero random variables.
The reason is that a symmetrization argument reduces the analysis of chaoses generated by mean zero random variables to the analysis of chaoses generated by symmetric random variables.
More precisely, a standard symmetrization shows that
\begin{align*}
        2^{-c} \pnorm{\sum_{i_0, \ldots, i_{c - 1} \in [n]} a_{i_0, \ldots, i_{c - 1}} \prod_{j \in [c]} \eps_{i_j}^{(j)} X_{i_j}^{(j)}}{p}
            &\le \pnorm{\sum_{i_0, \ldots, i_{c - 1} \in [n]} a_{i_0, \ldots, i_{c - 1}} \prod_{j \in [c]} X_{i_j}^{(j)}}{p}
            \\&\le 2^c \pnorm{\sum_{i_0, \ldots, i_{c - 1} \in [n]} a_{i_0, \ldots, i_{c - 1}} \prod_{j \in [c]} \eps_{i_j}^{(j)} X_{i_j}^{(j)}}{p}
        \; ,
\end{align*}
where $(\eps_{i}^{(j)})_{i \in [n], j \in [n]}$ are independent Rademacher variables.\footnote{A Rademacher variable, $\eps$, is a random variable chosen uniformly from the set $\set{-1, 1}$, i.e., $\prb{\eps = -1} = \prb{\eps = 1} = \tfrac{1}{2}$.}

In our case, we can assume that $v(x, h(x))$ is a mean zero random variable but is not necessarily symmetric.
We can remedy this by using the same idea of symmetrization.
We define $\eps \colon \Sigma^c \to \set{-1, 1}$ to be a simple tabulation sign function, more precisely, we have a fully random table, $T_\eps \colon [c] \times \Sigma \to \set{-1, 1}$, and $\eps$ is then defined by $\eps(\alpha_0, \ldots, \alpha_{c - 1}) = \prod_{i \in [c]} T(i, \alpha_i)$.
We then prove that for all $p \ge 2$
\begin{align}\label{eq:simple-symmetrization}
    2^{-c}\pnorm{\sum_{x \in \Sigma^c} \eps(x) v(x, h(x))}{p}
        \le \pnorm{\sum_{x \in \Sigma^c} v(x, h(x))}{p}
        \le 2^c \pnorm{\sum_{x \in \Sigma^c} \eps(x) v(x, h(x))}{p}
    \; .
\end{align}
The power of symmetrization lies in the fact that we get to assume that $v$ is symmetric in the analysis without actually changing the value functions.

Somewhat surprisingly, we are able to improve the moment bound of Dahlgaard et al.~\cite{Dahlgaard2015} just by using symmetrization.
Their result has a doubly exponential dependence on the size of the moment, $p$, which stems from a technical counting argument where they bound the number of terms which does not have an independent factor when expanding the expression $\left(\sum_{x \in \Sigma^c} v(x, h(x)) \right)^p$.
It appears difficult to directly improve their counting argument but by using \cref{eq:simple-symmetrization} we are able to circumvent this.
Thus, just by using symmetrization and the insights of Dahlgaard et al.~\cite{Dahlgaard2015} we obtain the following result.
\begin{lemma}\label{lem:simple-gaussian-simplified}
    Let $h \colon \Sigma^c \to [m]$ be a simple tabulation function, $\eps \colon \Sigma^c \to \set{-1, 1}$ be a simple tabulation sign function, and $v \colon \Sigma^{c} \times [m] \to \R$ be a value function.
    Then for every real number $p \ge 2$
    \begin{align*}\begin{split}
        &\pnorm{\sum_{x \in \Sigma^{c}} v(x, h(x)) }{p}
            \le 2^c \pnorm{\sum_{x \in \Sigma^{c}} \eps(x) v(x, h(x)) }{p}
            \le \sqrt{4p}^{c} \sqrt{\sum_{x \in \Sigma^c} \norm{v[x]}{\infty}^2 }
        \; .
    \end{split}\end{align*}
\end{lemma}

\paragraph{General value functions}
For most applications of hashing, we are either interested in the number of balls landing in a bin or in the number of elements hashing below a threshold.
But we are studying the more general setting where we have a value function.
A natural question is whether it is possible to obtain a simpler proof for the simpler settings.
We do not believe this to be the case since the general setting of value functions will naturally show up when proving results by induction on $c$.
More precisely, let us consider the case where we are interested in the number of elements from a set, $S \subseteq \Sigma^c$, that hash to $0$.
We then want to bound $\sum_{x \in S} \left(\indicator{h(x) = 0} - \tfrac{1}{m}\right) = \sum_{x \in \Sigma^c} \indicator{x \in S}\left(\indicator{h(x) = 0} - \tfrac{1}{m}\right)$.
This can be rewritten as\footnote{For a partial key $y = (\beta_0, \ldots, \beta_{c - 2}) \in \Sigma^{c - 1}$, we let $h(y) = \bigxor_{i \in [c - 1]} T(i, \beta_i)$.}
\begin{align*}
    \sum_{x \in \Sigma^c} \indicator{x \in S}\left(\indicator{h(x) = 0} - \tfrac{1}{m}\right)
        = \sum_{\alpha \in \Sigma} \sum_{y \in \Sigma^{c - 1}} \indicator{(y, \alpha) \in S}\left(\indicator{h(y) \xor T(c - 1, \alpha) = 0} - \tfrac{1}{m}\right)
    \; .
\end{align*}
So if we define the value function $v' \colon \Sigma \times [m] \to \R$ by
\[v'(\alpha, j) = \sum_{y \in \Sigma^{c - 1}} \indicator{(y, \alpha) \in S}\left(\indicator{h \xor j = 0} - \tfrac{1}{m}\right) \; ,\]
then we get that $\sum_{x \in S} \left(\indicator{h(x) = 0} - \tfrac{1}{m}\right) = \sum_{\alpha \in \Sigma} v'(\alpha, T(c - 1, \alpha))$.
Thus, we see that general value functions are natural to consider in the context of tabulation hashing.

Instead of shying away from general value functions, we embrace them.
This force us look at the problem differently and guides us in the correct direction.
Using this insight naturally leads us to use \cref{eq:best-sub-gaussian} and we prove the following moment bound, which is strictly stronger than \Cref{lem:simple-gaussian-simplified}.
\restateable[lemma]{simple-tab-hoeffding}{
    Let $h \colon \Sigma^c \to [m]$ be a simple tabulation function, $\eps \colon \Sigma^c \to \set{-1, 1}$ be a simple tabulation sign function, and $v \colon \Sigma^c \times [m] \to \R$ be value function.
    Then for every real number $p \ge 2$
    \begin{align*}
        &\pnorm{\sum_{x \in \Sigma^c} \eps(x) v(x, h(x))}{p}
            \le \sqrt{K_c \frac{p \left(\max\!\set{p, \log(m)}\right)^{c - 1}}{ \log\!\left( 1 + \frac{m \sum_{x \in \Sigma^{c}} \norm{v[x]}{\infty}^2}{\sum_{x \in \Sigma^{c}} \norm{v[x]}{2}^2} \right)^c }} \sqrt{\sum_{x \in \Sigma^c} \norm{v[x]}{\infty}^2}
        \; ,
    \end{align*}
    where $K_c = (L c)^c$ for a universal constant $L$.
}{
    Let $h \colon \Sigma^c \to [m]$ be a simple tabulation function, $\eps \colon \Sigma^c \to \set{-1, 1}$ be a simple tabulation sign function, and $v \colon \Sigma^c \times [m] \to \R$ be value function.
    Then for every real number $p \ge 2$
    \begin{align*}
        &\pnorm{\sum_{x \in \Sigma^c} \eps(x) v(x, h(x))}{p}
            \le \sqrt{K_c p \left(\max\!\set{p, \log(m)}\right)^{c - 1} \frac{\sum_{x \in \Sigma^c} \norm{v[x]}{\infty}^2}{\log\!\left( \frac{e^2 m \sum_{x \in \Sigma^{c}} \norm{v[x]}{\infty}^2}{\sum_{x \in \Sigma^{c}} \norm{v[x]}{2}^2} \right)^c}}
        \; ,
    \end{align*}
    where $K_c = (L c)^c$ for a universal constant $L$.
}
This statement is often weaker than \Cref{thm:simple-tab-moments} but perhaps a bit surprisingly, we will use \Cref{lem:simple-tab-hoeffding} as an important step in the proof of \Cref{thm:simple-tab-moments}.

\paragraph{Sum of squares of simple tabulation hashing}
A key element when proving \Cref{thm:simple-tab-moments} is bounding the sums of squares
\begin{align}\label{eq:intro-sum-of-squares}
    \sum_{j \in [m]} \left( \sum_{x \in \Sigma^c} v(x, h(x) \xor j) \right)^2
    \; .
\end{align}
This was also one of the main technical challenges for the analysis of Aamand et al.~\cite{aamand2020}.
Instead of analyzing \cref{eq:intro-sum-of-squares}, we will analyze a more general problem: Let $v_i \colon \Sigma^c \times [m] \to \R$ be a value function $i \in [k]$, we then want to understand the random variable.
\begin{align}\label{eq:intro-simple-chaos}
    \sum_{\substack{j_0, \ldots, j_{k - 1} \in [m] \\ \bigxor_{i \in [k]} j_i = 0}}\sum_{x_0, \ldots, x_{k - 1} \in \Sigma^c} \prod_{i \in [k]} v_i(x_i, j_i \xor h(x_i))
\end{align}
If we have $k = 2$ and $v_0 = v_1$ then this corresponds to \cref{eq:intro-sum-of-squares}.
By using a decoupling argument, it is possible to reduce the analysis of \cref{eq:intro-simple-chaos} to the analysis of hash-based sums for simple tabulation hashing.
We can then use \Cref{lem:simple-tab-hoeffding} to obtain the following lemma.

\restateable[lemma]{simple-sum-squares}{
    Let $h \colon \Sigma^c \to [m]$ be a simple tabulation function, $\eps \colon \Sigma^c \to \set{-1, 1}$ be a simple tabulation sign function, and $v_i \colon \Sigma^c \times [m] \to \R$ be a value function for $i \in [k]$.
    For every real number $p \ge 2$
    \begin{align*}\begin{split}
        &\pnorm{\sum_{j \in [m]} \left( \sum_{x \in \Sigma^c} \eps(x) v(x, h(x)) \right)^2 }{p}
            \le \left( \frac{L c \max\!\set{p, \log(m)}}{\log\!\left( \frac{e^2 m \sum_{x \in \Sigma^c} \norm{v[x]}{2}^2}{\sum_{x \in \Sigma^c} \norm{v[x]}{1}^2} \right)} \right)^{c} \sum_{x \in \Sigma^c} \norm{v[x]}{2}^2 
        \; ,
    \end{split}\end{align*}
    where $L$ is a universal constant.
}{
    Let $h \colon \Sigma^c \to [m]$ be a simple tabulation function, $\eps \colon \Sigma^c \to \set{-1, 1}$ be a simple tabulation sign function, and $v_i \colon \Sigma^c \times [m] \to \R$ be a value function for $i \in [k]$.
    For every real number $p \ge 2$
    \begin{align*}\begin{split}
        &\pnorm{\sum_{j \in [m]} \left( \sum_{x \in \Sigma^c} \eps(x) v(x, h(x) \xor j) \right)^2 }{p}
            \le \left( \frac{L c \max\!\set{p, \log(m)}}{\log\!\left( \frac{e^2 m \sum_{x \in \Sigma^c} \norm{v[x]}{2}^2}{\sum_{x \in \Sigma^c} \norm{v[x]}{1}^2} \right)} \right)^{c} \sum_{x \in \Sigma^c} \norm{v[x]}{2}^2 
        \; ,
    \end{split}\end{align*}
    where $L$ is a universal constant.
}

\paragraph{Proving the main result}
The proof of \Cref{thm:simple-tab-moments} is by induction on $c$.
We will use \Cref{thm:sampling-moments} on one of the characters while fixing the other characters.
This will give us an expression of the form
\begin{align*}
    \pnorm{\Psi_p\left(\max_{\alpha \in \Sigma, j \in [m]} \abs{\sum_{y \in \Sigma^{c - 1}} v((y, \alpha), h(y) \xor j)}, \frac{\displaystyle \sum_{\alpha \in \Sigma, j \in [m]} \left( \sum_{y \in \Sigma^{c - 1}} v((y, \alpha), h(y) \xor j) \right)^2 }{m} \right)}{p}
        \; .
\end{align*}
By applying \Cref{lem:fn-moment}, we bound this by
\begin{align*}
    \Psi_p\left( \pnorm{\max_{\alpha \in \Sigma, j \in [m]} \abs{\sum_{y \in \Sigma^{c - 1}} v((y, \alpha), h(y) \xor j)} }{p}, \pnorm{\frac{\displaystyle \sum_{\alpha \in \Sigma, j \in [m]} \left( \sum_{y \in \Sigma^{c - 1}} v((y, \alpha), h(y) \xor j) \right)^2 }{m} }{p} \right)
    \; .
\end{align*}
We will bound $\pnorm{\max_{\alpha \in \Sigma, j \in [m]} \abs{\sum_{y \in \Sigma^{c - 1}} v((y, \alpha), h(y) \xor j)} }{p}$ by using the induction hypothesis, and we bound $\pnorm{\frac{\sum_{\alpha \in \Sigma, j \in [m]} \left( \sum_{y \in \Sigma^{c - 1}} v((y, \alpha), h(y) \xor j) \right)^2 }{m} }{p}$ by using \Cref{lem:simple-sum-squares}.
While this sketch is simple, the actual proof is quite involved and technical since one has to be very careful with the estimates.

%% file: introduction/related-work.tex

\subsection{Mixed Tabulation Hashing in Context}\label{sec:context}

Our concentration bounds for mixed tabulation hashing are similar to
those Aamand et al. \cite{aamand2020} for their tabulation-permutation
hashing scheme and the schemes also have very similar efficiency,
roughly a factor 2 slower than simple tabulation and orders of
magnitude faster than any alternative with similar known concentration
bounds. We shall make a more detailed comparison with
tabulation-permutation in Section \ref{sec:tab-perm}.

As mentioned in the beginning of the introduction, the big advantage of proving concentration bounds for mixed tabulation hashing rather than for
tabulation-permutation is that mixed
tabulation hashing has many other strong probabilistic properties
that can now be used in tandem with strong concentration. This makes
mixed tabulation an even stronger candidate to replace abstract uniform
hashing in real implementations of algorithms preserving many of the
asymptotic performance guarantees.

Mixed tabulation inherits all the nice probabilistic properties known
for simple and twisted tabulation\footnote{This is not a black box reduction, but
both twisted and mixed tabulation hashing applies simple tabulation to a
some changed keys, so any statement holding for arbitrary sets of input keys is still
valid. Moreover, mixed tabulation with one derived character
corresponds to mixed tabulation applied to keys with an added 0-character head, and having more derived characters does not give worse results.}. Dahlgaard et
al.~\cite{Dahlgaard2015} introduced mixed tabulation hashing to further get
good statistics over $k$-partitions as used in classic streaming
algorithms for counting of distinct elements by Flajolet et
al. \cite{Flajolet85counting,Flajolet07hyperloglog,Heule13hyperloglog},
and for fast set similarity in large-scale machine learning by Li et
al.  \cite{li12oneperm,Shrivastava14oneperm,Shrivastava14densify}.

\paragraph*{Selective full randomness with mixed tabulation}
The main result of Dahlgaard et al.~\cite{Dahlgaard2015} for mixed
tabulation is that it has a certain kind of selective full randomness
(they did not have a word for it). An $\ell$-bit mask $M$ with
don't cares is of the form $\{0,1,\textnormal?\}^\ell$. An $\ell$-bit
string $B\in \{0,1\}^\ell$ matches $M$ if it is obtained from $M$ by
replacing each ? with a 0 or a 1. Given a hash function returning
$\ell$-bit hash values, we can use $M$ to select the set $Y$ of keys
that match $M$. Consider a mixed tabulation hash function
$h:\Sigma^c\to \{0,1\}^\ell$ using $d$ derived characters. The main
result of Dahlgaard et al.~{\cite[Theorem 4]{Dahlgaard2015}} is that
if the expected number of selected keys is less than $|\Sigma|/2$,
then, w.h.p., the free (don't care) bits of the hash values of $Y$
are fully random and independent. More formally,
  
\begin{theorem}[Dahlgaard et al.~{\cite[Theorem 4]{Dahlgaard2015}}]\label{thm:mix-indep}
  Let $h:\Sigma^c\to \{0,1\}^\ell$ be a mixed tabulation hash function
  using $d$ derived characters. Let $M$ be an $\ell$-bit
  mask with don't cares. For a given
  key set $X\subseteq \Sigma^c$, let $Y$ be the set of keys from
  $X$ with hash values matching $M$. If $\ep{|Y|} \le                                               |\Sigma|/(1+\Omega(1))$, then the free bits of the hash values in $Y$
    are fully random with probability
  $1 - O(|\Sigma|^{1-\floor{d/2}})$.
\end{theorem}
The above result is best possible in that since we only have $O(|\Sigma|)$
randomness in the tables, we cannot hope for full randomness of
an asymptotically larger set $Y$.

In the applications from \cite{Dahlgaard2015},
we also want the size of the set $Y$ to be concentrated
around its mean and by \Cref{cor:mixed-tail}, the concentration is
essentially as strong as with fully random hashing and it holds
for any $d\geq 1$.

In \cite{Dahlgaard2015} they only proved weaker concentration
bounds for the set $Y$ selected in \Cref{thm:mix-indep}.
Based on the concentration bounds for simple tabulation by P\v{a}tra\c{s}cu
and Thorup \cite{Patrascu2012}, they proved that if the set $Y$
from \ref{thm:mix-indep} had
$\E[Y]\in [|\Sigma|/8, 3|\Sigma|/4]$, then within the
  same probability of   $1 - O(|\Sigma|^{1-\floor{d/2}})$, it has
  \begin{equation}\label{eq:old-mix-tab-Y}
  |Y|=\E[Y]\left(1\pm O\left(\sqrt{\frac{\log |\Sigma|(\log\log|\Sigma|)^2}
    {|\Sigma|}}\right)\right).
  \end{equation}
  With \Cref{cor:mixed-tail}, for $\E[Y]=\Theta(|\Sigma|)$, we immediately tighten \req{eq:old-mix-tab-Y}
  to the cleaner
  \begin{equation}\label{eq:mix-tab-Y}
    |Y|=\E[Y]\left(1\pm O\left(
    \sqrt{\frac{\log |\Sigma|}{|\Sigma|}}\right)\right).
  \end{equation}
  While the improvement is ``only'' a factor $(\log\log|\Sigma|)^2$,
  the important point here is that \req{eq:mix-tab-Y} is the
  asymptotic bound we would get with fully-random hashing. Also,
  while Dahlgaard et al.\ only proved \req{eq:old-mix-tab-Y} for the
  special case of $\E[Y]\in [|\Sigma|/8, 3|\Sigma|/4]$, our
  \req{eq:mix-tab-Y} is just a special case of \Cref{cor:mixed-tail}
  which holds for arbitrary values of $\E[Y]$ and arbitrary value functions.

  Dahlgaard et al. presented some very nice applications of mixed
  tabulation to problems in counting and machine learning and machine
  learning. The way they use \Cref{thm:mix-indep} is rather
  subtle.

  \subsubsection{Mixed Tabulation Hashing Versus Tabulation-Permutation Hashing}\label{sec:tab-perm}

  As mentioned earlier, our new concentration bounds are similar
to those proved  by 
Aamand et al. \cite{aamand2020} for their tabulation-permutation hashing
scheme. However, now we also have moment bounds covering
the tail, and we have the first understanding of what happens when
$c$ is not constant. It is not clear if this new understanding
applies to tabulation-permutation. As discussed above, the advantage of 
having the concentration bounds for mixed tabulation hashing is that we can
use them in tandem with the independence result from Theorem \ref{thm:mix-indep}, which does not hold for tabulation-permutation.

Tabulation-permutation is similar to mixed tabulation hashing in its resource
consumption. Consider the mapping $\Sigma^c\to\Sigma^c$. Tabulation-permutation first uses simple tabulation $h:\Sigma^c\to\Sigma^c$. Next it applies
a random permutation $\pi_i:\Sigma\stackrel{1-1}{\to}\Sigma$ to each
output character $h(x)_i$, that is, $x\mapsto(\pi_1(h(x)_1),\ldots,\pi_c(h(x)_c)$. Aamand et al.~\cite{aamand2020} also suggest tabulation-1permutation
hashing, which only permutes the most significant character. This
scheme does not provide concentration for all value functions, but it
does work if we select keys from intervals.

Aamand et al.~\cite{aamand2020} already made a thorough experimental
and theoretical comparison between tabulation-permutation, mixed
tabulation, and many other schemes. In this comparison, mixed
tabulation played the role of a similar scheme with not as strong
known concentration bounds.  In the experiments, mixed tabulation hashing with
$c$ derived characters performed similar to tabulation-permutation in
speed. Here we proved stronger concentration bounds for mixed
tabulation even with a single character, where it should perform
similar to tabulation-1permutation (both use $c+1$ lookups). Both
mixed tabulation hashing and tabulation-permutation hashing were orders of magnitude
faster than any alternative with similar known concentration
bounds. We refer to~\cite{aamand2020,Dahlgaard2017} for more details.  In particular,
\cite{Dahlgaard2017} compares mixed tabulation with popular cryptographic hash
functions that are both slower and have no guarantees in these
algorithmic contexts.

One interesting advantage of mixed tabulation hashing over
tabulation-permutation hashing is that mixed tabulation hashing, like simple
tabulation hashing, only needs randomly filled character tables.
In contrast, tabulation-permutation needs tables that represent
permutations. Thus, all we need to run mixed tabulation hashing is a
pointer to some random bits.  These could be in read-only memory
shared across different applications.  Read-only memory is much less
demanding than standard memory since there can be no write-conflicts,
so we could imagine some special large, fast, and cheap read-only
memory, pre-filled with random bits, e.g., generated by a
quantum-device. This would open up for larger characters, e.g., 16- or
32-bit characters, and it would free up the cache for other
applications.

%% file: introduction/preliminaries.tex

In this section, we will introduce the notation which will be used throughout the paper.
We will start by introducing the notation from probability theory that we need and afterwards we will introduce notation that will help in the reasoning about tabulation hashing.

We will use the following basic mathematical notation: We define $\N$ the set of non-negative integers, for $n \in \N$ we shall define $[n] = \set{0, \ldots, n - 1}$, in particular $[n] = \emptyset$, and for an event $\mathcal{A}$ we shall define $\indicator{\mathcal{A}}$ to be the indicator on $\mathcal{A}$, i.e.,  $\indicator{\mathcal{A}} = 1$ if $\mathcal{A}$ is true and $\indicator{\mathcal{A}} = 0$ otherwise.

If $n \in \N$ is non-negative integer, $(X_i)_{i \in [n]}$ are real variables, and $j \in [n + 1]$ we shall define $X_{< j} = \sum_{\substack{i \in [n], i<j}} X_i = \sum_{i \in [j]} X_i$.
Similarly, for sets $(A_i)_{i \in [n]}$ and $j \in [n + 1]$ we shall define $A_{< j} = \bigcup_{i \in [j]} A_i$.    

We will be using the following version of Stirling's approximation~\cite{Maria1965} which holds for all integers $n$,
\begin{align}\label{eq:stirling}
    \Gamma(n + 1) = n! \le e \sqrt{n} \left( \frac{n}{e} \right)^n
        \; .
\end{align}

\subsection{Probability Theory}

In the following, we introduce the necessary notions of probability theory.
We will assume that we are given a probability space $(\Omega, \mathcal{F}, P)$ throughout the paper but we will often not state it explicitly.
We will be working with martingales and we shall therefore need notation and concepts from probability theory of a fairly general and abstract character.
For an introduction to measure and probability theory, see, for instance,~\cite{schilling2005}.

\begin{definition}
    Let $(X_i)_{i \in [n]}$ be random variables on the probability space $(\Omega, \mathcal{F}, P)$.
    We denote by $\mathcal{G} = \sigma((X_i)_{i \in [n]}) \subseteq \mathcal{F}$ the smallest $\sigma$-algebra where $(X_i)_{i \in [n]}$ are all $\mathcal{G}$-measurable.
\end{definition}

\begin{definition}[Conditional expectation]
    Let $X$ be a random variable on the probability space $(\Omega, \mathcal{F}, P)$, and let $\mathcal{G} \subseteq \mathcal{F}$ be a sub $\sigma$-algebra.
    If $\ep{\abs{X}} < \infty$ we can define the random variable $\epcond{X}{\mathcal{G}}$ to be the conditional expectation of $X$ given $\mathcal{G}$.
    It shall be $\mathcal{G}$-measurable and for all $G \in \mathcal{G}$ we have that $\ep{\mathds{1}_{G}\epcond{X}{\mathcal{G}}} = \ep{\mathds{1}_{G}X}$.

    We define the conditional expectation of $X$ given a random variable $Y$ by $\epcond{X}{Y} = \epcond{X}{\sigma(Y)}$.
\end{definition}

\begin{definition}[Filtration and adapted sequence]
    On a probability space $(\Omega, \mathcal{F}, P)$ a sequence of $(\mathcal{F}_i)_{i \in [n]}$ of sub $\sigma$-algebras is called a filtration if $\mathcal{F}_0 \subseteq \ldots \subseteq \mathcal{F}_{n - 1} \subseteq \mathcal{F}$.

    We say that a sequence of random variables $(X_i)_{i \in [n]}$ is adapted to a filtration $(\mathcal{F}_i)_{i \in [n]}$ if $X_i$ is $\mathcal{F}_i$-measurable for all $i \in [n]$.
    We call $(X_i, \mathcal{F}_i)_{i \in [n]}$ an adapted sequence.
\end{definition}

\begin{definition}[Martingale and martingale difference]
    We call an adapted sequence $(X_i, \mathcal{F}_i)_{i \in [n]}$ a martingale sequence if $\epcond{X_i}{\mathcal{F}_{i - 1}} = X_{i - 1}$ for all $i \in [n]$. (We define $\mathcal{F}_{-1} = \set{\emptyset, \Omega}$ and $X_{-1} = 0$.) 

    We call an adapted sequence $(Y_i, \mathcal{F}_i)_{i \in [n]}$ a martingale difference sequence if $\epcond{Y_i}{\mathcal{F}_{i - 1}} = 0$ for all $i \in [n]$. (We define $\mathcal{F}_{-1} = \set{\emptyset, \Omega}$.)

    It should be noted that if $(Y_i, \mathcal{F}_i)_{i \in [n]}$ is a martingale difference sequence then $(Y_{< i + 1}, \mathcal{F}_{i})_{i \in [n]}$ is a martingale sequence.
    Similarly, if $(X_i, \mathcal{F}_i)_{i \in [n]}$ is a martingale sequence then $(X_i - X_{i - 1}, \mathcal{F}_i)_{i \in [n]}$ is a martingale difference sequence.
\end{definition}


\begin{definition}[$p$-norm]
    Let $p \ge 1$ and $X$ be a random variable with $\ep{\abs{X}^p} < \infty$.
    We then define the $p$-norm of $X$ by $\pnorm{X}{p} = \ep{\abs{X}^p}^{1/p}$.

    Let $p \ge 1$, $\mathcal{G}$ be a sub $\sigma$-algebra, and $X$ be a random variable with $\epcond{\abs{X}^p}{\mathcal{G}} < \infty$.
    We then define the $p$-norm of $X$ conditioned on $\mathcal{G}$ by $\pnormcond{X}{p}{\mathcal{G}} = \epcond{\abs{X}^p}{\mathcal{G}}^{1/p}$.

    Similarly to conditional expectation, we will condition on random variables. 
    Let $p \ge 1$ and let $X$ and $Y$ be a random variables.
    We then define the $p$-norm of $X$ conditioned on $Y$ by $\pnormcond{X}{p}{Y} = \epcond{\abs{X}^p}{\sigma(Y)}^{1/p}$
\end{definition}

Now an important observation is that the $p$-norm is a seminorm which follows by the Minkowski inequality.

\begin{lemma}[Triangle inequality(Minkowski inequality)]
    Let $p \ge 1$ and let $X$ and $Y$ be random variables with $\ep{\abs{X}^p} < \infty$ and $\ep{\abs{Y}^p} < \infty$. Then $\ep{\abs{X + Y}^{p}} < \infty$ and $\pnorm{X + Y}{p} \le \pnorm{X}{p} + \pnorm{Y}{p}$.
\end{lemma}

\subsection{Tabulation Hashing}
We will need to reason about the individual characters of a key, $x \in \Sigma^c$, and for that, we need some notation.
Most of the definitions are taken from the paper by Aamand et~al.~\cite{aamand2020}.

\begin{definition}[Position characters]
    Let $\Sigma$ be an alphabet and $c > 0$ a positive integer.
    We call an element $(i, y) \in [c] \times \Sigma$ a position character of $\Sigma^c$.
\end{definition}

We will view a key $x = (y_0, \ldots, y_{c - 1}) \in \Sigma^c$ as a set of $c$ position characters, $\set{(0, y_0), \ldots (c - 1, y_{c - 1})} \subseteq [c] \times \Sigma$.
Let $h \colon \Sigma^c \to [2^l]$ be a simple tabulation hash function and let $T \colon [c] \times \Sigma \to [2^l]$ be the random function used to define $h$.
We will then overload the notation and for a set of position characters $y \subseteq [c] \times \Sigma$, define $h(y) = \bigxor_{\alpha \in y} T(y)$.
We note that this definition agrees with our correspondence between keys $x = (y_0, \ldots, y_{c - 1}) \in \Sigma^c$, and set of position characters, $\set{(0, y_0), \ldots (c - 1, y_{c - 1})} \subseteq [c] \times \Sigma$, that is, $h(x) = h(\set{(0, y_0), \ldots (c - 1, y_{c - 1})})$.
We have thus extended the domain of $h$ to $\powerset{[c] \times \Sigma}$.
For sets of position characters $x_1, x_2 \in \powerset{[c] \times \Sigma}$ we will write $x_1 \xor x_2$ for the symmetric difference, i.e., $x_1 \xor x_2 = (x_1 \cup x_2) \setminus (x_1 \cap x_2)$.
We note that with the extended domain for $h$ then $h(x_1 \xor x_2) = h(x_1) \xor h(x_2)$.




We will prove several of our statements by induction on the position characters and for this reason, we need the following definition.

\begin{definition}[Group of keys]
    Let $\set{\alpha_0, \ldots, \alpha_{r - 1}} = [c] \times \Sigma^c$ be an enumeration of the position characters of $\Sigma^c$.
    For each $i \in [r]$ we denote by $G_i \subseteq \Sigma^c$ the $i$'th group of keys with respect to the ordering of position characters, and define it by $G_i = \setbuilder{x \in \Sigma^c}{\alpha_i \in x, x \subseteq \set{\alpha_0, \ldots, \alpha_{i - 1}}}$.
\end{definition}

We will need the following lemmas by Aamand et~al.~\cite{aamand2020}.
\begin{lemma}[\cite{aamand2020}]\label{lem:group-of-keys}
    Let $w \colon \Sigma^c \to \R_{\ge 0}$ be a weight function, then there exists an ordering $\set{\alpha_0, \ldots, \alpha_{r - 1}}$ of the position characters of $\Sigma^c$, such that for all $i \in [r]$,
    \[
        \sum_{x \in G_i} w(x) \le \left( \max_{x \in \Sigma^c} w(x) \right)^{1/c} \left( \sum_{x \in \Sigma^c} w(x) \right)^{1 - 1/c}
        \; .
    \]
\end{lemma}

\begin{lemma}\label{lem:weighted-sum-of-dependence}
    Let $k \in \N$ be a positive integer, and $w_0, \ldots, w_{k - 1} \colon \Sigma^{c} \to \R$ be weight functions.
    Then,
    \begin{align}\label{eq:weighted-sum-of-dependence}
        \sum_{\substack{x_{0}, \ldots, x_{k - 1} \in \Sigma^{c} \\ \bigxor_{i \in [k]} x_i = \emptyset}}
                \prod_{j \in [k]} w_i(x_i)
            \le \sqrt{\tfrac{k}{2}}^{k c} \prod_{i \in [k]} \sqrt{\sum_{x \in \Sigma^{c}} w_i(x)^2}
        \; .
    \end{align}
\end{lemma}

%% file: auxiliary/intro.tex

The goal of this chapter is to establish a series of technical lemmas concerning moments which will be crucial in the later part of the paper.
An important tool will be the function $\Psi_p \colon \R_+ \times \R_+ \to \R_+$ which gives a qualitative way of measuring how close the centered moments of sums of weighted Bernoulli variables resembles the central moments of the Poisson distribution.

\begin{definition}
    For $p \ge 2$ we define the function $\Psi_p \colon \R_+ \times \R_{+} \to \R_{+}$ as follows,
    \begin{align*}
        \Psi_p(M, \sigma^2) = \begin{cases}
            \left(\frac{\sigma^2}{p M^2}\right)^{1/p} M &\text{if $p < \log \frac{p M^2}{\sigma^2}$}\\
            \tfrac{1}{2}\sqrt{p}\sigma &\text{if $p < e^{2} \frac{\sigma^2}{M^2}$}\\
            \frac{p}{e \log \frac{p M^2}{\sigma^2}} M &\text{if $\max\!\set{\log \frac{p M^2}{\sigma^2}, e^{2} \frac{\sigma^2}{M^2}} \le p$}
        \end{cases}
        \; .
    \end{align*}
\end{definition}
\begin{remark}
    From the definition it is not clear that $\Psi_p$ is well-defined,
    but if $2 \le p < e^2\frac{\sigma^2}{M^2}$ then $\log \frac{p M^2}{\sigma^2} < 2$, hence at most one of the first two cases are satisfied at any given time.
    This shows that $\Psi_p$ is indeed well-defined.
\end{remark}

We show in \Cref{lem:psi-relation-poisson} that, up to constant, $\Psi_p(1, \lambda)$ is equal to the $p$-norm of a variable distributed as a centered Poisson variable with parameter $\lambda$.
This implies that, up to constant, $\Psi_p(M, \sigma^2)$ is equal to $M$ times the $p$-norm of a variable distributed as a centered Poisson variable with parameter $\frac{\sigma^2}{M^2}$.

When we later prove our concentration results for simple tabulation and mixed tabulation, we will need to bound $\pnorm{\Psi_p(X, Y)}{p}$ for random variables $X$ and $Y$.
Now to handle this we will develop a general lemma that bounds the moments of a function of random variables by the moments of the random variables.
This will be the focus of \Cref{sec:moments-from-tails}.

%% file: auxiliary/poisson-moments.tex

\subsection{Moments of Poisson Distributed Variables}

We start this section by proving a number of properties of the $\Psi_p$-function which we will use extensively.
Afterward, we will establish the connection between $\Psi_p$-function and the $p$-norm of Poisson distributed variables.

\begin{lemma}\label{lem:psi-properties}
    Let $p \ge 2$ then the function $\Psi_p$ satisfies the following properties:
    \begin{enumerate}
        \item For all positive reals $M > 0$ and $\sigma > 0$,
            \begin{align}
                \Psi_p(M, \sigma^2) = M \sup\setbuilder{\frac{p}{s} \left( \frac{\sigma^2}{p M^2} \right)^{1/s} }{2 \le s \le p}
                    \; . \label{eq:psi-relation}
            \end{align}
        \item For all positive reals $M > 0$ and $\sigma > 0$,
            \begin{align}
                \Psi_p(M, \sigma^2) \le \max\!\set{\tfrac{1}{2}\sqrt{p}\sigma, \tfrac{1}{2e} p M}
                    \; . \label{eq:psi-bernstein}
            \end{align}
        \item For all positive reals $M > 0$ and $\sigma > 0$,
            \begin{align}
                \tfrac{1}{2}\sqrt{p}\sigma \le \Psi_p(M, \sigma^2)
                    \; . \label{eq:psi-lower-bound}
            \end{align}
        \item For all positive reals $M > 0$ and $\sigma > 0$ with $e^{2} \frac{\sigma^2}{M^2} \le p$,
            \begin{align}
                \Psi_p(M, \sigma^2) \le \frac{p}{e \log \frac{p M^2}{\sigma^2}} M
                    \; . \label{eq:psi-far-out}
            \end{align}
        \item For all positive reals $M > 0$ and $\sigma > 0$ and all $\lambda \ge 1$,
            \begin{align}
                \Psi_p(\lambda M, \lambda \sigma^2) \le \lambda \Psi_p(M, \sigma^2)
                    \; . \label{eq:psi-growth}
            \end{align}
        \item For all positive reals $M > 0$ and $\sigma > 0$ and all $\lambda \ge 1$,
            \begin{align}
                \lambda \Psi_p(M, \sigma^2) \le \Psi_p(\lambda^2 M, \lambda^2 \sigma^2)
                    \; . \label{eq:psi-reverse-growth}
            \end{align}
        \item If $f : \R_+ \to \R_+$ is an increasing function, where $p \mapsto f(1/p)$ is log-convex
            and where there exists positive reals $K > 0, M > 0$, and $\sigma > 0$ such that
            $f(p) \le K \Psi_p(M, \sigma^2)$ for all even integers $p \ge 2$, then
            \begin{align}
                f(p) \le 2 K \Psi_p(M, \sigma^2)
                \; , \label{eq:psi-integers-to-reals}
            \end{align}
            for all reals $p \ge 2$.
    \end{enumerate}
\end{lemma}
\begin{proof}
    \item
    \paragraph*{Proof of \cref{eq:psi-relation}.}
    Let $\alpha = \frac{\sigma^2}{p M^2}$ and define the function $f(s) = \frac{1}{s} \alpha^{1/s}$.
    Taking the derivative we get that $f'(s) = \frac{1}{s^2}\left( -\frac{\log \alpha}{s} \alpha^{1/s} - \alpha^{1/s} \right)$.
    From this it is clear that $f$ is maximized at the point $s^* = \min\set{\max\!\set{2, \log 1/\alpha}, p}$ on the interval $[2, p]$.
    This implies that
    \begin{align*}
        M \sup\setbuilder{\frac{p}{s} \left( \frac{\sigma^2}{p M^2} \right)^{1/s} }{2 \le s \le p}
            &= M \sup\setbuilder{p f(s) }{2 \le s \le p}
            \\&= M p f(s^*)
            \\&= M p \begin{cases}
                \frac{1}{p} \left( \frac{\sigma^2}{p M^2} \right)^{1/p} &\text{if $p < \log \frac{p M^2}{\sigma^2}$} \\
                \frac{1}{2} \sqrt{\frac{\sigma^2}{p M^2}} &\text{if $\log \frac{p M^2}{\sigma^2} < 2$} \\
                \frac{1}{\log \frac{p M^2}{\sigma^2}} e^{-1} &\text{if $2 \le \log \frac{p M^2}{\sigma^2} \le p$}
            \end{cases}
            \\&= \begin{cases}
                \left(\frac{\sigma^2}{p M^2}\right)^{1/p} M &\text{if $p < \log \frac{p M^2}{\sigma^2}$}\\
                \frac{1}{2}\sqrt{p}\sigma &\text{if $p < e^{2} \frac{\sigma^2}{M^2}$}\\
                \frac{p}{e \log \frac{p M^2}{\sigma^2}} M &\text{if $\max\!\set{\log \frac{p M^2}{\sigma^2}, e^{2} \frac{\sigma^2}{M^2}} \le p$}
            \end{cases}
    \end{align*}
    which proves the claim.

    \paragraph*{Proof of \cref{eq:psi-bernstein}.}
    We note that if $p \le \log \frac{p M^2}{\sigma^2}$ then
    \[
        \left( \frac{\sigma^2}{p M^2} \right)^{1/p} M \le \frac{1}{e} M \le \frac{1}{2e} p M
        \; ,    
    \]
    and if $p > e^{2} \frac{\sigma^2}{M^2}$ then,
    \[
        \frac{p}{e \log \frac{p M^2}{\sigma^2}} M
            \le \frac{1}{2e} p M
        \; .
    \]
    This shows the upper bound.

    \paragraph*{Proof of \cref{eq:psi-lower-bound}.}
    The lower bound follows from \cref{eq:psi-relation} since,
    \begin{align*}
        \Psi_p(M, \sigma^2)
            = M \sup\setbuilder{\frac{p}{s} \left( \frac{\sigma^2}{p M^2} \right)^{1/s} }{2 \le s \le p}
            \ge M \frac{p}{2} \left( \frac{\sigma^2}{p M^2} \right)^{1/2}
            = \frac{1}{2} \sqrt{p} \sigma
        \; .
    \end{align*}

    \paragraph*{Proof of \cref{eq:psi-far-out}.}
    From \cref{eq:psi-relation} we know that $\Psi_p(M, \sigma^2) = M \sup\setbuilder{\frac{p}{s} \left( \frac{\sigma^2}{p M^2} \right)^{1/s} }{2 \le s \le p}$.
    We then get the upper bound,
    \begin{align*}
        \Psi_p(M, \sigma^2) 
            \le M \sup\setbuilder{\frac{p}{s} \left( \frac{\sigma^2}{p M^2} \right)^{1/s} }{2 \le s}
        \; .
    \end{align*}
    Now using the same method as in the proof of \cref{eq:psi-relation}, we get that the expression is maximized at $s^* = \frac{\sigma^2}{p M^2} \ge 2$
    and we get that
    \begin{align*}
        \Psi_p(M, \sigma^2) 
            \le \frac{p}{e \log \frac{p M^2}{\sigma^2}} M
        \; .
    \end{align*}

    \paragraph*{Proof of \cref{eq:psi-growth}.}
    We first notice that $\lambda\Psi_p(M, \sigma^2) = \Psi_p(\lambda M, \lambda^2 \sigma^2)$.
    Since $y \mapsto \Psi_p(\lambda M, y)$ is monotonically increasing then,
    \begin{align*}
        \Psi_p(\lambda M, \lambda \sigma^2)
            \le \Psi_p(\lambda M, \lambda^2 \sigma^2)
            = \lambda\Psi_p(M, \sigma^2)
        \; .
    \end{align*}

    \paragraph*{Proof of \cref{eq:psi-reverse-growth}.}
    We will again use that $\lambda \Psi_p(M, \sigma^2) = \Psi_p(\lambda M, \lambda^2 \sigma^2)$.
    This time we will use that $x \mapsto \Psi_p(x, \lambda^2 \sigma^2)$ is monotonically increasing,
    \begin{align*}
        \lambda \Psi_p(M, \sigma^2)
            = \Psi_p(\lambda M, \lambda^2 \sigma^2)
            \le \Psi_p(\lambda^2 M, \lambda^2 \sigma^2)
        \; .
    \end{align*}

    \paragraph*{Proof of \cref{eq:psi-integers-to-reals}.}
    Let $p \ge 2$ be a real and define let $2 \le q$ be the largest even integer with $q \le p$, that is, $q$ is the unique even integer satisfying that $q \le p < 2q$.
    Now let $\theta \ge [0, 1]$ be defined by the equation $\tfrac{1}{p} = \theta \tfrac{1}{q} + (1 - \theta) \tfrac{1}{2q}$.
    By the log-convexity of $p \mapsto f(1/p)$ we get that $f(p) \le f(q)^{\theta} f(2q)^{1 - \theta}$.
    We then have to consider two different cases.

    If $\log \frac{p M^2}{\sigma^2} < p$ then it is easy to check that $\Psi_{2p}(M, \sigma^2) \le 2 \Psi_{p}(M, \sigma^2)$,
    hence we get that
    \begin{align*}
        f(p)
            \le f(q)^{\theta} f(2q)^{1 - \theta} 
            \le f(2q)
            \le K \Psi_{2q}(M, \sigma^2)
            \le K \Psi_{2p}(M, \sigma^2)
            \le 2K \Psi_{p}(M, \sigma^2)
        \; .
    \end{align*}

    If $p \le \log \frac{p M^2}{\sigma^2}$ then we also know that $q \le \log \frac{q M^2}{\sigma^2}$
    and $\Psi_q(M, \sigma^2) = \left(\frac{\sigma^2}{q M^2}\right)^{1/q} M  \le \sqrt{2} \left(\frac{\sigma^2}{p M^2}\right)^{1/q} M$,
    where we have used that $p < 2q$ and $2 \le q$.
    Let $p < q'$ be defined by $q' = \log \frac{q'}{M^2}{\sigma^2}$.
    If $2q \le q'$ then we have that $\Psi_{2q}(M, \sigma^2) = \left(\frac{\sigma^2}{2q M^2}\right)^{1/2q} M \le \left(\frac{\sigma^2}{p M^2}\right)^{1/2q} M$
    and we get that
    \begin{align*}
        f(p)
            \le f(q)^{\theta} f(2q)^{1 - \theta} 
            \le \sqrt{2} K \left(\frac{\sigma^2}{q M^2}\right)^{\theta/q} \left(\frac{\sigma^2}{2q M^2}\right)^{(1 - \theta)/2q} M
            = \sqrt{2} K \Psi_q(M, \sigma^2)
        \; .
    \end{align*}
    If $q \le q' \le 2q$ then 
    \[
        \Psi_{2q}(M, \sigma^2)
            \le \Psi_{2q'}(M, \sigma^2)
            \le 2 \Psi_{q'}(M, \sigma^2)
            = 2 \left( \frac{\sigma^2}{q' M^2} \right)^{1/q'} M
            \le 2 \left( \frac{\sigma^2}{p M^2} \right)^{1/2q} M
        \; .
    \]  
    Combining these two facts give us that
    \begin{align*}
        f(p)
            \le f(q)^{\theta} f(2q)^{1 - \theta} 
            \le 2 K \left(\frac{\sigma^2}{q M^2}\right)^{\theta/q} \left(\frac{\sigma^2}{2q M^2}\right)^{(1 - \theta)/2q} M
            = 2 K \Psi_q(M, \sigma^2)
        \; .
    \end{align*}
    This finishes the proof of \Cref{lem:psi-properties}.
\end{proof}

We are now ready to establish the connection between the $\Psi_p$-function and the $p$-norms of Poisson distributed random variables.

\restateLemPsiRelationPoisson

For the proof, we need the following result by Latała~\cite{Latala1997} that gives a tight bound for $p$-norms of sums of independent and identically distributed symmetric variables.
\begin{lemma}[Latała~\cite{Latala1997}]\label{lem:Latala}
    If $(X_i)_{i \in [n]}$ are independent and identically distributed symmetric variables and $p \ge 2$ then,
    \begin{align*}
        \pnorm{\sum_{i \in [n]} X_i}{p}
            \le K_1 \sup\setbuilder{\frac{p}{s}\left( \frac{n}{p} \right)^{1/s} \pnorm{X_0}{s}}{\max\!\set{2, \tfrac{p}{n}} \le s \le p}
        \; ,
    \end{align*}
    and 
    \begin{align*}
        \pnorm{\sum_{i \in [n]} X_i}{p}
            \ge K_2 \sup\setbuilder{\frac{p}{s}\left( \frac{n}{p} \right)^{1/s} \pnorm{X_0}{s}}{\max\!\set{2, \tfrac{p}{n}} \le s \le p}
        \; .
    \end{align*}
    Here $K_1$ and $K_2$ are universal constants.
\end{lemma}

\begin{proof}[Proof of \Cref{lem:psi-relation-poisson}]
    We will use the standard fact that the Poisson distribution is the limit of a binomial distribution, with fixed mean $\lambda$, as the number of trials go to infinity.
    Let $\left(Y^{(n)}_i\right)_{i \in [n]}$ be independent Bernoulli variables with $\prb{Y^{(n)}_i = 1} = \tfrac{\lambda}{n}$ for $n \ge 1$.
    We then get that
    \begin{align*}
        \pnorm{X - \lambda}{p}
            = \lim_{n \to \infty} \pnorm{\sum_{i \in [n]} \left(Y^{(n)}_i - \tfrac{\lambda}{n} \right)}{p}
        \; .
    \end{align*}

    We let $(\eps_i)_{i \in \N}$ be independent Rademacher variables.
    We will argue that
    \begin{align}\label{eq:symmetrization}
        \tfrac{1}{2} \pnorm{\sum_{i \in [n]} \eps_i \left(Y^{(n)}_i - \tfrac{\lambda}{n} \right)}{p} 
            \le \pnorm{\sum_{i \in [n]} \left(Y^{(n)}_i - \tfrac{\lambda}{n} \right)}{p}
            \le 2 \pnorm{\sum_{i \in [n]} \eps_i \left(Y^{(n)}_i - \tfrac{\lambda}{n} \right)}{p}
        \; .
    \end{align}
    We defer the proof of \cref{eq:symmetrization} to the end.
    Using \cref{eq:symmetrization} it is enough to show that $\lim_{n \to \infty} \pnorm{\sum_{i \in [n]} \eps_i \left(Y^{(n)}_i - \tfrac{\lambda}{n} \right)}{p}$ is at most a constant away from $\Psi_p(1, \lambda)$.
    We use \Cref{lem:Latala} to get that
    \begin{align*}
        &\lim_{n \to \infty} \pnorm{\sum_{i \in [n]} \eps_i \left(Y^{(n)}_i - \tfrac{\lambda}{n} \right)}{p}
            \\&\qquad\qquad\le \lim_{n \to \infty} K_1 \sup\setbuilder{\frac{p}{s}\left( \frac{n}{p} \right)^{1/s} \pnorm{Y^{(n)}_0}{s}}{\max\!\set{2, \tfrac{p}{n}} \le s \le p}
            \\&\qquad\qquad= \lim_{n \to \infty} K_1 \sup\setbuilder{\frac{p}{s}\left( \frac{n}{p} \left(\tfrac{\lambda}{n}\left(1 - \tfrac{\lambda}{n}\right)^s + \left(1 - \tfrac{\lambda}{n}\right) \tfrac{\lambda^s}{n^s} \right) \right)^{1/s} }{\max\!\set{2, \tfrac{p}{n}} \le s \le p}
            \\&\qquad\qquad= \lim_{n \to \infty} K_1 \sup\setbuilder{\frac{p}{s}\left( \frac{\lambda}{p} \left(\left(1 - \tfrac{\lambda}{n}\right)^s + \left(1 - \tfrac{\lambda}{n}\right) \tfrac{\lambda^{s - 1}}{n^{s - 1}} \right) \right)^{1/s} }{\max\!\set{2, \tfrac{p}{n}} \le s \le p}
            \\&\qquad\qquad= K_1 \sup\setbuilder{\frac{p}{s}\left( \frac{\lambda}{p}  \right)^{1/s} }{2 \le s \le p}
            \\&\qquad\qquad= K_1 \Psi_p(1, \lambda)
        \; .
    \end{align*}
    The last equality follows by \cref{eq:psi-relation}.
    The proof of the lower bound is analogous.
    
    Now we just need to prove \cref{eq:symmetrization}.
    We first consider the lower bound.
    Fixing $(\eps_i)_{i \in \N}$ and using the triangle inequality we get that
    \begin{align*}
        &\pnormcond{\sum_{i \in [n]} \eps_i \left(Y^{(n)}_i - \tfrac{\lambda}{n} \right)}{p}{(\eps_i)_{i \in \N}}
            \\&\qquad\qquad= \pnormcond{\sum_{i \in [n]} [\eps_i = 1] \left(Y^{(n)}_i - \tfrac{\lambda}{n} \right) - \sum_{i \in [n]} [\eps_i = -1] \left(Y^{(n)}_i - \tfrac{\lambda}{n} \right)}{p}{(\eps_i)_{i \in \N}}
            \\&\qquad\qquad\le \pnormcond{\sum_{i \in [n]} [\eps_i = 1] \left(Y^{(n)}_i - \tfrac{\lambda}{n} \right)}{p}{(\eps_i)_{i \in \N}} + \pnormcond{\sum_{i \in [n]} [\eps_i = -1] \left(Y^{(n)}_i - \tfrac{\lambda}{n} \right)}{p}{(\eps_i)_{i \in \N}}
        \; .
    \end{align*}
    Now we use Jensen's inequality on each of the terms.
    \begin{align*}
        &\pnormcond{\sum_{i \in [n]} [\eps_i = 1] \left(Y^{(n)}_i - \tfrac{\lambda}{n} \right)}{p}{(\eps_i)_{i \in \N}}
            \\&\qquad\qquad= \pnormcond{\sum_{i \in [n]} [\eps_i = 1] \left(Y^{(n)}_i - \tfrac{\lambda}{n} \right) + \sum_{i \in [n]} [\eps_i = -1] \ep{Y^{(n)}_i - \tfrac{\lambda}{n}}}{p}{(\eps_i)_{i \in \N}}
            \\&\qquad\qquad\le \pnormcond{\sum_{i \in [n]} \left(Y^{(n)}_i - \tfrac{\lambda}{n} \right)}{p}{(\eps_i)_{i \in \N}}
        \; .
    \end{align*}
    Unfixing $(\eps_i)_{i \in \N}$ we get that
    \begin{align*}
        \pnorm{\sum_{i \in [n]} \eps_i \left(Y^{(n)}_i - \tfrac{\lambda}{n} \right)}{p} 
            \le 2\pnorm{\sum_{i \in [n]} \left(Y^{(n)}_i - \tfrac{\lambda}{n} \right)}{p}
        \; ,
    \end{align*}
    which establishes the lower bound of \cref{eq:symmetrization}.

    For the upper bound of \cref{eq:symmetrization} we first define $\left(Z^{(n)}_i\right)_{i \in [n]}$ to be independent copies of $\left(Y^{(n)}_i\right)_{i \in [n]}$.
    We then use Jensen's inequality to get that
    \begin{align*}
        \pnorm{\sum_{i \in [n]} \left(Y^{(n)}_i - \tfrac{\lambda}{n} \right)}{p}
            &= \pnorm{\sum_{i \in [n]} \left(\left(Y^{(n)}_i - \tfrac{\lambda}{n}\right) - \ep{Z^{(n)}_i - \tfrac{\lambda}{n}}\right)}{p}
            \\&\le \pnorm{\sum_{i \in [n]} \left( \left(Y^{(n)}_i - \tfrac{\lambda}{n}\right) - \left(Z^{(n)}_i - \tfrac{\lambda}{n}\right) \right)}{p}
        \; .
    \end{align*}
    We then note that due to independence $\left(Y^{(n)}_i - \tfrac{\lambda}{n}\right) - \left(Z^{(n)}_i - \tfrac{\lambda}{n}\right)$ is a symmetric variable, thus it has the same distribution as $\eps_i\left(\left(Y^{(n)}_i - \tfrac{\lambda}{n}\right) - \left(Z^{(n)}_i - \tfrac{\lambda}{n}\right)\right)$.
    We use this and the triangle inequality to finish the upper bound,
    \begin{align*}
        \pnorm{\sum_{i \in [n]} \left( \left(Y^{(n)}_i - \tfrac{\lambda}{n}\right) - \left(Z^{(n)}_i - \tfrac{\lambda}{n}\right) \right)}{p}
            &= \pnorm{\sum_{i \in [n]} \eps_i\left( \left(Y^{(n)}_i - \tfrac{\lambda}{n}\right) - \left(Z^{(n)}_i - \tfrac{\lambda}{n}\right) \right)}{p}
            \\&\le 2 \pnorm{\sum_{i \in [n]} \eps_i \left(Y^{(n)}_i - \tfrac{\lambda}{n} \right)}{p}
        \; .
    \end{align*}
    This finishes the proof of \Cref{lem:psi-relation-poisson}.
\end{proof}

%% file: auxiliary/general-moments.tex

\subsection{Moments of General Random Variables}

We start by proving a lemma that bounds the moments of weighted sums of independent and identically distributed variables.
The lemma is similar to \Cref{lem:Latala} by Latała~\cite{Latala1997} but it is not tight for all distributions.

\begin{lemma}\label{lem:gen-moment-sym}
    Let $(X_i)_{i \in [n]}$ be independent and identically distributed symmetric random variables, and let $(a_i)_{i \in [n]}$ be a sequence of integers.
    If $p \ge 2$ is an even integer then,
    \begin{align*}
        \pnorm{\sum_{i \in [n]} a_{i} X_i}{p}
            \le K \sup\setbuilder{\frac{p}{s} \left( \frac{\sum_{i \in [n]} a_i^s}{p} \right)^{1/s} \pnorm{X}{s}}{2 \le s \le p}
        \; ,
    \end{align*}
    where $K \le 4 e$ is a universal constant.
\end{lemma}

In the proof, we will need the following folklore result.
We provide a proof of the result for completeness.
\begin{lemma}\label{lem:folklore}
    Let $n \ge 0$ be a positive integer and let $a_1, \ldots, a_{n} \in \R$ be real numbers.
    If $x \in \R$ is a real number satisfying,
    \[
        x^n \le \sum_{i = 1}^{n} a_{i} x^{n - i} \; .
    \]
    Then,
    \[
        x \le 2 \max_{1 \le i \le n} \abs{a_i}^{1/i} \; .
    \]
\end{lemma}
\begin{proof}
    The proof follows by noticing that
    \begin{align*}
        x^n \le \sum_{i = 1}^{n} a_{i} x^{n - i}
            \le \sum_{i = 1}^n \abs{a_i} \abs{x}^{n - i}
        \; .
    \end{align*}
    We will now show that $\sum_{i = 1}^n \abs{a_i} \abs{x}^{n - i} \le \max_{i = 1}^{n} 2^i \abs{a_i} \abs{x}^{n - i}$ by induction on $n$.
    The result is clearly true for $n = 1$.
    Now assume that the result holds for integers less than $n$ then,
    \begin{align*}
        \sum_{i = 1}^n \abs{a_i} \abs{x}^{n - i}
            \le \max\!\set{2 a_n, 2 \sum_{i = 1}^{n - 1} \abs{a_i} \abs{x}^{n - i}}
            \le \max_{i = 1}^{n} 2^i \abs{a_i} \abs{x}^{n - i}
        \; .
    \end{align*}
    We then have that $x^n \le \max_{i = 1}^{n} 2^i \abs{a_i} \abs{x}^{n - i}$.
    This is equivalent with the statement that there exists an integer $1 \le i \le n$ with $x^n \le 2^i \abs{a_i} \abs{x}^{n - i}$.
    This implies that there exists an integer $1 \le i \le n$ with $x \le 2 \abs{a_i}^{1/i}$.
    This is again equivalent with $x \le \max_{i = 1}^n \abs{a_i}^{1/i}$ which is what we wanted to prove.
\end{proof}

We now turn to the proof of \Cref{lem:gen-moment-sym}.
\begin{proof}[Prood of \Cref{lem:gen-moment-sym}]
    Since $p$ is an even integer then,
    \begin{align*}
        \pnorm{\sum_{i \in [n]} a_{i} X_i}{p}^p
            &= \sum_{i \in [n]} \sum_{s = 1}^p \binom{p - 1}{s - 1} \ep{(a_i X_i)^s \left(\sum_{j \in [n]\setminus\set{i}} a_j X_j \right)^{p - s}}
            \\&= \sum_{i \in [n]} \sum_{s = 1}^p \binom{p - 1}{s - 1} \ep{(a_i X_i)^s} \ep{\left(\sum_{j \in [n]\setminus\set{i}} a_j X_j \right)^{p - s}}
        \; .
    \end{align*}
    The first step follows by noticing that when we expand $(\sum_{i \in [n]} a_{i} X_i)^p$ then for each term the first factor will give an $i \in [n]$.
    Now if $i$ has multiplicity $s \ge 1$ in the term then there are $\binom{p - 1}{s - 1}$ ways to choose the other factors for $i$.

    We note that since the variables are symmetric then $\ep{(a_i X_i)^s} = 0$ for $s$ odd.
    So in the following, we assume that $s$ is even, which implies that $p - s$ is even.
    Now we use Jensen's inequality to obtain,
    \begin{align*}
        \ep{\left(\sum_{j \in [n]\setminus\set{i}} a_j X_j \right)^{p - s}}
            &= \ep{\left(\sum_{j \in [n]\setminus\set{i}} a_j X_j +   a_i \ep{X_i} \right)^{p - s}}
            \le \ep{\left(\sum_{j \in [n]} a_j X_j \right)^{p - s}}
        \; .
    \end{align*}
    Another usage of Jensen's inequality gives us that
    \begin{align*}
        \ep{\left(\sum_{j \in [n]} a_j X_j \right)^{p - s}}
            = \pnorm{\sum_{j \in [n]} a_j X_j}{p - s}^{p - s}
            \le \pnorm{\sum_{j \in [n]} a_j X_j}{p}^{p - s}
        \; .
    \end{align*}
    Combining these we get that
    \begin{align*}
        \pnorm{\sum_{i \in [n]} a_{i} X_i}{p}^p
            &\le \sum_{i \in [n]} \sum_{s = 2}^p \binom{p - 1}{s - 1} \ep{(a_i X_i)^s} \pnorm{\sum_{j \in [n]} a_j X_j}{p}^{p - s}
            \\&= \sum_{s = 2}^p \binom{p - 1}{s - 1} \sum_{i \in [n]} a_i^s \ep{X_i^s} \pnorm{\sum_{j \in [n]} a_j X_j}{p}^{p - s}
            \\&= \sum_{s = 2}^p \binom{p - 1}{s - 1} \ep{X^s} \pnorm{\sum_{j \in [n]} a_j X_j}{p}^{p - s} \sum_{i \in [n]} a_i^s
        \; .
    \end{align*}
    Now using \Cref{lem:folklore} we get that
    \begin{align*}
        \pnorm{\sum_{i \in [n]} a_{i} X_i}{p}
            &\le \sup\setbuilder{ 2 \binom{p - 1}{s - 1}^{1/s} \left(\sum_{i \in [n]} a_i^s\right)^{1/s} \pnorm{X}{s}}{2 \le s \le p}
        \; .
    \end{align*}
    Now using Stirling's approximation, we get that
    $\binom{p - 1}{s - 1} = \binom{p}{s} \frac{s}{p} \le \left( \frac{e p}{s} \right)^{s} \frac{s}{p}$.
    Plugging this estimate into our equation gives us that
    \begin{align*}
        \pnorm{\sum_{i \in [n]} a_{i} X_i}{p}
            &\le \sup\setbuilder{ 2 \frac{e p}{s} \frac{s^{1/s}}{p^{1/s}} \left(\sum_{i \in [n]} a_i^s\right)^{1/s} \pnorm{X}{s}}{2 \le s \le p}
            \\&\le 4 e \sup\setbuilder{ \frac{p}{s} \left(\frac{\sum_{i \in [n]} a_i^s}{p}\right)^{1/s} \pnorm{X}{s}}{2 \le s \le p}
        \; .
    \end{align*}
\end{proof}

We will not be using the result directly instead we will use the following corollary where we further simplify the expression by bounding only in terms of the largest weight and the Euclidean norm.

\begin{corollary}\label{cor:max-moments-sym}
    Let $(X_i)_{i \in [n]}$ be independent and identically distributed symmetric random variables, and let $(a_i)_{i \in [n]}$ be a sequence of reals.
    If $p \ge 2$ is an even integer then,
    \begin{align*}
        \pnorm{\sum_{i \in [n]} a_{i} X_i}{p}
            \le K \max_{i \in [n]} \abs{a_{i}} \sup\setbuilder{\frac{p}{s} \left( \frac{\sum_{i \in [n]} a_i^2}{p\max_{i \in [n]} \abs{a_i}^2} \right)^{1/s} \pnorm{X}{s}}{2 \le s \le p}
        \; ,
    \end{align*}
    where $K = 4 e$ is a universal constant.
\end{corollary}
\begin{proof}
    This follows from \Cref{lem:gen-moment-sym} and the fact that $a_i^s \le a_i^2 \max_{i \in [n]} \abs{a_i}^{s - 2}$ for all $i \in [n], s \ge 2$.
    \begin{align*}
        \pnorm{\sum_{i \in [n]} a_{i} X_i}{p}
            &\le K \sup\setbuilder{\frac{p}{s} \left( \frac{\sum_{i \in [n]} a_i^s}{p} \right)^{1/s} \pnorm{X}{s}}{2 \le s \le p}
            \\&\le K \sup\setbuilder{\frac{p}{s} \left( \frac{\sum_{i \in [n]} a_i^2 \max_{i \in [n]} \abs{a_i}^{s - 2}}{p} \right)^{1/s} \pnorm{X}{s}}{2 \le s \le p}
            \\&\le K \max_{i \in [n]} \abs{a_{i}} \sup\setbuilder{\frac{p}{s} \left( \frac{\sum_{i \in [n]} a_i^2}{p\max_{i \in [n]} \abs{a_i}^2} \right)^{1/s} \pnorm{X}{s}}{2 \le s \le p}
        \; .
    \end{align*}
\end{proof}

We will now use \Cref{cor:max-moments-sym} to bound the sum of different types of random variables with $\Psi_p$-function.
We start by looking at Bernoulli-Rademacher variables.

\begin{lemma}\label{lem:bernoulli-moments}
    Let $(X_i)_{i \in [n]}$ be independent Bernoulli-Rademacher variables with parameter $\alpha$, that is, $\prb{X_i = 1} = \prb{X_i = -1} = \tfrac{1}{2} - \prb{X_i = 0} = \tfrac{\alpha}{2}$,
    and let $(a_i)_{i \in [n]}$ be a sequence of reals.

    If $p \ge 2$ is an even integer then,
    \begin{align*}
        \pnorm{\sum_{i \in [n]} a_{i} X_i}{p}
            \le 4e \Psi_p\left(\max_{i \in [n]} \abs{a_i}, \alpha \sum_{i \in [n]} a_i^2 \right)
        \; .
    \end{align*}

    And if $p \ge 2$ is a real number then,
    \begin{align*}
        \pnorm{\sum_{i \in [n]} a_{i} X_i}{p}
            \le 8 e \Psi_p\left(\max_{i \in [n]} \abs{a_i}, \alpha \sum_{i \in [n]} a_i^2 \right)
        \; .
    \end{align*}
\end{lemma}
\begin{proof}
    We note that $\pnorm{X}{s} = \alpha^{1/s}$ for all $s \ge 2$.
    Let $p \ge 2$ be an even integer then using \Cref{cor:max-moments-sym} we then get that
    \begin{align*}
        \pnorm{\sum_{i \in [n]} a_{i} X_i}{p}
            &\le 4 e \max_{i \in [n]} \abs{a_i} \sup\setbuilder{\frac{p}{s} \left( \frac{\sum_{i \in [n]} a_i^2}{p \max_{i \in [n]} \abs{a_i}^2} \right)^{1/s} \alpha^{1/s} }{2 \le s \le p}
            \\&= 4 e \max_{i \in [n]} \abs{a_i} \sup\setbuilder{\frac{p}{s} \left( \frac{\alpha \sum_{i \in [n]} a_i^2}{p \max_{i \in [n]} \abs{a_i}^2} \right)^{1/s} }{2 \le s \le p}
        \; .
    \end{align*}
    Now \cref{eq:psi-relation} proves the first claim.
    By Hölder's inequality we have that $p \mapsto \pnorm{\sum_{i \in [n]} a_{i} X_i}{1/p}$ is log-convex
    and Jensen's inequality implies that $p \mapsto \pnorm{\sum_{i \in [n]} a_{i} X_i}{p}$ is increasing,
    thus \cref{eq:psi-integers-to-reals} proves the second claim.
\end{proof}

We are now almost ready to prove \Cref{thm:sampling-moments} for fully random hash functions.
This will be a principal lemma in the sequel when we prove concentration results for tabulation hashing.
But first, we need to prove a symmetrization lemma for fully random functions.
\begin{lemma}\label{lem:sampling-symmetrize}
    Let $h \colon U \to [m]$ be a uniformly random function, let $v \colon U \times [m] \to \R$ be a fixed value function, and assume that $\sum_{j \in [m]} v(x, j) = 0$ for all keys $x \in U$.
    Let $\eps \colon U \to \set{-1, 1}$ be a uniformly random sign function.
    Define the random variable $X_v = \sum_{x \in U} v(x, h(x))$.
    Then for all $p \ge 2$,
    \begin{align*}
        2^{-1}\pnorm{\sum_{x \in U} \eps(x) v(x, h(x))}{p}
            \le \pnorm{\sum_{x \in U} v(x, h(x))}{p}
            \le 2\pnorm{\sum_{x \in U} \eps(x) v(x, h(x))}{p}
        \; .
    \end{align*}
\end{lemma}
\begin{proof}
    We first consider the lower bound.
    Fixing $\eps$ and using the triangle inequality we get that
    \begin{align*}
        \pnormcond{\sum_{x \in U} \eps(x) v(x, h(x))}{p}{\eps}
            &= \pnormcond{\sum_{x \in U} [\eps(x) = 1] v(x, h(x)) - \sum_{x \in U} [\eps(x) = -1] v(x, h(x))}{p}{\eps}
            \\&\le \pnormcond{\sum_{x \in U} [\eps(x) = 1] v(x, h(x))}{p}{\eps} + \pnormcond{\sum_{x \in U} [\eps(x) = -1] v(x, h(x))}{p}{\eps}
        \; .
    \end{align*}
    Now we use Jensen's inequality on each of the terms.
    \begin{align*}
        \pnormcond{\sum_{x \in U} [\eps(x) = 1] v(x, h(x))}{p}{\eps}
            &= \pnormcond{\sum_{x \in U} [\eps(x) = 1] v(x, h(x)) + \sum_{x \in U} [\eps(x) = 1] \ep{v(x, h(x))}}{p}{\eps}
            \\&\le \pnormcond{\sum_{x \in U} v(x, h(x))}{p}{\eps}
        \; .
    \end{align*}
    Unfixing $\eps$ we get that
    \begin{align*}
        \pnorm{\sum_{x \in U} \eps(x) v(x, h(x))}{p} 
            \le 2\pnorm{\sum_{x \in U} v(x, h(x))}{p}
        \; ,
    \end{align*}
    which establishes the lower bound.

    For the upper bound, we first define $h' \colon U \to [m]$ to be an independent copy of $h$.
    We then use Jensen's inequality to get that
    \begin{align*}
        \pnorm{\sum_{x \in U} v(x, h(x))}{p}
            &= \pnorm{\sum_{x \in U} v(x, h(x)) - \sum_{x \in U} \ep{v(x, h'(x))}}{p}
            \\&\le \pnorm{\sum_{x \in U} \left( v(x, h(x)) - v(x, h'(x)) \right)}{p}
        \; .
    \end{align*}
    We then note that due to the independence $v(x, h(x)) - v(x, h'(x))$ is a symmetric variable, thus it has the same distribution as $\eps(x)\left(v(x, h(x)) - v(x, h'(x))\right)$.
    We use this and the triangle inequality to finish the upper bound,
    \begin{align*}
        \pnorm{\sum_{x \in U} \left( v(x, h(x)) - v(x, h'(x)) \right)}{p}
            &= \pnorm{\sum_{x \in U} \eps(x)\left( v(x, h(x)) - v(x, h'(x)) \right)}{p}
            \\&\le 2 \pnorm{\sum_{x \in U} \eps(x) v(x, h(x))}{p}
        \; ,
    \end{align*}
    which establishes the upper bound.
\end{proof}

\restateThmSamplingMoments
\begin{proof}
    We start by using \Cref{lem:sampling-symmetrize} to get that
    \begin{align*}
        \pnorm{\sum_{x \in U} v(x, h(x))}{p}
            \le 2\pnorm{\sum_{x \in U} \eps(x) v(x, h(x))}{p}
        \; .
    \end{align*}

    Let $(Y^{(j)}_x)_{x \in U, j \in [m]}$ be independent Bernoulli-Rademacher variables with parameter $\tfrac{1}{m}$.
    The idea of the proof is to show that for $p \ge 2$ then,
    \begin{align*}
        \pnorm{\sum_{x \in U} \eps(x) v(x, h(x))}{p}   
            \le \pnorm{\sum_{x \in U, j \in [m]} v(x, j)Y^{(j)}_x}{p}
        \; .     
    \end{align*}
    This is nontrivial to do for general $p$.
    Instead, we will focus on $p$ being an even integer.

    Let $q \ge 2$ be an even integer. Then for all $x \in U$ we have that
    \begin{align*}
        \pnorm{\eps(x) v(x, h(x))}{q}
            = \left( \frac{\sum_{j \in [m]} v(x, j)^q}{m} \right)^{1/q}
        \; .
    \end{align*}
    But it is easy to check that
    \begin{align*}
        \pnorm{\sum_{j \in [m]} v(x, j)Y^{(j)}_x}{q}
            \ge \left( \frac{\sum_{j \in [m]} v(x, j)^q}{m} \right)^{1/q}
            \; .
    \end{align*}
    We can now show what we want,
    \begin{align*}
        \ep{\left(\sum_{x \in U} \eps(x) v(x, h(x))\right)^p}
            &= \sum_{\sum_{x \in U} q_x = p} \binom{p}{(q_x)_{x \in U}} \prod_{x \in U} \ep{\left(\eps(x) v(x, h(x))\right)^{q_x}}
            \\&= \sum_{\substack{\sum_{x \in U} q_x = p \\ \forall x \in U : \text{$q_x$ is even}}} \binom{p}{(q_x)_{x \in U}} \prod_{x \in U} \ep{\left(\eps(x) v(x, h(x))\right)^{q_x}}
            \\&\le \sum_{\substack{\sum_{x \in U} q_x = p \\ \forall x \in U : \text{$q_x$ is even}}} \binom{p}{(q_x)_{x \in U}} \prod_{x \in U} \ep{\left(\sum_{j \in [m]} v(x, j)Y^{(j)}_x\right)^{q_x}}
            \\&= \sum_{\sum_{x \in U} q_x = p} \binom{p}{(q_x)_{x \in U}} \prod_{x \in U} \ep{\left(\sum_{j \in [m]} v(x, j)Y^{(j)}_x\right)^{q_x}}
            \\&= \ep{\left(\sum_{x \in U, j \in [m]} v(x, j)Y^{(j)}_x\right)^{p}}
    \end{align*}
    Now using \Cref{lem:bernoulli-moments} we get that
    \begin{align*}
        \pnorm{\sum_{x \in U} \eps(x) v(x, h(x))}{p}
            &\le \pnorm{\sum_{x \in U, j \in [m]} v(x, j)Y^{(j)}_x}{p}
            \\&\le 4e \Psi_p\left(\max_{x \in U, j \in [m]} \abs{v(x, j)},  \frac{\sum_{x \in U, j \in [m]} v(x, j)^2}{m} \right)
        \; .
    \end{align*}
    for all even integers $p \ge 2$.
    Now we use \cref{eq:psi-integers-to-reals} as in the proof of \Cref{lem:bernoulli-moments} which proves the second claim.
\end{proof}

We also need the standard fact that a sum of independent sub-Gaussian is also sub-Gaussian.
We include a proof for completeness.

\begin{lemma}\label{lem:sum-of-sub-guassian}
    Let $(X_i)_{i \in [n]}$ be a sequence of independent symmetric random variables.
    Let $p \ge 2$ and assume that there exists a sequence of real numbers $(a_i)_{i \in [n]}$ such that for all even integers $2 \le q \le p$ and all $i \in [n]$ it holds that
    \[
        \pnorm{X_i}{q} \le \sqrt{q} a_i \; .
    \]
    Then the sum of the random variables satisfies,
    \[
        \pnorm{\sum_{i \in [n]} X_i}{p}
            \le \sqrt{p} \sqrt{2e \sum_{i \in [n]} a_i^2}
        \; .
    \]
\end{lemma}
\begin{proof}
    The main idea of the proof is to compare the random variables $(X_i)_{i \in [n]}$ with a sequence of independent Gaussian $(g_i)_{i \in [n]}$, and then exploit that the sum of Gaussian variables is a Gaussian variable.
    We will use the standard fact that for all even integers $2 \le q$, Gaussian variables satisfies $\pnorm{g_i}{q} = \left( (q - 1)!! \right)^{1/q}$.
    A simple lower bound for this follows by using Stirling's approximation,
    \[
        \left( (q - 1)!! \right)^{1/q}
            \ge \left( \left(q/2\right)! \right)^{1/q}
            \ge \left( \left(\frac{q}{2e}\right)^{q/2} \right)^{1/q}
            = \sqrt{\frac{q}{2e}}
        \; .
    \]
    For an upper bound we note that by the AM-GM inequality we have that $(q - 2i - 1)(2i + 1) \le \left( \tfrac{q}{2} \right)^2$ so we get that
    \[
        \left( (q - 1)!! \right)^{1/q}
            \le \left( \left(\frac{q}{2}\right)^{q/2} \right)^{1/q}
            = \sqrt{\frac{q}{2}}
        \; .
    \]
    Now the lower bound gives us the estimate,
    \[
        \pnorm{X_i}{q}
            \le \sqrt{q} a_i
            \le \sqrt{2e} \pnorm{g_i}{q}
        \; .
    \]

    We start by proving the case where $p \ge 2$ is an even integer.
    We then get that
    \begin{align*}
        \ep{\left(\sum_{i \in [n]} X_i \right)^p}
            &= \sum_{\sum_{i \in [n]} q_i = p} \binom{p}{(q_i)_{i \in [n]}} \prod_{i \in [n]} \ep{\left(X_i\right)^{q_i}}
            \\&= \sum_{\substack{\sum_{i \in [n]} q_i = p \\ \forall i \in [n] : \text{$q_i$ is even}}} \binom{p}{(q_i)_{i \in [n]}} \prod_{i \in [n]} \ep{\left( X_i \right)^{q_i}}
            \\&\le \sum_{\substack{\sum_{i \in [n]} q_i = p \\ \forall i \in [n] : \text{$q_i$ is even}}} \binom{p}{(q_i)_{i \in [n]}} \prod_{i \in [n]} \ep{\left(\sqrt{2e} a_i g_i\right)^{q_i}}
            \\&= \sum_{\sum_{i \in [n]} q_i = p} \binom{p}{(q_i)_{i \in [n]}} \prod_{i \in [n]} \ep{\left(\sqrt{2e} a_i g_i\right)^{q_i}}
            \\&= \ep{\left(\sum_{i \in [n]} \sqrt{2e} a_i g_i \right)^{p}}
        \; .
    \end{align*}
    Now we use that if $g$, $g'$, and $g''$ are independent standard Gaussian variables then $a g + b g'$ is distributed as $\sqrt{a^2 + b^2}g''$.
    This give us that
    \begin{align*}
        \pnorm{\sum_{i \in [n]} X_i}{p}
            \le \pnorm{\sum_{i \in [n]} \sqrt{2e} a_i g_i}{p}
            = \pnorm{\sqrt{\sum_{i \in [n]} 2e a_i^2} g}{p}
            \le \sqrt{p} \sqrt{e \sum_{i \in [n]} a_i^2}
        \; . 
    \end{align*}
    If $p \ge 2$ is not an even integer then let $p' \ge p$ be the smallest even integer larger than $p$.
    We note that $p' \le 2p$ and using Jensen's inequality we get that
    \begin{align*}
        \pnorm{\sum_{i \in [n]} X_i}{p}
            \le \pnorm{\sum_{i \in [n]} X_i}{p'}
            \le \sqrt{p} \sqrt{e\sum_{i \in [n]} a_i^2}
            \le \sqrt{p} \sqrt{2e\sum_{i \in [n]} a_i^2}
        \; .
    \end{align*}
    This finishes the proof.
\end{proof}

We can use this lemma to prove another useful bound for uniformly random functions. 
\begin{lemma}\label{lem:sampling-hoeffding}
    Let $h \colon U \to [m]$ be a uniformly random function, let $\eps \colon U \to \set{-1, 1}$ be a uniformly random sign function, and let $v : U \times [m] \to \R$ be a fixed value function.
    Then for all $p \ge 2$,
    \begin{align*}
        \pnorm{\sum_{x \in U} \eps(s) v(x, h(x))}{p}
            \le L \sqrt{\frac{p}{\log\!\left( \frac{e^2 m \sum_{x \in \Sigma^{c}} \norm{v[x]}{\infty}^2}{\sum_{x \in \Sigma^{c}} \norm{v[x]}{2}^2} \right)} \sum_{x \in \Sigma^{c}} \norm{v[x]}{\infty}^2} 
        \; ,
    \end{align*}
    where $L \le e$ is a universal constant.
\end{lemma}

Before we prove \Cref{lem:sampling-hoeffding} we need the following technical lemma.

\begin{lemma}\label{lem:sum-concave}
    Let $(a_i)_{i \in [n]}$ and $(b_i)_{i \in [n]}$ be two sequences of positive integers.
    If $\frac{a_i}{b_i} \ge 1$ for all $i \in [n]$ then,
    \begin{align}\label{eq:sum-concave}
        \sum_{i \in [n]} \frac{a_i}{\log\!\left( \frac{e^2 a_i}{b_i} \right)}
            \le \frac{\sum_{i \in [n]} a_i}{\log\!\left( \frac{e^2 \sum_{i \in [n]} a_i}{\sum_{i \in [n]} b_i} \right)}
        \; .
    \end{align}
\end{lemma}
\begin{proof}
    We define the sequence $(r_i)_{i \in [n]}$ by $r_i = \frac{a_i}{b_i}$ for all $i \in [n]$, define the random variable $R$ by $\prb{R = r_i} = \frac{b_i}{\sum_{j \in [n]} b_i}$, and the function $f \colon \R_+ \to \R_+$ by $f(r) = \frac{r}{\log\!\left(e^2 r \right)}$.
    Now we note that
    \begin{align*}
        \sum_{i \in [n]} \frac{a_i}{\log\!\left( \frac{e^2 a_i}{b_i} \right)}
            &= \sum_{i \in [n]} b_i r_i{\log\!\left( e^2 r_i \right)}
            = \sum_{i \in [n]} b_i f(r_i)
            = \left(\sum_{j \in [n]} b_i\right) \ep{f(R)}
        \; , \\
        \frac{\sum_{i \in [n]} a_i}{\log\!\left( \frac{e^2 \sum_{i \in [n]} a_i}{\sum_{i \in [n]} b_i} \right)}
            &= \frac{\sum_{i \in [n]} b_i r_i}{\log\!\left( \frac{e^2 \sum_{i \in [n]} b_i r_i}{\sum_{i \in [n]} b_i} \right)}
            = \left(\sum_{j \in [n]} b_i\right) f(\ep{R})
        \; .
    \end{align*}
    Thus we get that \cref{eq:sum-concave} is equivalent with showing that $\ep{f(R)} \le f(\ep{R})$.
    It easy to check that $f$ is concave on the interval $[1, \infty)$ and since $R \ge \min_{i \in [n]} r_i = \min_{i \in [n]} \frac{a_i}{b_i} \ge 1$, Jensen's inequality implies the result.
\end{proof}

\begin{proof}[Proof of \Cref{lem:sampling-hoeffding}]
    Let $x \in U$ be fixed and consider $2 \le q \le p$.
    We then have that
    \begin{align*}
        \pnorm{\eps(x)v(x, h(x))}{q}
            &= \left( \frac{\sum_{j \in [m]} \abs{v(x, h(x))}^q}{m} \right)^{1/q}
            \\&\le \left( \frac{\sum_{j \in [m]} v(x, h(x))^2}{m \max_{j \in [m]} v(x, h(x))^2} \right)^{1/q} \max_{j \in [m]} \abs{v(x, h(x))}
            \\&= \left( \frac{\norm{v[x]}{2}^2}{m \norm{v[x]}{\infty}^2} \right)^{1/q} \norm{v[x]}{\infty}
        \; .
    \end{align*}
    Now a simple estimate give us that $y^{1/q} \le e (y/e^2)^{1/q} \le e \sqrt{\frac{q}{2e \log \tfrac{e^2}{y}}} = \sqrt{\frac{e q}{2 \log \tfrac{e^2}{y}}}$ for all $y \le 1$ and all $q \ge 2$.
    Clearly, $\frac{\norm{v[x]}{2}^2}{m \norm{v[x]}{\infty}^2} \le 1$, hence we get that
    \begin{align*}
        \pnorm{\eps(x)v(x, h(x))}{q}
            \le \sqrt{\frac{e q}{2 \log\!\left( \frac{e^2 m \norm{v[x]}{\infty}^2}{\norm{v[x]}{2}^2} \right) }} \norm{v[x]}{\infty}
        \; .
    \end{align*}
    This shows that $\eps(x)v(x, h(x))$ is sub-Gaussian hence we can use \Cref{lem:sum-of-sub-guassian} to get that
    \begin{align*}
        \pnorm{\sum_{x \in U} \eps(s) v(x, h(x))}{p}
            \le e \sqrt{p} \sqrt{\sum_{x \in U} \frac{ \norm{v[x]}{\infty}^2}{\log\!\left( \frac{e^2 m \norm{v[x]}{\infty}^2}{\norm{v[x]}{2}^2} \right)} }
        \; .
    \end{align*}
    Now an application of \Cref{lem:sum-concave} finishes the proof.
\end{proof}

We end the section by bounding the simple case of weighted sums of Rademacher variables.
We will need the lemma later and it is known as Khintchine's inequality.
For completeness we include a proof the lemma.

\begin{lemma}[Khintchine's inequality]\label{cor:sum-of-Rademacher}
    Let $(\eps_i)_{i \in [n]}$ be a sequence of independent Rademacher variables, and let $(a_i)_{i \in [n]}$ be a sequence of real numbers.
    For all $p \ge 2$ we have that
    \[
        \pnorm{\sum_{i \in [n]} a_i \eps_i}{p}
            \le \sqrt{p} \sqrt{e \sum_{i \in [n]} a_i^2}
        \; .
    \]
\end{lemma}
\begin{proof}
    We note that for all $i \in [n]$ and all $q \ge 2$ we have that $\pnorm{a_i \eps_i}{q} = \abs{a_i} \le \tfrac{\sqrt{q}}{\sqrt{2}} \abs{a_i}$.
    We now use \Cref{lem:sum-of-sub-guassian} to get that
    \[
        \pnorm{\sum_{i \in [n]} a_i \eps_i}{p}
            \le \sqrt{p} \sqrt{2e \sum_{i \in [n]} \tfrac{a_i^2}{2}}
            = \sqrt{p} \sqrt{e \sum_{i \in [n]} a_i^2}
        \; .
    \]
    This finishes the proof.
\end{proof}

%% file: auxiliary/tails-to-moments.tex





\subsection{Moments of Functions of Random Variables}\label{sec:moments-from-tails}
The goal of this section is to prove \Cref{lem:fn-moment}.
\restateLemFnMoment
\begin{proof}
    We define $\lambda = \max_{i \in [n]} \left( \frac{X_i}{t_i} \right)^{1/\alpha_i}$ and note that $X_i \le \lambda^{\alpha_i} t_i$ for all $i \in [n]$.
    Since $f$ is an increasing function then $f(X_0, \ldots, X_{n - 1}) \le f(\lambda^{\alpha_0} t_0, \ldots, \lambda^{\alpha_{n - 1}} t_{n - 1})$.
    We can then use the condition on $f$ to get that
    \begin{align*}
        \pnorm{f(X_0, \ldots, X_{n - 1})}{p}
            \le \pnorm{f(\lambda^{\alpha_0} t_0, \ldots, \lambda^{\alpha_{n - 1}} t_{n - 1})}{p}
            \le \pnorm{\lambda}{p} f(t_0, \ldots, t_{n - 1})
        \; .
    \end{align*}
    Now we just need to prove that $\pnorm{\lambda}{p} \le n^{1/p} \max_{i \in [n]} \left( \frac{\pnorm{X_i}{p/\alpha_i}}{t_i} \right)^{1/\alpha_i}$.
    We note that $\lambda = \max_{i \in [n]} \left( \frac{X_i}{t_i} \right)^{1/\alpha_i} \le \left(\sum_{i \in [n]} \left( \frac{X_i}{t_i} \right)^{p/\alpha_i} \right)^{1/p}$.
    Hence we get that
    \begin{align*}
        \pnorm{\lambda}{p}
            &\le \left(\sum_{i \in [n]} \ep{\left( \frac{X_i}{t_i} \right)^{p/\alpha_i}} \right)^{1/p}
            \\&\le \left(n \max_{i \in [n]} \ep{\left( \frac{X_i}{t_i} \right)^{p/\alpha_i}} \right)^{1/p}
            \\&= n^{1/p} \max_{i \in [n]} \left( \frac{\pnorm{X_i}{p/\alpha_i}}{t_i} \right)^{1/\alpha_i}
        \; ,
    \end{align*}
    which finishes the proof.
\end{proof}

%% file: auxiliary/martingale-decoupling.tex

\subsection{Decoupling of Adapted Sequences}
In the paper, we will need to analyse sums of martingale differences which are not independent.
This poses a problem because the lemmas of the previous section assumes that random variables are independent.
We will handle this issue by using a powerful result of Hitczenko~\cite{Hitczenko1994} to reduce the sums of martingale differences to a sum of independent variables.
Before the theorem, we need a bit of notation.

\begin{definition}
    Let $(X_i)_{i \in [n]}$ and $(Y_i)_{i \in [n]}$ be two sequences of random variables adapted
    to a filtration $(\mathcal{F}_i)_{i \in [n]}$. Then $(X_i)_{i \in [n]}$ and $(Y_i)_{i \in [n]}$ are
    \emph{tangent} with respect to $(\mathcal{F}_i)_{i \in [n]}$ if
    $(X_i \mid \mathcal{F}_{i - 1})$ and $(Y_i \mid \mathcal{F}_{i - 1})$ has the
    same distribution for all $i \in [n]$.
\end{definition}

\begin{definition}
    Let $(X_i)_{i \in [n]}$ be a sequence of random variables adapted
    to a filtration $(\mathcal{F}_i)_{i \in [n]}$ and let $\mathcal{G} \subseteq \mathcal{F}_{n - 1}$ be a $\sigma$-algebra.
    Then $(X_i)_{i \in [n]}$ satisfies the \emph{conditional independence condition} with respect to $\mathcal{G}$
    if $(X_i \mid \mathcal{F}_{i - 1})$ and $(X_i \mid \mathcal{G})$ have the same distribution for
    all $i \in [n]$, and $(X_i)_{i \in [n]}$ are conditionally independent given $\mathcal{G}$.
\end{definition}

\begin{definition}
    Let $(X_i)_{i \in [n]}$ and $(Y_i)_{i \in [n]}$ be two sequences of random variables
    which are tangent with respect to the filtration $(\mathcal{F}_i)_{i \in [n]}$.
    Let $\mathcal{G} \subseteq \mathcal{F}_{n - 1}$ be a $\sigma$-algebra.
    If $(Y_i)_{i \in [n]}$ satisfies the conditional independence condition with respect to $\mathcal{G}$
    then we say that $(Y_i)_{i \in [n]}$ is a \emph{decoupled sequence tangent} to $(X_i)_{i \in [n]}$.
\end{definition}

We can now state the theorem of Hitczenko~\cite{Hitczenko1994}.

\begin{theorem}[Hitczenko~\cite{Hitczenko1994}]\label{thm:Hitczenko-martingale-decoupling}
    There exists a universal constant $0 < M < \infty$ such that,
    for all $p \ge 1$ and all sequences of random variables $(X_i)_{i \in [n]}$
    and $(Y_i)_{i \in [n]}$ where $(Y_i)_{i \in [n]}$ is a decoupled sequence tangent to $(X_i)_{i \in [n]}$, then
    \begin{align*}
        \pnorm{\sum_{i \in [n]} X_i}{p}
            \le M \pnorm{\sum_{i \in [n]} Y_i}{p}
        \; .
    \end{align*}
\end{theorem}

Instead of using the result directly, we will instead use the following consequence of the theorem.
\begin{lemma}\label{lem:martingale-decoupling}
    Let $(X_i, \mathcal{F}_i)_{i \in [n]}$ be a filtered sequence.
    Assume there exists a sequence of random variables $(Y_i)_{i \in [n]}$ satisfying the following:
    \begin{enumerate}
        \item $(X_i \mid \mathcal{F}_{i - 1})$ and $(Y_i \mid \mathcal{F}_{i - 1})$ have the same distribution for every $i \in [n]$.
        \item The sequence $(Y_i)_{i \in [n]}$ is conditionally independent given $\mathcal{F}_{n - 1}$.
        \item $(Y_i \mid \mathcal{F}_{i - 1})$ and $(Y_i \mid \mathcal{F}_{n - 1})$ have the same distribution for every $i \in [n]$.
        \item $(X_i \mid \mathcal{F}_{i - 1})$ and $(X_i \mid \sigma(\mathcal{F}_{i - 1}, (Y_j)_{j \in [i + 1]})$ have the same distribution for every $i \in [n]$.
    \end{enumerate}
    Then for all $p \ge 1$,
    \begin{align*}
        \pnorm{\sum_{i \in [n]} X_i}{p}
            \le M \pnorm{\sum_{i \in [n]} Y_i}{p}
        \; ,
    \end{align*}
    where $M$ is a universal constant.
\end{lemma}
\begin{proof}
    We define the filtration $(\mathcal{H}_i)_{i \in [n]}$ by $\mathcal{H}_i = \sigma(\mathcal{F}_i, (Y_j)_{j \in [i + 1]})$ for $i \in [n]$.
    We then clearly have that $(X_i)_{i \in [n]}$ and $(Y_i)_{i \in [n]}$ are adapted to $(\mathcal{H}_i)_{i \in [n]}$.
    We will also see that they are tangent.
    Let $A \subseteq \R$ then,
    \begin{align*}
        \prbcond{Y_i \in A}{\mathcal{H}_{i - 1}}
            &= \epcond{\indicator{Y_i \in A}}{\mathcal{H}_{i - 1}}
            \\&= \epcond{\epcond{\indicator{Y_i \in A}}{\sigma(\mathcal{F}_{n - 1}, (Y_j)_{j \in [i]})}}{\mathcal{H}_{i - 1}}
            \\&= \epcond{\epcond{\indicator{Y_i \in A}}{\mathcal{F}_{n - 1}}}{\mathcal{H}_{i - 1}}
            \\&= \epcond{\epcond{\indicator{Y_i \in A}}{\mathcal{F}_{i - 1}}}{\mathcal{H}_{i - 1}}
            \\&= \epcond{\epcond{\indicator{X_i \in A}}{\mathcal{F}_{i - 1}}}{\mathcal{H}_{i - 1}}
            \\&= \prbcond{X_i \in A}{\mathcal{F}_{i - 1}}
            \\&= \prbcond{X_i \in A}{\mathcal{H}_{i - 1}}
    \end{align*}
    The first equality uses the power property of conditional expectation,
    the second equality uses the conditional independence property,
    the next two equalities follow by the equivalences of distributions,
    the second last equality follows by $\mathcal{F}_{i - 1} \subseteq \mathcal{H}_{i - 1}$,
    and the last equality follows by the equivalences of distributions.

    We have that $(Y_i)_{i \in [n]}$ are conditionally independent given $\mathcal{F}_{n - 1} \in \mathcal{H}_{n - 1}$, hence $ (Y_i)_{i \in [n]}$ is a decoupled sequence tangent to $(X_i)_{i \in [n]}$.
    Now \Cref{thm:Hitczenko-martingale-decoupling} give us the result.
\end{proof}

%% file: concentration/intro.tex


The goal of this chapter is to prove strong concentration results for tabulation based hashing.
The chapter is divided into three parts:
In the first part we generalize some of the results by Aamand. et~al.~\cite{aamand2020} to the case where we have partial keys, and we prove some auxiliary results which will be used in the later parts.
In the second part we improve the analysis of simple tabulation and provide a moment bound which holds for all moments.
In order to prove this result we first have to bound a technical quantity which will show up as a conditional variance in the proof.
Finally, in the last part we prove moment bounds for mixed tabulation.
They can be thought versions of Khintchine's inequality and Chernoff bound for mixed tabulation.

One of the main insights we use that differs from the previous analyses is that we work with symmetrized versions of simple tabulation and mixed tabulation.
We will in their respective section argue that this assumption is valid.

%% file: concentration/general-facts.tex


We need the following simple lemma that compares the growth rate of powers and logarithms.
This will be used extensively. 

\begin{lemma}\label{lem:log-vs-poly}
    Let $a > 0$ and $b > 0$ be positive reals.
    It then holds that for all $x > 1$,
    \begin{align*}
        \frac{x^a}{\log(x)^b}
            \ge \left( \frac{e}{b/a} \right)^b
        \; .
    \end{align*} 
\end{lemma}
\begin{proof}
    We write $\frac{x^a}{\log(x)^b} = \left( \frac{x^{a/b}}{\log(x)} \right)^b$ so we just need to minimize $\frac{x^{a/b}}{\log(x)}$.
    Taking the derivative we get,
    \begin{align*}
        \frac{d}{dx} \frac{x^{a/b}}{\log(x)}
            = \frac{\tfrac{a}{b}x^{a/b - 1}\log(x) - x^{a/b - 1}}{\log(x)}
        \; .
    \end{align*}
    From this it is clear that $\frac{x^{a/b}}{\log(x)}$ is minimized at $\hat{x} = e^{b/a}$.
    We then get that
    \begin{align*}
        \left( \frac{x^{a/b}}{\log(x)} \right)^b
            \ge \left( \frac{\hat{x}^{a/b}}{\log(\hat{x})} \right)^b
            = \left( \frac{e}{b/a} \right)^b
        \; .
    \end{align*}
\end{proof}

%% file: concentration/simple-tabulation.tex

\subsection{Improved Analysis for Simple Tabulation}

We start the section by arguing why we can assume that the simple tabulation functions are symmetrized.

\begin{lemma}\label{lem:simple-sym}
    Let $h \colon \Sigma^c \to [m]$ be a simple tabulation function, $v \colon \Sigma^{c} \times [m] \to \R$ a value function, and assume that $\sum_{j \in [m]} v(x, j) = 0$ for all keys $x \in \Sigma^c$.
    Then for every $p \ge 2$,
    \begin{align*}
        2^{-c} \pnorm{\sum_{x \in \Sigma^c} \eps(x) v(x, h(x))}{p}
            \le \pnorm{\sum_{x \in \Sigma^c} v(x, h(x))}{p}
            \le 2^c \pnorm{\sum_{x \in \Sigma^c} \eps(x) v(x, h(x))}{p}
        \; ,
    \end{align*}
    where $\eps \colon \Sigma^c \to \set{-1, 1}$ a simple tabulation sign function.
\end{lemma}
\begin{proof}
    We will prove the result by induction on $c$.
    The case $c = 1$ corresponds to \Cref{lem:sampling-symmetrize}.

    Now assume that $c > 1$ and that the results is true for values less than $c$.
    We define the $\sigma$-algebra $\mathcal{G} = \sigma((T(c-1, \alpha))_{\alpha \in \Sigma})$.
    Fix $\mathcal{G}$ and define $v' \colon \Sigma^{c - 1} \times [m] \to \R$ by
    \begin{align*}
        v'(x, j) = \sum_{\alpha \in \Sigma} v(x \cup \set{(c - 1, \alpha)}, T(c - 1, \alpha) \xor j)
        \; .
    \end{align*}
    Clearly, we have that $v'(x, j)$ is $\mathcal{G}$-measurable for $x \in \Sigma^{c - 1}, j \in [m]$
    and $\epcond{v'(x, h(x)}{\mathcal{G}} = 0$ for all $x \in \Sigma^{c - 1}$.
    We fix $\mathcal{G}$ and use the induction hypothesis to get that
    \begin{align*}
        2^{-(c - 1)}\pnormcond{\sum_{x \in \Sigma^{c - 1}} \eps(x) v'(x, h(x)) }{p}{\mathcal{G}}
            \le  \pnormcond{\sum_{x \in \Sigma^{c - 1}} v'(x, h(x)) }{p}{\mathcal{G}}
            \le 2^{c - 1} \pnormcond{\sum_{x \in \Sigma^{c - 1}} \eps(x) v'(x, h(x)) }{p}{\mathcal{G}}
        \; .
    \end{align*}
    So if we unfix $\mathcal{G}$ then we have that
    \begin{align}\label{eq:simple-symmetrize-step}
        2^{-(c - 1)}\pnorm{\sum_{x \in \Sigma^{c - 1}} \eps(x) v'(x, h(x)) }{p}
            \le  \pnorm{\sum_{x \in \Sigma^{c - 1}} v'(x, h(x)) }{p}
            \le 2^{c - 1} \pnorm{\sum_{x \in \Sigma^{c - 1}} \eps(x) v'(x, h(x)) }{p}
        \; .
    \end{align}
    Now we define the $\sigma$-algebra $\mathcal{H} = \sigma\!\left((h(\set{(j, \alpha)}), \eps(\set{(j, \alpha)}))_{j \in [c - 1], \alpha \in \Sigma}\right)$, and define $v'' \colon \Sigma \times [m] \to \R$ by $v''(\alpha, j) = \sum_{x \in \Sigma^{c - 1}} \eps(x)  v(x \cup \set{(c - 1, \alpha)}, h(x) \xor j)$.
    We then get that
    \begin{align*}
        \pnorm{\sum_{x \in \Sigma^{c - 1}} \eps(x) v'(x, h(x)) }{p}
            &= \pnorm{\sum_{x \in \Sigma^{c - 1}} \eps(x) \sum_{\alpha \in \Sigma} v(x \cup \set{(c - 1, \alpha)}, T(c - 1, \alpha) \xor h(x)) }{p}
            \\&= \pnorm{\sum_{\alpha \in \Sigma} v''(\alpha, T(c - 1, \alpha)) }{p}
        \; .
    \end{align*}
    We fix $\mathcal{H}$ and use \Cref{lem:sampling-symmetrize} to get that
    \begin{align*}
        2^{-1} \pnormcond{\sum_{\alpha \in \Sigma} \eps(\set{(c-1, \alpha)}) v''(\alpha, T(c - 1, \alpha)) }{p}{\mathcal{H}}
            &\le \pnormcond{\sum_{\alpha \in \Sigma} v''(\alpha, T(c - 1, \alpha)) }{p}{\mathcal{H}}
            \\&\le 2\pnormcond{\sum_{\alpha \in \Sigma} \eps(\set{(c-1, \alpha)}) v''(\alpha, T(c - 1, \alpha)) }{p}{\mathcal{H}}
        \; .
    \end{align*}
    We unfix ${\mathcal{G}}$ and get that
    \begin{align}\label{eq:simple-symmetrize-last}
        2^{-1} \pnorm{\sum_{\alpha \in \Sigma} \eps(\set{(c-1, \alpha)}) v''(\alpha, T(c - 1, \alpha)) }{p}
            &\le \pnorm{\sum_{\alpha \in \Sigma} v''(\alpha, T(c - 1, \alpha)) }{p}
            \\&\le 2\pnorm{\sum_{\alpha \in \Sigma} \eps(\set{(c-1, \alpha)}) v''(\alpha, T(c - 1, \alpha)) }{p}
        \; .
    \end{align}
    We now note that
    \begin{align*}
        \sum_{\alpha \in \Sigma} \eps(\set{(c-1, \alpha)}) v''(\alpha, T(c - 1, \alpha))
7            &= \sum_{\alpha \in \Sigma} \eps(\set{(c-1, \alpha)}) \sum_{x \in \Sigma^{c - 1}} \eps(x)  v(x \cup \set{(c - 1, \alpha)}, h(x) \xor T(c - 1, \alpha))
            \\&= \sum_{x \in \Sigma^c} \eps(x) v(x, h(x))
        \; .
    \end{align*}
    Thus combining \cref{eq:simple-symmetrize-step} and \cref{eq:simple-symmetrize-last} finishes the proof.
\end{proof}

We can then generalize a result by Aamand et~al.~\cite{aamand2020}.
The previous bound was only valid for $p = O(1)$ constant while we expand the applicability to all $p \ge 2$.
Surprisingly, the main insight is that by symmetrizing, the combinatorial arguments become much simpler.

\begin{lemma}\label{lem:simple-tab-gaussian}
    Let $h \colon \Sigma^c \to [m]$ be a simple tabulation function, $\eps \colon \Sigma^c \to \set{-1, 1}$ be a simple tabulation sign function, and $v_i \colon \Sigma^{c} \times [m] \to \R$ be value functions for $i \in [k]$.
    Then for every $p \ge 2$,
    \begin{align}\begin{split}
        &\pnorm{\sum_{x_0, \ldots, x_{k - 1} \in \Sigma^{c}} \sum_{\substack{j_0, \ldots, j_{k - 1} \in [m] \\ \bigxor_{i \in [k]} j_i = 0}} \prod_{i \in [k]} \eps(x_i) v_i(x_i, j_i \xor h(x_i)) }{p}
            \\&\qquad\qquad\le \sqrt{pk}^{c k} \prod_{i \in [k]} \norm{v_i}{2} \left( \frac{\sum_{x \in \Sigma^c} \norm{v_i[x]}{1}^2}{\sum_{x \in \Sigma^c} \norm{v_i[x]}{2}^2} \right)^{1/2 - 1/k}
        \; .
    \end{split}\end{align}
\end{lemma}
\begin{proof}
    We will argue that for every even integer $q \ge 2$,
    \begin{align}\begin{split}\label{eq:simple-tab-gaussian-helper}
        &\pnorm{\sum_{x_0 \in \Sigma^{I_0}, \ldots, x_{k - 1} \in \Sigma^{I_{k - 1}}} \sum_{\substack{j_0, \ldots, j_{k - 1} \in [m] \\ \bigxor_{i \in [k]} j_i = 0}} \prod_{i \in [k]} \eps(x_i) v_i(x_i, j_i \xor h(x_i)) }{q}
            \\&\qquad\qquad\le \sqrt{\frac{qk}{2}}^{c k} \prod_{i \in [k]} \norm{v_i}{2} \left( \frac{\sum_{x \in \Sigma^c} \norm{v_i[x]}{1}^2}{\sum_{x \in \Sigma^c} \norm{v_i[x]}{2}^2} \right)^{1/2 - 1/k}
    \end{split}\end{align}
    We claim that the result follows from this.
    Let $p \ge 2$ be a real number and let $q \ge 2$ be the unique even number such that $q \le p < q + 2$.
    Since $q \ge 2$ then $p \le 2q$ and we can then use Jensen's inequality to get that
    \begin{align*}
        &\pnorm{\sum_{x_0 \in \Sigma^{I_0}, \ldots, x_{k - 1} \in \Sigma^{I_{k - 1}}} \sum_{\substack{j_0, \ldots, j_{k - 1} \in [m] \\ \bigxor_{i \in [k]} j_i = 0}} \prod_{i \in [k]} \eps(x_i) v_i(x_i, j_i \xor h(x_i)) }{p}
            \\&\qquad\qquad\le \pnorm{\sum_{x_0 \in \Sigma^{I_0}, \ldots, x_{k - 1} \in \Sigma^{I_{k - 1}}} \sum_{\substack{j_0, \ldots, j_{k - 1} \in [m] \\ \bigxor_{i \in [k]} j_i = 0}} \prod_{i \in [k]} \eps(x_i) v_i(x_i, j_i \xor h(x_i)) }{2q}
            \\&\qquad\qquad\le \sqrt{\frac{2q k}{2}}^{c k} \prod_{i \in [k]} \norm{v_i}{2}^2 \left( \frac{\sum_{x \in \Sigma^c} \norm{v_i[x]}{1}}{\sum_{x \in \Sigma^c} \norm{v_i[x]}{2}^2} \right)^{1/2 - 1/k}
            \\&\qquad\qquad\le \sqrt{pk}^{c k} \prod_{i \in [k]} \norm{v_i}{2} \left( \frac{\sum_{x \in \Sigma^c} \norm{v_i[x]}{1}^2}{\sum_{x \in \Sigma^c} \norm{v_i[x]}{2}^2} \right)^{1/2 - 1/k}
        \; .
    \end{align*}

    All we need to do now is to prove \cref{eq:simple-tab-gaussian-helper}.
    Let $q \ge 2$ be an even integer.
    The goal is to apply \Cref{lem:weighted-sum-of-dependence} to prove the claim.
    First we define $f \colon \prod_{i \in [k]} \Sigma^{I_i} \to \R$ by
    \[
        f(x_0, \ldots, x_{k - 1}) = \sum_{\substack{j_0, \ldots, j_{k - 1} \in [m] \\ \bigxor_{i \in [k]} j_i = 0}} \prod_{i \in [k]} v_i(x_i, j_i \xor h(x_i))
            \; .
    \]
    We then want to bound
    \begin{align*}
        \pnorm{\sum_{x_0, \ldots, x_{k - 1} \in \Sigma^{c}} \left( \prod_{i \in [k]} \eps(x_i) \right) f(x_0, \ldots, x_{k - 1}) }{q}
    \end{align*}
    If we fix $h$ then we get that
    \begin{align*}
        &\pnormcond{\sum_{x_0, \ldots, x_{k - 1} \in \Sigma^{c}} \left( \prod_{i \in [k]} \eps(x_i) \right) f(x_0, \ldots, x_{k - 1}) }{q}{h}^q
            \\&\qquad\qquad= \sum_{\substack{x^{(0)}_0, \ldots, x^{(0)}_{k - 1}, \ldots,  x^{(q-1)}_0, \ldots, x^{(q-1)}_{k - 1} \in \Sigma^{c} \\ \bigxor_{j \in [q], i \in [k]} x^{(j)}_i = \emptyset}} \prod_{j \in [q]} f(x^{(j)}_0, \ldots, x^{(j)}_{k - 1})
        \; .
    \end{align*}
    Now we want to bound $f$ to a form such that we can use \Cref{lem:weighted-sum-of-dependence}.
    This will be done by use of the Cauchy-Schwartz inequality.
    \begin{align*}
        f(x_0, \ldots, x_{k - 1}) 
            &= \sum_{\substack{j_0, \ldots, j_{k - 1} \in [m]] \\ \bigxor_{i \in [k]} j_i = 0}} \prod_{i \in [k]} v_i(x_i, j_i \xor h(x_i))
            \\&= \sum_{j_0, \ldots, j_{k - 3} \in [m]} \left( \prod_{i \in [k - 2]} v_i(x_i, j_i \xor h(x_i)) \right) 
                \\&\cdot\sum_{j_{k - 2} \in [m]} v_{k - 2}(x_{k - 2}, j_{k - 2} \xor h(x_{k - 2})) v_{k - 1}\left(x_{k - 1}, h(x_{k - 1}) \xor \bigxor_{s \in [k - 1]} j_s \right)
            \\&\le \sum_{j_0, \ldots, j_{k - 3} \in [m]} \abs{\prod_{i \in [k - 2]} v_i(x_i, j_i \xor h(x_i))} \norm{v_{k - 2}[x_{k - 2}]}{2} \norm{v_{k - 1}[x_{k - 1}]}{2}
            \\&= \left( \prod_{i \in [k - 2]} \norm{v_i[x_i]}{1} \right) \norm{v_{k - 2}[x_{k - 2}]}{2} \norm{v_{k - 1}[x_{k - 1}]}{2}
            \\&= \left( \prod_{i \in [k - 2]} \frac{\norm{v_i[x_i]}{1}}{\norm{v_i[x_i]}{2}} \right) \prod_{i \in [k]} \norm{v_{i}[x_{i}]}{2}
        \; .
    \end{align*}
    Similarly, for all $i_1 \neq i_2 \in [k]$ we can prove that
    \begin{align*}
        f(x_0, \ldots, x_{k - 1}) 
            &\le \left( \prod_{i \in [k] \setminus \set{i_1, i_2}} \frac{\norm{v_i[x_i]}{1}}{\norm{v_i[x_i]}{2}} \right) \prod_{i \in [k]} \norm{v_{i}[x_{i}]}{2}
        \; .
    \end{align*}
    This implies that
    \begin{align*}
        f(x_0, \ldots, x_{k - 1}) 
            \le \prod_{i \in [k]} \norm{v_{i}[x_{i}]}{2} \left( \frac{\norm{v_i[x_i]}{1}}{\norm{v_i[x_i]}{2}} \right)^{1 - 2/k}
    \end{align*}
    We are now ready to use \Cref{lem:weighted-sum-of-dependence}.
    \begin{align*}
        &\pnormcond{\sum_{x_0, \ldots, x_{k - 1} \in \Sigma^{c}} \left( \prod_{i \in [k]} \eps(x_i) \right) f(x_0, \ldots, x_{k - 1}) }{q}{h}^q
            \\&\qquad\qquad\le \sum_{\substack{x^{(0)}_0, \ldots, x^{(0)}_{k - 1}, \ldots,  x^{(q-1)}_0, \ldots, x^{(q-1)}_{k - 1} \in \Sigma^{c} \\ \bigxor_{j \in [q], i \in [k]} x^{(j)}_i = \emptyset}} \prod_{j \in [q], i \in [k]} \norm{v_{i}[x^{(j)}_{i}]}{2} \left( \frac{\norm{v_i[x^{(j)}_i]}{1}}{\norm{v_i[x^{(j)}_i]}{2}} \right)^{1 - 2/k} 
            \\&\qquad\qquad\le \sqrt{\tfrac{qk}{2}}^{q c k } \left( \prod_{i \in [k]} \sqrt{\sum_{x \in \Sigma^{c}} \norm{v_{i}[x]}{2}^2 \left( \frac{\norm{v_i[x]}{1}^2}{\norm{v_i[x]}{2}^2} \right)^{1 - 2/k}} \right)^{q}
        \; .
    \end{align*}
    Now we define the random variables $R_i$ by $\prb{R_i = \frac{\norm{v_i[x]}{1}^2}{\norm{v_i[x]}{2}^2}} = \frac{\norm{v_i[x]}{2}^2}{\sum_{x \in \Sigma^c} \norm{v_i[x]}{2}^2}$, and note that
    \begin{align*}
        \sum_{x \in \Sigma^{c}} \norm{v_{i}[x]}{2}^2 \left( \frac{\norm{v_i[x]}{1}^2}{\norm{v_i[x]}{2}^2} \right)^{1 - 2/k}
            &= \left( \sum_{x \in \Sigma^{c}} \norm{v_{i}[x]}{2}^2 \right) \ep{R_i^{1 - 2/k}}
            \\&\le \left( \sum_{x \in \Sigma^{c}} \norm{v_{i}[x]}{2}^2 \right) \ep{R_i}^{1 - 2/k}
            \\&= \left( \sum_{x \in \Sigma^{c}} \norm{v_{i}[x]}{2}^2 \right) \left( \frac{\sum_{x \in \Sigma^c} \norm{v_i[x]}{1}^2}{\sum_{x \in \Sigma^c} \norm{v_i[x]}{2}^2} \right)^{1 - 2/k}
            \\&= \norm{v_{i}}{2}^2\left( \frac{\sum_{x \in \Sigma^c} \norm{v_i[x]}{1}^2}{\sum_{x \in \Sigma^c} \norm{v_i[x]}{2}^2} \right)^{1 - 2/k}
        \; .
    \end{align*}
    The inequality follows by Jensen's inequality.
    This implies that
    \begin{align*}
        &\pnormcond{\sum_{x_0, \ldots, x_{k - 1} \in \Sigma^{c}} \left( \prod_{i \in [k]} \eps(x_i) \right) f(x_0, \ldots, x_{k - 1}) }{q}{h}^q
            \\&\qquad\qquad\le \sqrt{\tfrac{qk}{2}}^{q c k } \left( \prod_{i \in [k]} \norm{v_{i}}{2} \left( \frac{\sum_{x \in \Sigma^c} \norm{v_i[x]}{1}^2}{\sum_{x \in \Sigma^c} \norm{v_i[x]}{2}^2} \right)^{1/2 - 1/k} \right)^{q}
        \; .
    \end{align*}
    Now taking the $q$'th root, give us that
    \begin{align*}
        &\pnorm{\sum_{x_0, \ldots, x_{k - 1} \in \Sigma^{c}} \left( \prod_{i \in [k]} \eps(x_i) \right) f(x_0, \ldots, x_{k - 1})}{q}
            \\&\qquad\qquad= \pnorm{\pnormcond{\sum_{x_0, \ldots, x_{k - 1} \in \Sigma^{c}} \left( \prod_{i \in [k]} \eps(x_i) \right) f(x_0, \ldots, x_{k - 1}) }{q}{h}}{q}
            \\&\qquad\qquad\le \pnorm{\sqrt{\tfrac{qk}{2}}^{c k } \left( \prod_{i \in [k]} \norm{v_{i}}{2} \left( \frac{\sum_{x \in \Sigma^c} \norm{v_i[x]}{1}^2}{\sum_{x \in \Sigma^c} \norm{v_i[x]}{2}^2} \right)^{1/2 - 1/k} \right)}{q}
            \\&\qquad\qquad= \sqrt{\tfrac{qk}{2}}^{c k } \prod_{i \in [k]} \norm{v_{i}}{2} \left( \frac{\sum_{x \in \Sigma^c} \norm{v_i[x]}{1}^2}{\sum_{x \in \Sigma^c} \norm{v_i[x]}{2}^2} \right)^{1/2 - 1/k}
        \; ,
    \end{align*}
    which finishes the proof of the lemma.
\end{proof}

\subsubsection{Bounding the Sum of Squares}

The goal of this section is to prove \Cref{lem:simple-chaos} from which we then get a bound of sum of squares of simple tabulation hashing.
We start by proving a result for simple tabulation hashing that will serve as the base for the proof.


\restateLemSimpleTabHoeffding
\begin{proof}
    The proof will be by induction on $c$.
    For $c = 1$ the result follows by \Cref{lem:sampling-hoeffding}.

    Now we assume that the result is true for $c - 1$.
    We define $v' \colon \Sigma^{c - 1} \times [m] \to \R$ by $v'(x, j) = \sum_{\alpha \in \Sigma} \eps_{c-1}(\alpha) v(x \cup \set{(c-1, \alpha)}, j \xor T(c-1, \alpha))$.
    The induction hypothesis then give us that
    \begin{align}\begin{split}\label{eq:simple-hoeffding-induction-step}
        \pnorm{\sum_{x \in \Sigma^c} \eps(x) v(x, h(x))}{p}
            &= \pnorm{\pnormcond{\sum_{x \in \Sigma^{c - 1}} \eps(x) v'(x, h(x))}{p}{T(\set{c - 1}\times\Sigma)}}{p}
            \\&\le \pnorm{\sqrt{K_{c - 1} p \left(\max\!\set{p, \log(m)}\right)^{c - 2} \frac{\sum_{x \in \Sigma^{c - 1}} \norm{v'[x]}{\infty}^2}{\log\!\left( \frac{e^2 m \sum_{x \in \Sigma^{c - 1}} \norm{v'[x]}{\infty}^2}{\sum_{x \in \Sigma^{c - 1}} \norm{v'[x]}{2}^2} \right)^{c - 1}}}}{p}
            \\&\le \sqrt{K_{c - 1} p \left(\max\!\set{p, \log(m)}\right)^{c - 2}} \pnorm{\frac{\sum_{x \in \Sigma^{c - 1}} \norm{v'[x]}{\infty}^2}{\log\!\left( \frac{e^2 m \sum_{x \in \Sigma^{c - 1}} \norm{v'[x]}{\infty}^2}{\sum_{x \in \Sigma^{c - 1}} \norm{v'[x]}{2}^2} \right)^{c - 1}}}{p/2}^{1/2}
        \; .
    \end{split}\end{align}
    We define the function $f \colon \R_{\ge 0}^2 \to \R_{\ge 0}$ by
    \begin{align*}
        f(x, y) = \begin{cases}
            0                                                    &\text{if $y = 0$} \\
            \frac{x}{\log\!\left( \frac{e^2 x}{y} \right)^{c - 1}} &\text{if $0 < y \le x$} \\
            \frac{x}{2^{c-1}}                                    &\text{otherwise}
        \end{cases}
        \; .
    \end{align*}
    Clearly, $\sum_{x \in \Sigma^{c - 1}} \norm{v'[x]}{\infty}^2 \ge \tfrac{1}{m} \sum_{x \in \Sigma^{c - 1}} \norm{v'[x]}{2}^2$, hence we have that
    \begin{align*}
        \frac{\sum_{x \in \Sigma^{c - 1}} \norm{v'[x]}{\infty}^2}{\log\!\left( \frac{e^2 m \sum_{x \in \Sigma^{c - 1}} \norm{v'[x]}{\infty}^2}{\sum_{x \in \Sigma^{c - 1}} \norm{v'[x]}{2}^2} \right)^{c - 1}}
            = f\left(\sum_{x \in \Sigma^{c - 1}} \norm{v'[x]}{\infty}^2, \tfrac{1}{m} \sum_{x \in \Sigma^{c - 1}} \norm{v'[x]}{2}^2 \right)
        \; .
    \end{align*}
    It is easy to check that $f(\lambda x, \lambda y) = \lambda f(x, y)$ so by \Cref{lem:fn-moment} we get that
    \begin{align}\begin{split}\label{eq:simple-hoeffding-f}
        &\pnorm{f\left(\sum_{x \in \Sigma^{c - 1}} \norm{v'[x]}{\infty}^2, \tfrac{1}{m} \sum_{x \in \Sigma^{c - 1}} \norm{v'[x]}{2}^2 \right)}{p/2}
            \\&\qquad\qquad\qquad\le 2^{1/p} f\left(\pnorm{\sum_{x \in \Sigma^{c - 1}} \norm{v'[x]}{\infty}^2}{p/2}, \tfrac{1}{m} \pnorm{\sum_{x \in \Sigma^{c - 1}} \norm{v'[x]}{2}^2}{p/2} \right)
            \\&\qquad\qquad\qquad\le \sqrt{2} f\left(\pnorm{\sum_{x \in \Sigma^{c - 1}} \norm{v'[x]}{\infty}^2}{p/2}, \tfrac{1}{m} \pnorm{\sum_{x \in \Sigma^{c - 1}} \norm{v'[x]}{2}^2}{p/2} \right)
        \; .
    \end{split}\end{align}

    We define $v_x \colon \Sigma \times [m] \to \R$ for every $x \in \Sigma^{c - 1}$ by $v_x(\alpha, j) = v(x \cup \set{({c-1}, \alpha)}, j)$.
    We then have that $v'(x, j) = \sum_{\alpha \in \Sigma} \eps_{c-1}(\alpha) v_x(\alpha, j \xor T(c - 1, \alpha))$.
    
    Let $\bar{p} = \max\!\set{p, \log(m)}$.
    \begin{align*}
        \pnorm{\sum_{x \in \Sigma^{c - 1}} \norm{v'[x]}{\infty}^2}{p/2}
            &= \pnorm{\sum_{x \in \Sigma^{c - 1}} \max_{j \in [m]} v'(x, j)^2}{p/2}
            \\&\le \sum_{x \in \Sigma^{c - 1}}\pnorm{\max_{j \in [m]} v'(x, j)^2}{p/2}
            \\&= \sum_{x \in \Sigma^{c - 1}} \pnorm{\max_{j \in [m]} \abs{v'(x, j)}}{p}^2
            \\&\le \sum_{x \in \Sigma^{c - 1}} \pnorm{\max_{j \in [m]} \abs{v'(x, j)}}{\bar{p}}^2
            \\&\le \sum_{x \in \Sigma^{c - 1}} \left(\sum_{j \in [m]} \pnorm{v'(x, j)}{\bar{p}}^{\bar{p}} \right)^{2/\bar{p}}
            \\&\le \sum_{x \in \Sigma^{c - 1}} \left(m \max_{j \in [m]} \pnorm{v'(x, j)}{\bar{p}}^{\bar{p}} \right)^{2/\bar{p}}
            \\&\le e \sum_{x \in \Sigma^{c - 1}} \max_{j \in [m]} \pnorm{v'(x, j)}{\bar{p}}^2
    \end{align*}
    Now we will use that $v'(x, j) = \sum_{\alpha \in \Sigma} \eps(\set{(c-1,\alpha)}) v_x(\alpha, j \xor T(c - 1, \alpha))$ and \Cref{lem:sampling-hoeffding}.
    \begin{align*}
        \pnorm{v'(x, j)}{\bar{p}}^2
            = \pnorm{\sum_{\alpha \in \Sigma} \eps_{c-1}(\alpha) v_x(\alpha, j \xor T(c - 1, \alpha))}{\bar{p}}^2
            \le C_1 \bar{p} \frac{\sum_{\alpha \in \Sigma} \norm{v_x[\alpha]}{\infty}^2}{\log \left( \frac{e^2 m \sum_{\alpha \in \Sigma} \norm{v_x[\alpha]}{\infty}^2}{\sum_{\alpha \in \Sigma} \norm{v_x[\alpha]}{2}^2} \right)} 
    \end{align*}
    So we have that
    \begin{align*}
        \pnorm{\sum_{x \in \Sigma^{c - 1}} \norm{v'[x]}{\infty}^2}{p/2}
            &\le e \sum_{x \in \Sigma^{c - 1}} \max_{j \in [m]} \pnorm{v'(x, j)}{\bar{p}}^2
            \\&\le C_1 e \bar{p} \sum_{x \in \Sigma^{c - 1}} \frac{\sum_{\alpha \in \Sigma} \norm{v_x[\alpha]}{\infty}^2}{\log \left( \frac{e^2 m \sum_{\alpha \in \Sigma} \norm{v_x[\alpha]}{\infty}^2}{\sum_{\alpha \in \Sigma} \norm{v_x[\alpha]}{2}^2} \right)} 
        \; .    
    \end{align*}
    Now we use \Cref{lem:sum-concave} to get that
    \begin{align}\begin{split}\label{eq:simple-hoeffding-max}
        \pnorm{\sum_{x \in \Sigma^{c - 1}} \norm{v'[x]}{\infty}^2}{p/2}
            &\le C_1 e \bar{p} \sum_{x \in \Sigma^{c - 1}} \frac{\sum_{\alpha \in \Sigma} \norm{v_x[\alpha]}{\infty}^2}{\log \left( \frac{e^2 m \sum_{\alpha \in \Sigma} \norm{v_x[\alpha]}{\infty}^2}{\sum_{\alpha \in \Sigma} \norm{v_x[\alpha]}{2}^2} \right)} 
            \\&\le C_1 e \bar{p} \frac{\sum_{x \in \Sigma^{c - 1}} \sum_{\alpha \in \Sigma} \norm{v_x[\alpha]}{\infty}^2}{\log \left( \frac{e^2 m \sum_{x \in \Sigma^{c - 1}} \sum_{\alpha \in \Sigma} \norm{v_x[\alpha]}{\infty}^2}{\sum_{x \in \Sigma^{c - 1}} \sum_{\alpha \in \Sigma} \norm{v_x[\alpha]}{2}^2} \right)} 
            \\&= C_1 e \bar{p} \frac{\sum_{x \in \Sigma^{c}} \norm{v[x]}{\infty}^2}{\log \left( \frac{e^2 m \sum_{x \in \Sigma^{c}} \norm{v[x]}{\infty}^2}{\sum_{x \in \Sigma^{c}} \norm{v[x]}{2}^2} \right)} 
    \end{split}\end{align}

    We again use that $v'(x, j) = \sum_{\alpha \in \Sigma} \eps_{c-1}(\alpha) v_x(\alpha, j \xor T(c - 1, \alpha))$ to get that
    \begin{align*}
        \pnorm{\sum_{x \in \Sigma^{c - 1}} \norm{v'[x]}{2}^2}{p/2} 
            &= \pnorm{\sum_{x \in \Sigma^{c - 1}} \sum_{j \in [m]} v'(x, j)^2}{p/2}
            \\&\le \sum_{x \in \Sigma^{c - 1}} \pnorm{\sum_{j \in [m]} v'(x, j)^2}{p/2}
            \\&= \sum_{x \in \Sigma^{c - 1}} \pnorm{\sum_{j \in [m]} \left( \sum_{\alpha \in \Sigma} \eps_{c-1}(\alpha) v_x(\alpha, j \xor T(c - 1, \alpha)) \right)^2}{p/2}
            \\&= \sum_{x \in \Sigma^{c - 1}} \pnorm{\sum_{j \in [m]} \sum_{\alpha, \beta \in \Sigma} \eps_{c-1}(\alpha) \eps_{c-1}(\beta) v_x(\alpha, j \xor T(c - 1, \alpha)) v_x(\beta, j \xor T(c - 1, \beta))}{p/2}
    \end{align*}
    Now we can use \Cref{lem:simple-tab-gaussian} to get that
    \begin{align*}
        \pnorm{\sum_{j \in [m]} \sum_{\alpha, \beta \in \Sigma} \eps_{c-1}(\alpha) \eps_{c-1}(\beta) v_x(\alpha, j \xor T(c - 1, \alpha)) v_x(\beta, j \xor T(c - 1, \beta))}{p/2}
            \le 2 p \sum_{\alpha \in \Sigma} \norm{v_x[\alpha]}{2}^2
    \end{align*}
    This implies that
    \begin{align}\label{eq:simple-hoeffding-squares}
        \pnorm{\sum_{x \in \Sigma^{c - 1}} \norm{v'[x]}{2}^2}{p/2}
            \le 2 p \sum_{x \in \Sigma^{c - 1}} \sum_{\alpha \in \Sigma} \norm{v_x[\alpha]}{2}^2
            = 2 p \sum_{x \in \Sigma^{c}} \norm{v[x]}{2}^2
            \le 2 \bar{p} \sum_{x \in \Sigma^{c}} \norm{v[x]}{2}^2
        \; .
    \end{align}

    Combining \cref{eq:simple-hoeffding-induction-step}, \cref{eq:simple-hoeffding-f}, \cref{eq:simple-hoeffding-max}, and \cref{eq:simple-hoeffding-squares} we get that
    \begin{align*}
        \pnorm{\sum_{x \in \Sigma^c} \eps(x) v(x, h(x))}{p}
            \le &\sqrt{K_{c - 1} p \left(\max\!\set{p, \log(m)}\right)^{c - 2}} 
            \\&\cdot f\left(C_1 e \bar{p} \frac{\sum_{x \in \Sigma^{c}} \norm{v[x]}{\infty}^2}{\log \left( \frac{e^2 m \sum_{x \in \Sigma^{c}} \norm{v[x]}{\infty}^2}{\sum_{x \in \Sigma^{c}} \norm{v[x]}{2}^2} \right)}, \frac{1}{m} 2 \bar{p} \sum_{x \in \Sigma^{c}} \norm{v[x]}{2}^2 \right)^{1/2}
        \; .    
    \end{align*}
    We will now argue that
    \begin{align*}
        f\left(C_1 e \bar{p} \frac{\sum_{x \in \Sigma^{c}} \norm{v[x]}{\infty}^2}{\log \left( \frac{e^2 m \sum_{x \in \Sigma^{c}} \norm{v[x]}{\infty}^2}{\sum_{x \in \Sigma^{c}} \norm{v[x]}{2}^2} \right)}, \frac{1}{m} 2 \bar{p} \sum_{x \in \Sigma^{c}} \norm{v[x]}{2}^2 \right)
            \le L c \frac{\bar{p}}{\log \left( \frac{e^2 m \sum_{x \in \Sigma^{c}} \norm{v[x]}{\infty}^2}{\sum_{x \in \Sigma^{c}} \norm{v[x]}{2}^2} \right)^c}
        \; ,
    \end{align*}
    which will finish the proof.

    We will use \Cref{lem:log-vs-poly} to get that
    \begin{align*}
        \frac{\left( \frac{e^2 m \sum_{x \in \Sigma^{c}} \norm{v[x]}{\infty}^2}{\sum_{x \in \Sigma^{c}} \norm{v[x]}{2}^2} \right)^{1/c}}{\log \left( \frac{e^2 m \sum_{x \in \Sigma^{c}} \norm{v[x]}{\infty}^2}{\sum_{x \in \Sigma^{c}} \norm{v[x]}{2}^2} \right)}
            \ge \frac{e}{c}
        \; .
    \end{align*}
    This implies that
    \begin{align*}
        \frac{C_1 e \bar{p} \frac{\sum_{x \in \Sigma^{c}} \norm{v[x]}{\infty}^2}{\log \left( \frac{e^2 m \sum_{x \in \Sigma^{c}} \norm{v[x]}{\infty}^2}{\sum_{x \in \Sigma^{c}} \norm{v[x]}{2}^2} \right)}}{\frac{1}{m} 2 \bar{p} \sum_{x \in \Sigma^{c}} \norm{v[x]}{2}^2}
            \ge \frac{C_1 }{2 c} \left( \frac{e^2 m \sum_{x \in \Sigma^{c}} \norm{v[x]}{\infty}^2}{\sum_{x \in \Sigma^{c}} \norm{v[x]}{2}^2} \right)^{1 - 1/c}
        \; .
    \end{align*}
    We then get that
    \begin{align*}
        f\left(C_1 e \bar{p} \frac{\sum_{x \in \Sigma^{c}} \norm{v[x]}{\infty}^2}{\log \left( \frac{e^2 m \sum_{x \in \Sigma^{c}} \norm{v[x]}{\infty}^2}{\sum_{x \in \Sigma^{c}} \norm{v[x]}{2}^2} \right)}, \frac{1}{m} 2 \bar{p} \sum_{x \in \Sigma^{c}} \norm{v[x]}{2}^2 \right)
            &\le f\left(2 e c \bar{p} \frac{\sum_{x \in \Sigma^{c}} \norm{v[x]}{\infty}^2}{\log \left( \frac{e^2 m \sum_{x \in \Sigma^{c}} \norm{v[x]}{\infty}^2}{\sum_{x \in \Sigma^{c}} \norm{v[x]}{2}^2} \right)}, \frac{1}{m} 2 \bar{p} \sum_{x \in \Sigma^{c}} \norm{v[x]}{2}^2 \right)
            \\&= 2 e c \frac{\bar{p}}{\log \left( \left(\frac{e^2 m \sum_{x \in \Sigma^{c}} \norm{v[x]}{\infty}^2}{\sum_{x \in \Sigma^{c}} \norm{v[x]}{2}^2} \right)^{1 - 1/c} \right)^c}
            \\&= 2 e c \frac{\bar{p}}{\log \left( \frac{e^2 m \sum_{x \in \Sigma^{c}} \norm{v[x]}{\infty}^2}{\sum_{x \in \Sigma^{c}} \norm{v[x]}{2}^2}  \right)^c \left(1 - \frac{1}{c}\right)^c}
            \\&\le 8 e c \frac{\bar{p}}{\log \left( \frac{e^2 m \sum_{x \in \Sigma^{c}} \norm{v[x]}{\infty}^2}{\sum_{x \in \Sigma^{c}} \norm{v[x]}{2}^2}  \right)^c}
        \; .
    \end{align*}
    This finishes the proof.
\end{proof}

We will now expand the result of \Cref{lem:simple-tab-hoeffding} to chaoses of simple tabulation hashing.
For this we need the following decoupling lemma of de la Pena et al.~\cite{delapena1994}.

\begin{lemma}[Decoupling~\cite{delapena1994}]\label{lem:general-decoupling}
    Let $(f^{(j)}_{i_0, \ldots, i_{k - 1}})_{i_0, \ldots, i_{k - 1} \in [n], j \in [m]}$ be a multiindexed array of real numbers, and assume that $f^{(j)}_{i_0, \ldots, i_{k - 1}} = 0$ if $i_l = i_{l'}$ for some $l \neq l'$.
    Let $(X^{(j)}_i)_{i \in [n]}$ be a sequence of independent and identically distributed random variables for every $j \in [m]$.

    Define $(Y^{(j)}_{i, l})_{i \in [n], l \in [k]}$ to be a sequence of independent and identically distributed random variables which has the same distribution as $X^{(j)}_0$ for every $j \in [m]$.
    Then for every $p \ge 2$,
    \begin{align*}
        \pnorm{\sum_{j \in [m]} \sum_{i_0, \ldots, i_{k - 1} \in [n]} f^{(d)}_{i_0, \ldots, i_{k - 1}} \prod_{l \in [k]} X^{(j)}_{i_l}}{p}
            \le L_k \pnorm{\sum_{j \in [m]} \sum_{i_0, \ldots, i_{k - 1} \in [n]} f^{(d)}_{i_0, \ldots, i_{k - 1}} \prod_{l \in [k]} Y^{(j)}_{i_l, l}}{p}
        \; ,
    \end{align*}
    where $L_k \le k^k$ if $\ep{X^{(j)}_0} = 0$ for all $j \in [m]$, and $L_k \le (2k + 1)^k$ otherwise.
\end{lemma}

If we specialize it to simple tabulation hashing we get the following corollary.

\begin{corollary}\label{cor:simple-decoupling}
    Let $(F_{\alpha_0, \ldots, \alpha_{k - 1}})_{\alpha_0, \ldots, \alpha_{k - 1} \in \Sigma}$ be a multiindexed array of real functions $F_{\alpha_0, \ldots, \alpha_{k - 1}} \colon [m] \to \R$, and assume that $F_{\alpha_0, \ldots, \alpha_{k - 1}} = 0$ if $x_l = x_{l'}$ for some $l \neq l'$.
    Let $h \colon \Sigma \to [m]$ be a fully random function and let $\eps \colon \Sigma \to \set{-1, 1}$ be a fully random sign function.
    Let $h' \colon \Sigma^{k} \to [m]$ be a simple tabulation hash function and let $\eps' \colon \Sigma^k \to \set{-1, 1}$ be a simple tabulation sign function.
    Then for every $p \ge 2$,
    \begin{align*}
        \pnorm{\sum_{\alpha_0, \ldots, \alpha_{k - 1} \in \Sigma} F_{\alpha_0, \ldots, \alpha_{k - 1}}(h(\alpha_0) \xor \ldots \xor h(\alpha_{k - 1})) \prod_{l \in [k]} \eps(\alpha_l) }{p}
            \le k^k \pnorm{\sum_{x \in \Sigma^k} F_{x}(h'(x)) \eps'(x) }{p}
        \; .
    \end{align*}
\end{corollary}
\begin{proof}
    We start by noticing that we can write the expression as follows,
    \begin{align*}
        &\pnorm{\sum_{\alpha_0, \ldots, \alpha_{k - 1} \in \Sigma} F_{\alpha_0, \ldots, \alpha_{k - 1}}(h(\alpha_0) \xor \ldots \xor h(\alpha_{k - 1})) \prod_{l \in [k]} \eps(\alpha_l) }{p}
            \\&\qquad\qquad\qquad= \pnorm{\sum_{\alpha_0, \ldots, \alpha_{k - 1} \in \Sigma} \sum_{j_0, \ldots, j_{k - 1}} F_{\alpha_0, \ldots, \alpha_{k - 1}}(j_0 \xor \ldots \xor j_{k - 1}) \prod_{l \in [k]} \eps(\alpha_l) \indicator{h(\alpha_l) = j_l} }{p}
        \; .
    \end{align*}
    We can then use \Cref{lem:general-decoupling} to get that
    \begin{align*}
        &\pnorm{\sum_{\alpha_0, \ldots, \alpha_{k - 1} \in \Sigma} \sum_{j_0, \ldots, j_{k - 1}} F_{\alpha_0, \ldots, \alpha_{k - 1}}(j_0 \xor \ldots \xor j_{k - 1}) \prod_{l \in [k]} \eps(\alpha_l) \indicator{h(\alpha_l) = j_l} }{p}
            \\&\qquad\qquad\qquad\le k^k \pnorm{\sum_{\alpha_0, \ldots, \alpha_{k - 1} \in \Sigma} \sum_{j_0, \ldots, j_{k - 1}} F_{\alpha_0, \ldots, \alpha_{k - 1}}(j_0 \xor \ldots \xor j_{k - 1}) \prod_{l \in [k]} \eps'(\set{l, \alpha_l}) \indicator{T(l, \alpha_l) = j_l} }{p}
        \; .
    \end{align*}
    We can then finish the proof by reversing the rewriting,
    \begin{align*}
        &\pnorm{\sum_{\alpha_0, \ldots, \alpha_{k - 1} \in \Sigma} \sum_{j_0, \ldots, j_{k - 1}} F_{\alpha_0, \ldots, \alpha_{k - 1}}(j_0 \xor \ldots \xor j_{k - 1}) \prod_{l \in [k]} \eps'(\set{l, \alpha_l}) \indicator{T(l, \alpha_l) = j_l} }{p}
            \\&\qquad\qquad\qquad= \pnorm{\sum_{\alpha_0, \ldots, \alpha_{k - 1} \in \Sigma} F_{\alpha_0, \ldots, \alpha_{k - 1}}(T(0, \alpha_0) \xor \ldots \xor T(k - 1, \alpha_{k  - 1})) \eps'(\alpha_0, \ldots, \alpha_{k - 1}) }{p}
            \\&\qquad\qquad\qquad= \pnorm{\sum_{x \in \Sigma^k} F_{x}(h'(x)) \eps'(x) }{p}
        \; .
    \end{align*}
\end{proof}

We can now prove our result for chaoses of simple tabulation hashing.

\begin{lemma}\label{lem:simple-chaos}
    Let $h \colon \Sigma^c \to [m]$ be a simple tabulation function, $\eps \colon \Sigma^c \to \set{-1, 1}$ be a simple tabulation sign function, and $v_i \colon \Sigma^c \times [m] \to \R$ be value function for $i \in [k]$.
    For every real number $p \ge 2$,
    \begin{align*}\begin{split}
        &\pnorm{\sum_{x_0, \ldots, x_{k - 1} \in \Sigma^{c}} \sum_{\substack{j_0, \ldots, j_{k - 1} \in [m] \\ \bigxor_{i \in [k]} j_i = 0}} \prod_{i \in [k]} \eps(x_i) v_i(x_i, j_i \xor h(x_i)) }{p}
            \\&\qquad\qquad\qquad\le \left( \frac{L c k^3 \max\!\set{p, \log(m)}}{\log\!\left( \prod_{i \in [k]} \left( \frac{e^2 m \sum_{x \in \Sigma^c} \norm{v_i[x]}{\infty}^2}{\sum_{x \in \Sigma^c} \norm{v_i[x]}{2}^2} \right)^{1/k} \right)} \right)^{ck/2} \prod_{i \in [k]} \norm{v_i}{2} \left( \frac{\norm{v_i}{1}}{\norm{v_i}{2}} \right)^{1 - 2/k} 
        \; ,
    \end{split}\end{align*}
    where $L$ is a universal constant.
\end{lemma}
\begin{proof}
    For every $j \in [c]$ we define $\pi_j \colon \Sigma^c \to \set{i} \times \Sigma$ to be the projection onto the $i$'th position character, i.e., for a key $x = \set{(0, \alpha_0), \ldots, (c - 1, \alpha_{c - 1})}$ we have that $\pi_i(x) = (i, \alpha_{i})$.

    We define $\bar{p} = \max\!\set{p, \log(m)}$ to ease notation.

    We make the observation that $\sum_{\substack{j_0, \ldots, j_{k - 1} \in [m] \\ \bigxor_{i \in [k]} j_i = 0}} \prod_{i \in [k]} \eps(x_i) v_i(x_i, j_i \xor h(x_i))$ depends only on $\bigxor_{i \in [k]} h(x_i)$.
    More precisely, we note that if we define $v' \colon \Sigma^{ck} \times [m] \to \R$ by,
    \begin{align*}
        v'((x_0, \ldots, x_{k - 1}), j) = \sum_{\substack{j_0, \ldots, j_{k - 1} \in [m] \\ \bigxor_{i \in [k]} j_i = j}} \prod_{i \in [k]} v_i(x_i, j_i) 
    \end{align*}
    then we have that
    \begin{align*}
        \sum_{\substack{j_0, \ldots, j_{k - 1} \in [m] \\ \bigxor_{i \in [k]} j_i = 0}} \prod_{i \in [k]} \eps(x_i) v_i(x_i, j_i \xor h(x_i))
            = v'((x_0, \ldots, x_{k - 1}), \bigxor_{i \in [k]} h(x_i))
    \end{align*} 
    This implies that we can the expression into sub-expressions depending on the number of distinct characters at each position.

    Let $t_0, \ldots, t_{c - 1}$ be even integers less than $k$.
    Now fix pairs $((s^{(i)}_l, r^{(i)}_l))_{l \in [t_i/2]}$ for $i \in [k]$ and define the set $X$ by,
    \begin{align*}        
        X = \Big\{(x_0, \ldots, x_{k - 1}) \in \left(\Sigma^{c}\right)^k \; \Big| \; \forall i \in [c] \colon \Big(\forall &l \in [t_i] \colon \pi_{i}(x_{s^{(i)}_l}) = \pi_{i}(x_{r^{(i)}_l}) 
            \\&\wedge (\pi_{i}(x_j))_{j \in [k] \setminus \bigcup_{l \in [t_i]} \set{s^{(i)}_l, r^{(i)}_l}} \text{ are all distinct} \Big) \Big\}
        \; .
    \end{align*}
    We define $T^{(i)} \colon [c] \times \Sigma \to [m]$ to be independent copies of $T$ for $i \in [k]$, and similarly, define $\eps^{(i)} \colon \Sigma \to \set{-1, 1}$ to be independent copies of $\eps$ for $i \in [k]$.
    We define the set, 
    \[
        R_i = \setbuilder{v \in [c]}{i \not\in \bigcup_{l \in [t_i/2]} \set{s^{(v)}_l, r^{(v)}_l}}
        \; ,
    \]
    for $i \in [k]$.
    We now use \Cref{cor:simple-decoupling} to get that
    \begin{align*}
        &\pnorm{\sum_{x_0, \ldots, x_{k - 1} \in X} v'((x_0, \ldots, x_{k - 1}), \bigxor_{i \in [k]} h(x_i)) \prod_{i \in [k]} \eps(x_i)}{p}
            \\&\qquad\le \prod_{i \in [c]} (k - t_i)^{k - t_i} \pnorm{\sum_{x_0, \ldots, x_{k - 1} \in X} v'((x_0, \ldots, x_{k - 1}), \bigxor_{i \in [k]} \bigxor_{l \in R_i} T^{(i)}(l, \pi_l(x_i))) \prod_{i \in [k]} \prod_{l \in R_i} \eps^{(i)}(l, \pi_l(x_i)) }{p}
    \end{align*}
    This corresponds to a simple tabulation function with $ck - \sum_{i \in [k]} t_i$ characters.
    We can then use \Cref{lem:simple-tab-hoeffding} to get that
    \begin{align*}
        &\pnorm{\sum_{x_0, \ldots, x_{k - 1} \in X} v'((x_0, \ldots, x_{k - 1}), \bigxor_{i \in [k]} \bigxor_{l \in R_i} T^{(i)}(l, \pi_l(x_i))) \prod_{i \in [k]} \prod_{l \in R_i} \eps^{(i)}(l, \pi_l(x_i)) }{p}
            \\&\qquad\qquad\le \sqrt{\left(L (ck - \sum_{i \in [k]} t_i) \bar{p}\right)^{ck - \sum_{i \in [k]} t_i} \frac{\sum_{x \in X} \norm{v'[x]}{\infty}^2}{\log\!\left( \frac{e^2 m \sum_{x \in X} \norm{v'[x]}{\infty}^2}{\sum_{x \in X} \norm{v'[x]}{2}^2} \right)^{ck - \sum_{i \in [k]} t_i}}}
    \end{align*}
    Now repeated use of Cauchy-Schwartz as in the proof of \Cref{lem:simple-tab-gaussian} implies that
    \begin{align*}
        \sum_{x \in X} \norm{v'[x]}{\infty}^2
            &\le \prod_{i \in [k]} \norm{v_{i}}{2}^2 \left( \frac{\sum_{x \in \Sigma^c} \norm{v_i[x]}{1}^2}{\sum_{x \in \Sigma^c} \norm{v_i[x]}{2}^2} \right)^{1 - 2/k} \; , \\
        \sum_{x \in X} \norm{v'[x]}{2}^2
            &\le \prod_{i \in [k]} \norm{v_{i}}{2}^2 \left( \frac{\sum_{x \in \Sigma^c} \norm{v_i[x]}{1}^2}{\sum_{x \in \Sigma^c} \norm{v_i[x]}{2}^2} \right)^{1 - 1/k} \; .
    \end{align*}
    We then get that
    \begin{align*}
        &\pnorm{\sum_{x_0, \ldots, x_{k - 1} \in X} v'((x_0, \ldots, x_{k - 1}), \bigxor_{i \in [k]} \bigxor_{l \in R_i} T^{(i)}(l, \pi_l(x_i))) \prod_{i \in [k]} \prod_{l \in R_i} \eps^{(i)}(l, \pi_l(x_i)) }{p}
            \\&\qquad\qquad\le \left( \frac{L (ck - \sum_{i \in [k]} t_i) \bar{p}}{\log\!\left( \prod_{i \in [k]} \left( \frac{e^2 m \sum_{x \in \Sigma^c} \norm{v_i[x]}{2}^2}{\sum_{x \in \Sigma^c} \norm{v_i[x]}{1}^2} \right)^{1/k} \right)}\right)^{(ck - \sum_{i \in [k]} t_i)/2}
             \prod_{i \in [k]} \norm{v_{i}}{2} \left( \frac{\sum_{x \in \Sigma^c} \norm{v_i[x]}{1}^2}{\sum_{x \in \Sigma^c} \norm{v_i[x]}{2}^2} \right)^{1/2 - 1/k}
    \end{align*}
    Now we note that given $(t_i)_{i \in [k]}$ we can choose the pairs $((s^{(i)}_l, r^{(i)}_l))_{l \in [t_i/2]}$ for $i \in [k]$ in $\prod_{i \in [k]} \binom{k}{t_i} (t_i - 1)!!$ ways.
    Summing all the possible values for $(t_i)_{i \in [k]}$ we get that
    \begin{align*}
        &\pnorm{\sum_{x_0, \ldots, x_{k - 1} \in \Sigma^{c}} \sum_{\substack{j_0, \ldots, j_{k - 1} \in [m] \\ \bigxor_{i \in [k]} j_i = 0}} \prod_{i \in [k]} \eps(x_i) v_i(x_i, j_i \xor h(x_i)) }{p}
            \\&\qquad\le \sum_{t_0, \ldots, t_{k - 1}} \left( \prod_{i \in [k]} \binom{k}{t_i} (t_i - 1)!! (k - t_i)^{k - t_i} \right) 
            \\&\qquad\qquad\cdot\left( \frac{L (ck - \sum_{i \in [k]} t_i) \bar{p}}{\log\!\left( \prod_{i \in [k]} \left( \frac{e^2 m \sum_{x \in \Sigma^c} \norm{v_i[x]}{2}^2}{\sum_{x \in \Sigma^c} \norm{v_i[x]}{1}^2} \right)^{1/k} \right)}\right)^{(ck - \sum_{i \in [k]} t_i)/2}
            \prod_{i \in [k]} \norm{v_{i}}{2} \left( \frac{\sum_{x \in \Sigma^c} \norm{v_i[x]}{1}^2}{\sum_{x \in \Sigma^c} \norm{v_i[x]}{2}^2} \right)^{1/2 - 1/k}
            \\&\qquad\le \left( \frac{L_2 c k^3 \bar{p}}{\log\!\left( \prod_{i \in [k]} \left( \frac{e^2 m \sum_{x \in \Sigma^c} \norm{v_i[x]}{2}^2}{\sum_{x \in \Sigma^c} \norm{v_i[x]}{1}^2} \right)^{1/k} \right)}\right)^{ck/2}
            \prod_{i \in [k]} \norm{v_{i}}{2} \left( \frac{\sum_{x \in \Sigma^c} \norm{v_i[x]}{1}^2}{\sum_{x \in \Sigma^c} \norm{v_i[x]}{2}^2} \right)^{1/2 - 1/k}
    \end{align*}
    This finishes the proof.
\end{proof}

It now becomes easy to prove \Cref{lem:simple-sum-squares}.

\restateLemSimpleSumSquares
\begin{proof}
    This follows by \Cref{lem:simple-chaos} since,
    \begin{align*}
        \pnorm{\sum_{j \in [m]} \left( \sum_{x \in \Sigma^c} \eps(x) v(x, h(x)) \right)^2 }{p}
            = \pnorm{\sum_{x, y \in \Sigma^c} \sum_{j \in [m]}  \eps(x)\eps(y) v(x, h(x)) v(y, h(y)) }{p}
        \; .
    \end{align*}
\end{proof}

\subsubsection{Concentration Result for Simple Tabulation Hashing}

We are now ready to prove the main result of the section.
Note that by using \Cref{lem:simple-sym}, the result can be extended to the case without symmetrization, which proves \Cref{thm:simple-tab-moments}.
We warn the reader that the proof is long and technical.


\restateThmSimpleTabMoments
\begin{proof}
    We will prove the result by induction on $c$.
    For $c = 1$ it corresponds to using a fully random hash function, and the result follows by \Cref{thm:sampling-moments}.

    Now we assume that $c > 1$ and that the result is true for values less than $c$.
    We note that without loss of generality we can assume that $M_v = 1$.
    Let $w \colon \Sigma^c \to \R$ be a function defined by $w(x) = \norm{v[x]}{2}^2$ and for $X \subseteq \Sigma^c$ we overload the notation of $w$ to write $w(X) = \sum_{x \in X} w(x)$.
    Furthermore, we define $w_\infty(X) = \max_{x \in X} w(x)$.

    Now applying \Cref{lem:group-of-keys} we obtain an ordering of position characters $\set{\alpha_0, \ldots, \alpha_{r - 1}} = [c] \times \Sigma$ where $r = c\abs{\Sigma}$,
    satisfying that the groups $G_i = \setbuilder{x \in \Sigma^{c}}{\alpha_i \in x \wedge x \subseteq \set{\alpha_0, \ldots, \alpha_i}}$ has the property that $w(G_i) \le w(\Sigma^c)^{1 - 1/c} w_{\infty}(\Sigma^c)^{1/c}$ for every $i \in [r]$.
    
    We define the random variables 
    \begin{align*}
        X_i^{(j)} &= \sum_{x \in G_i} \eps(x \setminus \set{\alpha_i}) v(x, j \xor h(x \setminus \set{\alpha_i}) )
            \; , \\
        Y_i &= \eps(\alpha_i) X_i^{(h(\alpha_i))} 
        \; ,
    \end{align*}
    for all $i \in [r], j \in [m]$.
    With this notation we have that
    \begin{align*}
        V^{\text{simple}}_v = V = \sum_{i \in [r]} Y_i
        \; .
    \end{align*}
    We let $(\mathcal{F}_i)_{i \in [r]}$ be a filtration where $\mathcal{F}_i = \sigma((h(\alpha_k), \eps(\alpha_k))_{k \in [i + 1]})$  for $i \in [r]$.
    It is easy to see that $X_i^{(j)}$ is $\mathcal{F}_{i - 1}$-measurable, $Y_i^{(j)}$ is $\mathcal{F}_{i}$-measurable, and $\epcond{Y_i^{(j)}}{\mathcal{F}_{i - 1}} = 0$ for all $j \in [m], i \in [r]$.
    We thus have that $(Y_i, \mathcal{F}_i)_{i \in [r]}$ is a martingale difference sequence.
    We furthermore notice that
    \[
        \varcond{Y_i}{\mathcal{F}_{i}} = \frac{1}{m} \sum_{k \in [m]} \left(X^{(k)}_i\right)^2
        \; ,
    \]
    for all $j \in [m], i \in [r]$.

    Let $h' \colon \Sigma^c \to [m]$ be a simple tabulation hash function independent of $h$, and $\eps' \colon \Sigma^c \to [m]$ be a simple tabulation sign function independent of $\eps$.
    We define the random variables $Z_i^{(j)} = \eps'(\alpha_i) X^{(j \xor h'(\alpha_i))}$.
    We can now easily check that $(Z_i)_{i \in [r]}$ satisfies the properties needed for \Cref{lem:martingale-decoupling}:
    \begin{enumerate}
        \item $(Y_i \mid \mathcal{F}_{i - 1})$ and $(Z_i \mid \mathcal{F}_{i - 1})$ have the same distribution for every $i \in [r]$.
        \item The sequence $(Z_i)_{i \in [r]}$ is conditionally independent given $\mathcal{F}_{r - 1}$.
        \item $(Z_i \mid \mathcal{F}_{i - 1})$ and $(Z_i \mid \mathcal{F}_{r - 1})$ have the same distribution for every $i \in [r]$.
        \item $(Y_i \mid \mathcal{F}_{i - 1})$ and $(Y_i \mid \sigma(\mathcal{F}_{i - 1}, (Z_k)_{k \in [i + 1]})$ have the same distribution for every $i \in [r]$.
    \end{enumerate}
    Now \Cref{lem:martingale-decoupling} then implies that
    \begin{align}\label{eq:simple-concentration-decoupling}
        \pnorm{\sum_{i \in [r]} Y_i}{p}
            &\le M\pnorm{\sum_{i \in [r]} Z_i}{p}
        \; .
    \end{align}

    We now use that $(Z_i)_{i \in [r]}$ are conditionally independent given $h$, so fixing $h$ and using \Cref{thm:sampling-moments} we get that
    \begin{align}\label{eq:simple-concentration-sampling}
        \pnormcond{\sum_{i \in [r]} Z_i}{p}{h}
            \le 16e \Psi_p\left( \max_{i \in [r], j \in [m]} \abs{X^{(j)}_i}, \tfrac{\sum_{i \in [r], j \in [m]} \left( X^{(j)}_i \right)^2}{m} \right)
        \; .
    \end{align}
    Since $\Psi_p(\lambda M, \lambda^2 \sigma^2) = \lambda \Psi_p(M, \sigma^2)$, we then use \Cref{lem:fn-moment} to get that
    \begin{align}\begin{split}\label{eq:simple-concentration-inner}
        \pnorm{\Psi_p\left( \max_{i \in [r], j \in [m]} \abs{X^{(j)}_i}, \tfrac{\sum_{i \in [r], j \in [m]} \left( X^{(j)}_i \right)^2}{m} \right)}{p}
            &\le 2^{1/p} \Psi_p\left( \pnorm{\max_{i \in [r], j \in [m]} \abs{X^{(j)}_i}}{p}, \frac{1}{m}\pnorm{\sum_{i \in [r], j \in [m]} \left( X^{(j)}_i \right)^2}{p/2}^{1/2} \right)
            \\&\le \sqrt{2} \Psi_p\left( \pnorm{\max_{i \in [r], j \in [m]} \abs{X^{(j)}_i}}{p}, \frac{1}{m}\pnorm{\sum_{i \in [r], j \in [m]} \left( X^{(j)}_i \right)^2}{p/2} \right)
        \; .
    \end{split}\end{align}

    We set $\bar{p} = \max\!\set{p, \log(m) + \log\!\left(\frac{w(\Sigma^c)}{w_\infty(\Sigma^c)}\right)/c}$.
    With this notation we have that
    \[
        \gamma_p = \frac{\bar{p}}{\log\!\left(\min_{x \in \Sigma^c} \tfrac{e^2 m \sum_{j \in [m]} v(x, j)^2}{\left(\sum_{j \in [m]} \abs{v(x, j)} \right)^2}\right)}
        \; .
    \]

    We start by bounding $\pnorm{\max_{i \in [r], j \in [m]} \abs{X^{(j)}_i}}{p}$.
    By the induction hypothesis we get that
    \begin{align*}
        \pnorm{\max_{i \in [r], j \in [m]} \abs{X^{(j)}_i}}{p}
            &\le \pnorm{\max_{i \in [r], j \in [m]} \abs{X^{(j)}_i}}{\bar{p}}
            \\&\le \left(\sum_{i \in [r], j \in [m]} \pnorm{X^{(j)}_i}{\bar{p}}^{\bar{p}} \right)^{1/\bar{p}}
            \\&\le \left(m \sum_{i \in [r]} \max_{j \in [m]} \pnorm{X^{(j)}_i}{\bar{p}}^{\bar{p}} \right)^{1/\bar{p}}
            \\&\le \left(m \sum_{i \in [r]} L_1^{p} \Psi_{\bar{p}}\left(K_{c - 1} \gamma_p^{c - 2},
                K_{c - 1} \frac{\gamma_p^{c - 2} w(G_i)}{m} \right)^{\bar{p}} \right)^{1/\bar{p}}
            \\&\le L_1 \left(m \sum_{i \in [r]} \Psi_{\bar{p}}\left(K_{c - 1} \gamma_p^{c - 2},
            K_{c - 1} \frac{\gamma_p^{c - 2} w(G_i)}{m} \right)^{\bar{p}} \right)^{1/\bar{p}}
        \; .
    \end{align*}
    An easy observation is that $\Psi_{\bar{p}}(M, \sigma^2)^{\bar{p}}$ is a convex function in $\sigma^2$.
    So using that $\max_{i \in [r]} w(G_i) \le w(\Sigma^c)^{1 - 1/c} w_{\infty}(\Sigma^c)^{1/c}$ and that $\sum_{i \in [r]} w(G_i) = w(\Sigma^c)$ we get that
    \begin{align}\begin{split}\label{eq:simple-concentration-max}
        \pnorm{\max_{i \in [r], j \in [m]} \abs{X^{(j)}_i}}{p}
            &\le L_1 \left(m \left(\tfrac{w(\Sigma^c)}{w_{\infty}(\Sigma^c)}\right)^{1/c} \right)^{1/\bar{p}}
                \Psi_{\bar{p}}\left(K_{c - 1} \gamma_p^{c - 2},
                K_{c - 1} \frac{\gamma_p^{c - 2} w(\Sigma^c)^{1 - 1/c} w_{\infty}(\Sigma^c)^{1/c}}{m} \right)
            \\&\le L_1 e \Psi_{\bar{p}}\left(K_{c - 1} \gamma_p^{c - 2}, K_{c - 1} \frac{\gamma_p^{c - 2} w(\Sigma^c)^{1 - 1/c} w_{\infty}(\Sigma^c)^{1/c}}{m} \right)
        \; .
    \end{split}\end{align}
    The last inequality follows since $\bar{p} \ge \log(m) + \log\!\left(\tfrac{w(\Sigma^c)}{w_{\infty}(\Sigma^c)}\right)/c$.

    We will bound $\pnorm{\sum_{i \in [r], j \in [m]} \left( X^{(j)}_i \right)^2}{p/2}$ by using the triangle inequality and \Cref{lem:simple-sum-squares}.
    \begin{align}\begin{split}\label{eq:simple-concentration-squares}
        \pnorm{\sum_{i \in [r], j \in [m]} \left( X^{(j)}_i \right)^2}{p/2}
            &\le \sum_{i \in [r]} \pnorm{\sum_{j \in [m]} \left( X^{(j)}_i \right)^2}{p/2}
            \\&\le  K'_{c - 1} \max\!\set{1, \left( \frac{p/2}{\log\!\left(\tfrac{e^2 m \sum_{x \in \Sigma^c} \norm{v[x]}{2}^2}{\sum_{x \in \Sigma^c} \norm{v[x]}{1}^2}\right)} \right)^c } \sum_{i \in [r]} w(G_i)
            \\&\le  K'_{c - 1} \max\!\set{1, \left( \frac{p}{\log\!\left( \min_{x \in \Sigma^c} \tfrac{e^2 m \norm{v[x]}{2}^2}{\norm{v[x]}{1}^2}\right)} \right)^c } \sum_{i \in [r]} w(G_i)
            \\&\le K'_{c - 1} \gamma_p^{c - 1} w(\Sigma^c)
            \\&\le K_c \gamma_p^{c - 1} w(\Sigma^c)
        \; .
    \end{split}\end{align}
    Here $K'_{c - 1}$ is the constant depending on $c - 1$ which we get from \Cref{lem:simple-sum-squares}.
    
    Now combining \cref{eq:simple-concentration-inner}, \cref{eq:simple-concentration-max}, and \cref{eq:simple-concentration-squares} we get that
    \begin{align}\begin{split}\label{eq:simple-concentration-psi-bound}
        &\pnorm{\Psi_p\left( \max_{i \in [r], j \in [m]} \abs{X^{(j)}_i}, \tfrac{\sum_{i \in [r], j \in [m]} \left( X^{(j)}_i \right)^2}{m} \right)}{p}
            \\&\qquad\le \sqrt{2} \Psi_p\left(L_1 e \Psi_{\bar{p}}\left(K_{c - 1} \gamma_p^{c - 2}, K_{c - 1} \frac{\gamma_p^{c - 2} w(\Sigma^c)^{1 - 1/c} w_{\infty}(\Sigma^c)^{1/c}}{m} \right), K_c \gamma_p^{c - 1} \frac{w(\Sigma^c)}{m}\right)
            \\&\qquad= \sqrt{2} \Psi_p\left(L_1 e \Psi_{\bar{p}}\left(K_{c - 1} \gamma_p^{c - 2}, K_{c - 1} \frac{\gamma_p^{c - 2} w(\Sigma^c)^{1 - 1/c} w_{\infty}(\Sigma^c)^{1/c}}{m} \right), K_c \gamma_p^{c - 1} \sigma_v^2\right)
        \; .
    \end{split}\end{align}

    Now we will consider two cases depending on $w(\Sigma^c)$.

    \paragraph{Case 1. $w(\Sigma^c) \le \left(\bar{p} e^{-2} K_{c - 1} \gamma_p^{c - 2} \right)^{c/(c - 1)} m \left( \frac{m}{w_\infty(\Sigma^c)} \right)^{1/(2c - 1)}$.}
    In this case we will show that, 
    \begin{align}\label{eq:simple-concentration-case-one}      
        L_1 e \Psi_{\bar{p}}\left(K_{c - 1} \gamma_p^{c - 2}, K_{c - 1} \frac{\gamma_p^{c - 2} w(\Sigma^c)^{1 - 1/c} w_{\infty}(\Sigma^c)^{1/c}}{m} \right)
            \le K_c \gamma_p^{c - 1}
        \; .
    \end{align}
    We first notice that
    \begin{align}\begin{split}\label{eq:simple-concentration-case-one-calc}
        \frac{ K_{c - 1}\frac{\gamma_p^{c - 2} w(\Sigma^c)^{1 - 1/c} w_{\infty}(\Sigma^c)^{1/c}}{m}}{ K_{c - 1}^2 \gamma_p^{2c - 4}}
            &= \frac{w(\Sigma^c)^{1 - 1/c} w_{\infty}(\Sigma^c)^{1/c}}{m K_{c - 1} \gamma_p^{c - 2}}
            \\&\le \frac{\left(\bar{p} e^{-2} K_{c - 1} \gamma_p^{c - 2} \right) m^{(2c - 2)/(2c - 1)} w_{\infty}(\Sigma^c)^{-(c - 1)/c(2c - 1)} w_{\infty}(\Sigma^c)^{1/c}}{m K_{c - 1} \gamma_p^{c - 2}}
            \\&= e^{-2} \bar{p} \left(\frac{w_{\infty}(\Sigma^c)}{m} \right)^{1/(2c - 1)}
            \\&= e^{-2} \bar{p} \left(\max_{x \in \Sigma^c} \frac{\norm{v[x]}{2}^2}{m} \right)^{1/(2c - 1)}
            \\&\le e^{-2} \bar{p} \left(\max_{x \in \Sigma^c} \frac{\norm{v[x]}{1}^2}{m \norm{v[x]}{2}^2} \right)^{1/(2c - 1)}
            \\&\le e^{-2} \bar{p}
        \; .
    \end{split}\end{align}

    Now by \Cref{eq:psi-bernstein} we get that
    \begin{align*}
        L_1 e \Psi_{\bar{p}}\left(K_{c - 1} \gamma_p^{c - 2}, K_{c - 1} \frac{\gamma_p^{c - 2} w(\Sigma^c)^{1 - 1/c} w_{\infty}(\Sigma^c)^{1/c}}{m} \right)
            &\le L_1 e \frac{\bar{p}}{e \log \left( \frac{\bar{p} m K_{c - 1}^2 \gamma_p^{2c - 4}}{\gamma_p^{c - 2} w(\Sigma^c)^{1 - 1/c} w_{\infty}(\Sigma^c)^{1/c}} \right) } K_{c - 1} \gamma_p^{c - 2}
            \\&\le L_1\frac{\bar{p}}{\log \left( e^2 \left( \min_{x \in \Sigma^c} \frac{m \norm{v[x]}{2}^2}{\norm{v[x]}{1}^2} \right)^{1/(2c - 1)} \right) } K_{c - 1} \gamma_p^{c - 2}
            \\&\le 2L_1 c K_{c - 1} \gamma_p^{c - 1}
            \\&\le K_{c} \gamma_p^{c - 1}
        \; .
    \end{align*}

    Where we have used that $\frac{ K_{c - 1}\frac{\gamma_p^{c - 2} w(\Sigma^c)^{1 - 1/c} w_{\infty}(\Sigma^c)^{1/c}}{m}}{ K_{c - 1}^2 \gamma_p^{2c - 4}} \le e^{-2} \bar{p} \left(\max_{x \in \Sigma^c} \frac{\norm{v[x]}{1}^2}{m \norm{v[x]}{2}^2} \right)^{1/(2c - 1)}$ which follows from \cref{eq:simple-concentration-case-one-calc}, and that $2L_1 c K_{c - 1} \le K_c$.

     Now combining \cref{eq:simple-concentration-decoupling}, \cref{eq:simple-concentration-sampling}, \cref{eq:simple-concentration-psi-bound}, and \cref{eq:simple-concentration-case-one} we get the result. 

    \paragraph{Case 2. $w(\Sigma^c) > \left(\bar{p} e^{-2} K_{c - 1} \gamma_p^{c - 2} \right)^{c/(c - 1)} m \left( \frac{m}{w_{\infty}(\Sigma^c)} \right)^{1/(2c - 1)}$.}
    In this case we will argue that
    \begin{align}\begin{split}\label{eq:simple-concentration-case-two}
        &\Psi_p\left(L_1 e \Psi_{\bar{p}}\left(K_{c - 1} \gamma_p^{c - 2}, K_{c - 1} \frac{\gamma_p^{c - 2} w(\Sigma^c)^{1 - 1/c} w_{\infty}(\Sigma^c)^{1/c}}{m} \right), K_c \gamma_p^{c - 1} w(\Sigma^c)\right)
            \\&\qquad\qquad\le \tfrac{1}{2} \sqrt{p} \sqrt{K_c \gamma_p^{c - 1} w(\Sigma^c)}
        \; .
    \end{split}\end{align}
    
    We use \cref{eq:psi-bernstein} to get that
    \begin{align*}
        &L_1 e \Psi_{\bar{p}}\left( K_{c - 1} \gamma_p^{c - 2}, K_{c - 1} \frac{\gamma_p^{c - 2} w(\Sigma^c)^{1 - 1/c} w_{\infty}(\Sigma^c)^{1/c}}{m} \right)
            \\&\qquad\qquad\le \max\!\set{L_1 \tfrac{e}{2} \sqrt{\bar{p} K_{c - 1} \frac{\gamma_p^{c - 2} w(\Sigma^c)^{1 - 1/c} w_{\infty}(\Sigma^c)^{1/c}}{m}},
                L_1 \tfrac{1}{2} \bar{p} K_{c - 1} \gamma_p^{c - 2}}
        \; .
    \end{align*}
    We apply \cref{eq:psi-bernstein} again to obtain that
    \begin{align*}
        &\Psi_p\left(L_1 e \Psi_{\bar{p}}\left(K_{c - 1} \gamma_p^{c - 2}, K_{c - 1} \frac{\gamma_p^{c - 2} w(\Sigma^c)^{1 - 1/c} w_{\infty}(\Sigma^c)^{1/c}}{m} \right), K_c \gamma_p^{c - 1} w(\Sigma^c)\right)
            \\&\qquad\qquad\le \max\!\set{\tfrac{1}{2} \sqrt{p} \sqrt{K_c \gamma_p^{c - 1} \frac{w(\Sigma^c)}{m}}, \tfrac{1}{2e}p L_1 e \Psi_{\bar{p}}\left(K_{c - 1} \gamma_p^{c - 2}, K_{c - 1} \frac{\gamma_p^{c - 2} w(\Sigma^c)^{1 - 1/c} w_{\infty}(\Sigma^c)^{1/c}}{m} \right)}
        \; .
    \end{align*}
    Combining the two estimates give us that
    \begin{align*}
        &\Psi_p\left(L_1 e \Psi_{\bar{p}}\left(K_{c - 1} \gamma_p^{c - 2}, K_{c - 1} \frac{\gamma_p^{c - 2} w(\Sigma^c)^{1 - 1/c} w_{\infty}(\Sigma^c)^{1/c}}{m} \right), K_c \gamma_p^{c - 1} w(\Sigma^c)\right)
            \\&\qquad\qquad\le \max\!\set{\tfrac{1}{2} \sqrt{p} \sqrt{K_c \gamma_p^{c - 1} \frac{w(\Sigma^c)}{m}}, \tfrac{L_1}{4} p \sqrt{\bar{p} K_{c - 1} \frac{\gamma_p^{c - 2} w(\Sigma^c)^{1 - 1/c} w_{\infty}(\Sigma^c)^{1/c}}{m}},
            \tfrac{L_1}{4e} p \bar{p} K_{c - 1} \gamma_p^{c - 2}}
    \end{align*}
    We will show that the max is equal to $\tfrac{1}{2} \sqrt{p} \sqrt{K_c \gamma_p^{c - 1} \frac{w(\Sigma^c)}{m}}$ which will show \cref{eq:simple-concentration-case-two}.

    First we show that $\tfrac{1}{2} \sqrt{p} \sqrt{K_c \gamma_p^{c - 1} \frac{w(\Sigma^c)}{m}} \ge \tfrac{L_1}{4} p \sqrt{\bar{p} K_{c - 1} \frac{\gamma_p^{c - 2} w(\Sigma^c)^{1 - 1/c} w_{\infty}(\Sigma^c)^{1/c}}{m}}$.
    We note that this equivalent with showing that $\frac{\tfrac{1}{4} p K_c \gamma_p^{c - 1} \frac{w(\Sigma^c)}{m}}{\tfrac{L_1^2}{16} p^2 \bar{p} K_{c - 1} \frac{\gamma_p^{c - 2} w(\Sigma^c)^{1 - 1/c} w_{\infty}(\Sigma^c)^{1/c}}{m}} \ge 1$,
    \begin{align*}
        \frac{\tfrac{1}{4} p K_c \gamma_p^{c - 1} \frac{w(\Sigma^c)}{m}}{\tfrac{L_1^2}{16} p^2 \bar{p} K_{c - 1} \frac {\gamma_p^{c - 2} w(\Sigma^c)^{1 - 1/c} w_{\infty}(\Sigma^c)^{1/c}}{m}}
            &= \frac{4 K_c}{L_1^2 K_{c - 1}} \cdot \frac {\gamma_p w(\Sigma^c)^{1/c}}{p \bar{p} w_{\infty}(\Sigma^c)^{1/c}}
            \\&> \frac{4 K_c}{L_1^2 K_{c - 1}} \cdot \frac {\gamma_p \left(\bar{p} e^{-2} K_{c - 1} \gamma_p^{c - 2} \right)^{1/(c - 1)} m^{1/c} \left( \frac{m}{w_{\infty}(\Sigma^c)} \right)^{1/(c(2c - 1))}}{\bar{p}^2 w_{\infty}(\Sigma^c)^{1/c}}
            \\&= \frac{4K_c}{L_1^2 K_{c - 1}} \frac{\left(e^{-2} K_{c - 1} \right)^{1/(c - 1)} }{ \log\!\left( \min_{x \in \Sigma^c} \frac{e^2 m \norm{v[x]}{2}^2}{\norm{v[x]}{1}^2} \right)^{2 - 1/(c -1)}}
                \left( \frac{m}{w_{\infty}(\Sigma^c)} \right)^{2/(2c - 1)}
            \\&\ge \frac{4 K_c}{e^2 L_1^2 K_{c - 1}^{1 - 1/(c - 1)}} \frac{1 }{ \log\!\left( \min_{x \in \Sigma^c} \frac{e^2 m \norm{v[x]}{2}^2}{\norm{v[x]}{1}^2} \right)^{2 - 1/(c -1)}}
                \left( \min_{x \in \Sigma^c} \frac{e^2 m \norm{v[x]}{2}^2}{\norm{v[x]}{1}^2} \right)^{2/(2c - 1)}
    \end{align*}
    Now by using \Cref{lem:log-vs-poly} we get that
    \begin{align*}
        &\frac{4 K_c}{e^2 L_1^2 K_{c - 1}^{1 - 1/(c - 1)}} \frac{1 }{ \log\!\left( \min_{x \in \Sigma^c} \frac{e^2 m \norm{v[x]}{2}^2}{\norm{v[x]}{1}^2} \right)^{2 - 1/(c -1)}}
                \left( \min_{x \in \Sigma^c} \frac{e^2 m \norm{v[x]}{2}^2}{\norm{v[x]}{1}^2} \right)^{2/(2c - 1)}
            \\&\qquad\qquad\qquad\ge \frac{4 K_c}{e^2 L_1^2 K_{c - 1}^{1 - 1/(c - 1)}} \left( \frac{2e}{(2 - 1/(c - 1))(2c - 1)} \right)^{2 - 1/(c - 1)}
            \\&\qquad\qquad\qquad\ge \frac{4 K_c}{e^2 L_1^2 K_{c - 1}^{1 - 1/(c - 1)}} \left( \frac{e}{2 c} \right)^{2}
            \\&\qquad\qquad\qquad= \frac{K_c}{L_1^2 K_{c - 1}^{1 - 1/(c - 1)}} c^{-2}
        \; .
    \end{align*}
    Now $K_c = \left( L_2 c \right)^c$ and $K_{c - 1} \le \left( L_2 c \right)^{c - 1}$ so we get that
    \begin{align*}
        \frac{K_c}{L_1^2 K_{c - 1}^{1 - 1/(c - 1)}} c^{-2}
            \ge \frac{\left( L_2 c \right)^c}{L_1^2 \left( L_2 c \right)^{c - 2}} c^{-2}
            = \left( \frac{L_2}{L_1} \right)^2
            \ge 1
        \; .
    \end{align*}
    The last inequality follows by choosing $L_2 \ge L_1$.
    Combining it all we get that
    \begin{align*}
        \frac{\tfrac{1}{4} p K_c \gamma_p^{c - 1} \frac{w(\Sigma^c)}{m}}{\tfrac{L_1^2}{16} p^2 \bar{p} K_{c - 1} \frac {\gamma_p^{c - 2} w(\Sigma^c)^{1 - 1/c} w_{\infty}(\Sigma^c)^{1/c}}{m}}
            \ge 1
        \; .
    \end{align*}
    This implies that $\tfrac{1}{2} \sqrt{p} \sqrt{K_c \gamma_p^{c - 1} \frac{w(\Sigma^c)}{m}} \ge \tfrac{L_1}{4} p \sqrt{\bar{p} K_{c - 1} \frac {\gamma_p^{c - 2} w(\Sigma^c)^{1 - 1/c} w_{\infty}(\Sigma^c)^{1/c}}{m}}$
    as we wanted.

    Next we show that $\tfrac{1}{2} \sqrt{p} \sqrt{K_c \gamma_p^{c - 1} \frac{w(\Sigma^c)}{m}} \ge \tfrac{L_1}{4 e} p \bar{p}  K_{c - 1} \gamma_p^{c - 2}$.
    Again we note that this equivalent with showing that $\frac{\tfrac{1}{4} K_c p \gamma_p^{c - 1} \frac{w(\Sigma^c)}{m}}{\tfrac{L_1^2}{16 e^2} p^2 \bar{p}^2 K_{c - 1}^2 \gamma_p^{2c - 4}} \ge 1$,
    \begin{align*}
        \frac{\tfrac{1}{4} K_c p \gamma_p^{c - 1} \frac{w(\Sigma^c)}{m}}{\tfrac{L_1^2}{16 e^2} p^2 \bar{p}^2 K_{c - 1}^2 \gamma_p^{2c - 4}}
            &= \frac{4 e^2 K_c}{L_1^2 K_{c - 1}^2} \frac{  w(\Sigma^c)}{m  p \bar{p}^2 \gamma_p^{c - 3}}
            \\&> \frac{4 e^2 K_c}{L_1^2 K_{c - 1}^2} \frac{ \left(\bar{p} e^{-2} K_{c - 1} \gamma_p^{c - 2} \right)^{c/(c - 1)} m \left( \frac{m}{w_{\infty}(\Sigma^c)} \right)^{1/(2c - 1)} }{m  \bar{p}^3 \gamma_p^{c - 3}}
            \\&\ge \frac{4 K_c}{e L_1^2 K_{c - 1}^{1 - 1/(c - 1)}} \frac{ 1 }{\log\!\left( \min_{x \in \Sigma^c} \frac{e^2 m \norm{v[x]}{2}^2}{\norm{v[x]}{1}^2} \right)^{2 - 1/(c -1)}} \left( \min_{x \in \Sigma^c} \frac{e^2 m \norm{v[x]}{2}^2}{\norm{v[x]}{1}^2} \right)^{1/(2c - 1)}
            \\&\ge \frac{4 K_c}{e L_1^2 K_{c - 1}^{1 - 1/(c - 1)}} \frac{ 1 }{\log\!\left( \min_{x \in \Sigma^c} \frac{e^2 m \norm{v[x]}{2}^2}{\norm{v[x]}{1}^2} \right)^{2 - 1/(c -1)}} \left( \min_{x \in \Sigma^c} \frac{e^2 m \norm{v[x]}{2}^2}{\norm{v[x]}{1}^2} \right)^{1/(2c - 1)}
    \end{align*}
    Again we use \Cref{lem:log-vs-poly} to obtain that
    \begin{align*}
        &\frac{4 K_c}{e L_1^2 K_{c - 1}^{1 - 1/(c - 1)}} \frac{ 1 }{\log\!\left( \min_{x \in \Sigma^c} \frac{e^2 m \norm{v[x]}{2}^2}{\norm{v[x]}{1}^2} \right)^{2 - 1/(c -1)}} \left( \min_{x \in \Sigma^c} \frac{e^2 m \norm{v[x]}{2}^2}{\norm{v[x]}{1}^2} \right)^{1/(2c - 1)}
            \\&\qquad\qquad\ge \frac{4 K_c}{e L_1^2K_{c - 1}^{1 - 1/(c - 1)}} \left( \frac{e}{(2c - 1)(2 - 1/(c - 1))} \right)^{2 - 1/(c - 1)}
            \\&\qquad\qquad\ge \frac{4 K_c}{e L_1^2 K_{c - 1}^{1 - 1/(c - 1)}} \left( \frac{e}{4c} \right)^{2 - 1/(c - 1)}
            \\&\qquad\qquad\ge \frac{4 K_c}{L_1^2 K_{c - 1}^{1 - 1/(c - 1)}} \left( \frac{1}{4c} \right)^{2 - 1/(c - 1)}
            \\&\qquad\qquad\ge \frac{K_c}{4 L_1^2 K_{c - 1}^{1 - 1/(c - 1)}} c^{-2}
        \; .    
    \end{align*}
    Now $K_c = \left( L_2 c \right)^c$ and $K_{c - 1} \le \left( L_2 c \right)^{c - 1}$ so we get that
    \begin{align*}
        \frac{ K_c}{4 L_1^2 K_{c - 1}^{1 - 1/(c - 1)}} c^{-2}
            \ge \frac{\left( L_2 c \right)^c}{4 L_1^2 \left( L_2 c \right)^{c - 2}} c^{-2}
            = \left( \frac{L_2}{2 L_1} \right)^2
            \ge 1
        \; .
    \end{align*}
    The last inequality follows by choosing $L_2 \ge 2 L_1$.
    Combining it all we get that
    \begin{align*}
        \frac{\tfrac{1}{4} K_c p \gamma_p^{c - 1} \frac{w(\Sigma^c)}{m}}{\tfrac{L_1^2}{16 e^2} p^2 \bar{p}^2 K_{c - 1}^2 \gamma_p^{2c - 4}}
            \ge 1
        \; .
    \end{align*}
    This implies that $\tfrac{1}{2} \sqrt{p} \sqrt{K_c \gamma_p^{c - 1} \frac{w(\Sigma^c)}{m}} \ge \tfrac{L_1}{4 e} p \bar{p} K_{c - 1} \gamma_p^{c - 2}$.

    This proves \cref{eq:simple-concentration-case-two} and combining this with \cref{eq:psi-lower-bound} we get that
    \begin{align}\begin{split}\label{eq:simple-concentration-case-two-end}
        &\Psi_p\left(L_1 e \Psi_{\bar{p}}\left(K_{c - 1} \gamma_p^{c - 2}, K_{c - 1} \frac{\gamma_p^{c - 2} w(\Sigma^c)^{1 - 1/c} w_{\infty}(\Sigma^c)^{1/c}}{m} \right), K_c \gamma_p^{c - 1} w(\Sigma^c)\right)
            \\&\qquad\qquad\le \Psi_p\left(K_c \gamma_p^{c - 1},  K_c \gamma_p^{c - 1} \frac{w(\Sigma^c)}{m}\right)
            \\&\qquad\qquad= \Psi_p\left(K_c \gamma_p^{c - 1},  K_c \gamma_p^{c - 1} \sigma_v^2\right)
        \; .
    \end{split}\end{align}

    Now combining \cref{eq:simple-concentration-decoupling}, \cref{eq:simple-concentration-sampling}, \cref{eq:simple-concentration-psi-bound}, and \cref{eq:simple-concentration-case-two-end} we get the result. 
\end{proof}

%% file: concentration/mixed-tabulation.tex

\subsection{Concentration Results for Mixed Tabulation Hashing}

In this section, we will prove two different concentration results for mixed tabulation hashing.
The first is a version of Khintchine's inequality for mixed tabulation hashing, and the proof is the simpler of the two.
The other result is a strengthening of \Cref{thm:simple-tab-moments} by using the strength of mixed tabulation hashing.

We will first introduce some notation.
\begin{definition}
    Let $h \colon \Sigma^c \to [m]$ be a mixed tabulation function with $d$ derived characters, and let $h_1 \colon \Sigma^c \to [m]$, $h_2 \colon \Sigma^c \to \Sigma$, and $h_3 \colon \Sigma^d \to [m]$ be the three simple tabulation function defining $h$, i.e., $h(x) = h_1(x) \xor h_3(h_2(x))$.

    Let $\eps_1 \colon \Sigma^c \to \set{-1, 1}$ and $\eps_3 \colon \Sigma^d \to \set{-1, 1}$ be independent simple tabulation sign functions.
    We define $\eps \colon \Sigma^c \to \set{-1, 1}$ by
    \[
        \eps(x) = \eps_1(x) \eps_3(h_2(x))
        \; .
    \]
    We say that $\eps$ is a mixed tabulation sign function associated with $h$.
\end{definition}

\begin{theorem}\label{thm:mixed-Khintchine}
    Let $\eps \colon \Sigma^c \to \set{-1, 1}$ be a mixed tabulation sign function with $d \ge 1$ derived characters, and let $w \colon \Sigma^c \to \R$ be a weight function. For all $p \ge 2$ then,
    \begin{align}
        \pnorm{\sum_{x \in \Sigma^c} w(x)\eps(x)}{p}
            \le \sqrt{e} K_c \sqrt{p} \gamma_p^{c/2} \sqrt{\sum_{x \in \Sigma^c} w(x)^2}
    \end{align}
    Here $K_c$ is as defined in \Cref{lem:simple-sum-squares} and $\gamma_p = \max\!\set{1, \tfrac{p}{\log\!\left( \abs{\Sigma} \right)}}$.
\end{theorem}
\begin{proof}
    We will prove the result for $d = 1$.
    For $d > 1$ we fix the last $d - 1$ derived characters and incorporate them into the weight function.
    This will only change the sign of the weight function for some keys, thus the result follows from the case with $d = 1$.

    We let $\eps_1 \colon \Sigma^c \to \set{-1, 1}$, $h \colon \Sigma^c \to \Sigma$, and $\eps_2 \colon \Sigma \to \set{-1, 1}$ be the three simple tabulation functions used to define $\eps$, i.e., $\eps(x) = \eps_1(x) \eps_2(h(x))$.

    We can now write,
    \begin{align*}
        \pnorm{\sum_{x \in \Sigma^c} w(x)\eps(x)}{p}
            = \pnorm{\sum_{x \in \Sigma^c} w(x) \eps_1(x) \eps_2(h(x)) }{p}
            = \pnorm{\sum_{\alpha \in \Sigma} \eps_2(\alpha) \sum_{x \in \Sigma^c} w(x) \indicator{h(x) = \alpha} \eps_1(x) }{p}
        \; .
    \end{align*}
    We fix $h$ and $\eps_1$ and use \Cref{cor:sum-of-Rademacher} to get that
    \begin{align*}
        \pnormcond{\sum_{\alpha \in \Sigma} \eps_2(\alpha) \sum_{x \in \Sigma^c} w(x) \indicator{h(x) = \alpha} \eps_1(x) }{p}{h, \eps_1}
            \le \sqrt{p} \sqrt{e \sum_{\alpha \in \Sigma} \left( \sum_{x \in \Sigma^c} w(x) \indicator{h(x) = \alpha} \eps_1(x) \right)^2 }
        \; .
    \end{align*}

    We define the value function $v \colon \Sigma^c \times \Sigma \to \R$ by $v(x, \alpha) = w(x) \indicator{\alpha = 0}$.
    We can then write,
    \begin{align*}
        \sqrt{e p} \pnorm{\sqrt{\sum_{\alpha \in \Sigma} \left( \sum_{x \in \Sigma^c} w(x) \indicator{h(x) = \alpha} \eps_1(x) \right)^2 }}{p}
            &= \sqrt{e p} \pnorm{\sum_{\alpha \in \Sigma} \left( \sum_{x \in \Sigma^c} \eps_1(x) V(x, \alpha \xor h(x)) \right)^2}{p/2}^{1/2}
        \; .
    \end{align*}
    Now we use \Cref{lem:simple-sum-squares} to get that
    \begin{align*}
        \pnorm{\sum_{\alpha \in \Sigma} \left( \sum_{x \in \Sigma^c} \eps_1(x) V(x, \alpha \xor h(x)) \right)^2}{p/2}
            &\le  K_c \gamma_{p/2}^c \sum_{x \in \Sigma^c} w(x)^2
            \\&\le  K_c \gamma_{p}^c \sum_{x \in \Sigma^c} w(x)^2
        \; .
    \end{align*}
    Putting it all together, we get that
    \begin{align*}
        \pnorm{\sum_{x \in \Sigma^c} w(x)\eps(x)}{p}
            &\le \sqrt{e p} \left( K_c \gamma_{p}^c \sum_{x \in \Sigma^c} w(x)^2 \right)^{1/2}
            \\&= \sqrt{e} K_c \sqrt{p} \gamma_{p}^{c/2} \sqrt{\sum_{x \in \Sigma^c} w(x)^2}
        \; .
    \end{align*}
\end{proof}

Before proving the next result, we will first argue that we only need the symmetric case, similarly, as we did for simple tabulation.

\begin{lemma}\label{lem:mixed-sym}
    Let $h \colon \Sigma^c \to [m]$ be a mixed tabulation function with $d$ derived characters and $v \colon \Sigma^{c} \times [m] \to \R$ a value function.
    Then for every $p \ge 2$,
    \begin{align*}
        \pnorm{\sum_{x \in \Sigma^c} \left(v(x, h(x)) - \ep{v(x, h(x))} \right)}{p}
            \le 2^{c + d} \pnorm{\sum_{x \in \Sigma^c} \eps(x) v(x, h(x))}{p}
        \; ,
    \end{align*}
    where $\eps \colon \Sigma^c \to \set{-1, 1}$ a mixed tabulation sign function associated with $h$.
\end{lemma}
\begin{proof}
    The result follows by two uses \Cref{lem:simple-sym}.
    Fixing $h_2$ and $h_3$ and using \Cref{lem:simple-sym} we get that
    \begin{align*}
        &\pnormcond{\sum_{x \in \Sigma^c} \left(v(x, h(x)) - \ep{v(x, h(x))} \right)}{p}{h_2, h_3}
            \\&\qquad\qquad= \pnormcond{\sum_{x \in \Sigma^c} \left(v(x, h_1(x) \xor h_3(h_2(x))) - \ep{v(x, h_1(x) \xor h_3(h_2(x)))} \right)}{p}{h_2, h_3}
            \\&\qquad\qquad\le 2^c \pnormcond{\sum_{x \in \Sigma^c} \eps_1(x) \left(v(x, h_1(x) \xor h_3(h_2(x))) - \ep{v(x, h_1(x) \xor h_3(h_2(x)))} \right)}{p}{h_2, h_3}
        \; .
    \end{align*}
    Now we fix $h_1, h_2$, and $\eps_1$ and use \Cref{lem:simple-sym},
    \begin{align*}
        &2^c \pnormcond{\sum_{x \in \Sigma^c} \eps_1(x) \left(v(x, h_1(x) \xor h_3(h_2(x))) - \ep{v(x, h_1(x) \xor h_3(h_2(x)))} \right)}{p}{h_1, h_2, \eps_1}
            \\&\qquad\qquad\le 2^{c + d} \pnormcond{\sum_{x \in \Sigma^c} \eps_1(x) \eps_3(h) v(x, h_1(x) \xor h_3(h_2(x))) }{p}{h_1, h_2, \eps_1}
            \\&\qquad\qquad= 2^{c + d} \pnormcond{\sum_{x \in \Sigma^c} \eps(x) v(x, h(x))}{p}{h_1, h_2, \eps_1}
        \; .
    \end{align*}
\end{proof}

We can now prove the concentration result for mixed tabulation.
The proof is very similar to the proof of \Cref{thm:simple-tab-moments} and is again quite long and technical.

\restateThmMixedTabMoments

\begin{proof}
    We will prove the result for $d = 1$.
    For $d > 1$ we fix the last $d - 1$ derived characters and incorporate them into the value function.
    This does not change the variance and the result then follows from the case with $d = 1$.

    We can assume without loss of generality that $M_v = 1$.

    We let $h_1 \colon \Sigma^c \to [m]$, $h_2 \colon \Sigma^c \to \Sigma$, and $h_3 \colon \Sigma \to [m]$ be the three simple tabulation functions used to define $h$, i.e., $h(x) = h_1(x) \xor h_3(h_2(x))$.
    Similarly, we let $\eps_1 \colon \Sigma^c \to \set{-1, 1}$ and $\eps_3 \colon \Sigma \to \set{-1, 1}$ be the two simple tabulation sign functions used to define $\eps$, i.e., $\eps(x) = \eps_1(x) \eps_3(h_2(x))$.

    We can then write,
    \begin{align*}
        V^{\text{mixed}}_v
            &= \sum_{x \in \Sigma^c} \eps(x) v(x, j \xor h(x))
            \\&= \sum_{x \in \Sigma^c} \eps_1(x) \eps_3(h_2(x)) v(x, j \xor h_1(x) \xor h_3(h_2(x)))
            \\&= \sum_{\alpha \in \Sigma} \eps_3(\alpha)
                \sum_{x \in \Sigma^c} \eps_1(x) \indicator{h_2(x) = \alpha} v(x, j \xor h_1(x) \xor h_3(\alpha))
            \; .
    \end{align*}
    We define the value function $v_{h} \colon \Sigma \times [m] \to \R$ by $v_{h}(\alpha, j) = \sum_{x \in \Sigma^c} \eps_1(x) \indicator{h_2(x) = \alpha} v(x, j \xor h_1(x))$.
    This allows us to write,
    \begin{align*}
        V^{\text{mixed}}_v
            &= \sum_{\alpha \in \Sigma} \eps_3(\alpha) v_{h}(\alpha, j \xor h_3(\alpha))
            \; .
    \end{align*}

    If we fix $h_1$ and $h_2$ then \Cref{thm:sampling-moments} give us that
    \begin{align}\label{eq:mixed-split}
        \pnormcond{V^{\text{mixed}}_v}{p}{h_1, h_2}
            \le 8e \Psi_p\left(\max_{\alpha \in \Sigma, j \in [m]} \abs{v_h(\alpha, j)}, \frac{\sum_{\alpha \in \Sigma, j \in [m]} v_h(\alpha, j)^2}{m} \right)
        \; .
    \end{align}
    As in the proof of \Cref{thm:simple-tab-moments} we will use \Cref{lem:fn-moment}.
    \begin{align}\label{eq:mixed-psi-norm-bound}
        &\pnorm{\Psi_p\left(\max_{\alpha \in \Sigma, j \in [m]} \abs{v_h(\alpha, j)}, \frac{\sum_{\alpha \in \Sigma, j \in [m]} v_h(\alpha, j)^2}{m} \right)}{p}
            \\&\qquad\qquad\qquad\qquad\qquad\le \sqrt{2} \Psi_p\left(\pnorm{\max_{\alpha \in \Sigma, j \in [m]} \abs{v_h(\alpha, j)}}{p}, \frac{1}{m} \pnorm{\sum_{\alpha \in \Sigma, j \in [m]} v_h(\alpha, j)^2}{p/2} \right)
    \end{align}
    Now we want to bound $\pnorm{\max_{\alpha \in \Sigma, j \in [m]} \abs{v_h(\alpha, j)}}{p}$ and $\pnorm{\sum_{\alpha \in \Sigma, j \in [m]} v_h(\alpha, j)^2}{p/2}$.
    We define the value function $v' \colon \Sigma^c \times \left([m] \times \Sigma \right) \to \R$ by
    $v'(x, (j, \alpha)) = \indicator{\alpha = 0} v(x, j)$.
    We then get that 
    \[
        v_h(\alpha, j) = \sum_{x \in \Sigma^c} \eps_1(x) v'(x, (j \xor h_1(x), \alpha \xor h_2(x)))
        \; .
    \]
    Clearly, we have that the support of $v'$ is at most $m$ for all $x \in \Sigma^c$, thus $\frac{\norm{v'[x]}{1}^2}{\norm{v'[x]}{2}^2} \le m$ for all $x \in \Sigma^c$.

    We define $\bar{p} = \max\!\set{q, \log(m \abs{\Sigma})}$.
    This implies that $\gamma_p = \frac{\bar{p}}{\log\!\left(e \abs{\Sigma} \right)}$.

    We can now bound the moments of $\max_{\alpha \in \Sigma, j \in [m]} \abs{v_h(\alpha, j)}$ by using \Cref{thm:simple-tab-moments}.
    \begin{align*}
        \pnorm{\max_{\alpha \in \Sigma, j \in [m]} \abs{v_h(\alpha, j)}}{p}
            &\le \pnorm{\max_{\alpha \in \Sigma, j \in [m]} \abs{v_h(\alpha, j)}}{\bar{p}}
            \\&\le \left(\sum_{\alpha \in \Sigma, j \in [m]} 
                \pnorm{v_h(\alpha, j)}{\bar{p}}^{\bar{p}} \right)^{1/\bar{p}}
            \\&\le \left(\sum_{\alpha \in \Sigma, j \in [m]} 
                \pnorm{\sum_{x \in \Sigma^c} \eps_1(x) v'(x, (j \xor h_1(x), \alpha \xor h_2(x)))}{\bar{p}}^{\bar{p}} \right)^{1/\bar{p}}
            \\&\le \left(
                \sum_{\alpha \in \Sigma, j \in [m]}
                    M_1^{\bar{p}} \Psi_{\bar{p}}\left(L^{(1)}_c \left(\gamma'_p\right)^{c - 1},  L^{(1)}_c \left(\gamma'_p\right)^{c - 1} \sigma_{v'}^2 \right)^{\bar{p}}
            \right)^{1/\bar{p}}
            \\&= \left(m \abs{\Sigma} \right)^{1/\bar{p}} 
                M_1 \Psi_{\bar{p}}\left(L^{(1)}_c \left(\gamma'_p\right)^{c - 1},  L^{(1)}_c \left(\gamma'_p\right)^{c - 1} \sigma_{v'}^2 \right)
            \\&= \left(m \abs{\Sigma} \right)^{1/\bar{p}} 
                M_1 \Psi_{\bar{p}}\left(L^{(1)}_c \left(\gamma'_p\right)^{c - 1},  L^{(1)}_c \left(\gamma'_p\right)^{c - 1} \frac{\sum_{x \in \Sigma^c} \norm{v[x]}{2}^2}{m \abs{\Sigma}} \right)
            \\&\le  e M_1 \Psi_{\bar{p}}\left(L^{(1)}_c \left(\gamma'_p\right)^{c - 1},  L^{(1)}_c \left(\gamma'_p\right)^{c - 1} \frac{\sigma_v^2}{\abs{\Sigma}} \right)
        \; ,
    \end{align*}
    where $L^{(1)}_c$ is a constant depending on $c$ as given by \Cref{thm:simple-tab-moments}, $M_1$ is universal constant, and 
    \[
        \gamma'_p = \max\!\set{\frac{\log(m \abs{\Sigma}) + \log\!\left( \frac{\sum_{x \in \Sigma^c} \norm{v'[x]}{2}^2}{\max_{x \in \Sigma^c} \norm{v'[x]}{2}^2} \right)/c}{\log\!\left(e \abs{\Sigma} \right)}, \frac{p}{\log \left( e \abs{\Sigma} \right)}}
    \]
    We note that $\gamma'_p \le 2 \gamma_p$ since,
    \begin{align*}
        \gamma'_p
            &= \max\!\set{\frac{\log(m \abs{\Sigma}) + \log\!\left( \frac{\sum_{x \in \Sigma^c} \norm{v'[x]}{2}^2}{\max_{x \in \Sigma^c} \norm{v'[x]}{2}^2} \right)/c}{\log\!\left( e \abs{\Sigma} \right)}, \frac{p}{\log\!\left( e \abs{\Sigma} \right)}}
            \\&\le \max\!\set{\frac{\log(m \abs{\Sigma}) + \log(m \abs{\Sigma}^c)/c}{\log \left( e \abs{\Sigma} \right)}, \frac{p}{\log \left( e \abs{\Sigma} \right)}}
            \\&\le 2 \max\!\set{\frac{\log(m \abs{\Sigma})}{\log \left( e \abs{\Sigma} \right)}, \frac{p}{\log \left( e \abs{\Sigma} \right)}}
            \\&= 2 \gamma_p
        \; .
    \end{align*}
    So we have that
    \begin{align}\label{eq:mixed-concentration-max}
        \pnorm{\max_{\alpha \in \Sigma, j \in [m]} \abs{v_h(\alpha, j)}}{p}
            \le e M_1 \Psi_{\bar{p}}\left(2L^{(1)}_c \gamma_p^{c - 1},  2 L^{(1)}_c \gamma_p^{c - 1} \frac{\sigma^2}{\abs{\Sigma}} \right)
        \; .
    \end{align}

    We bound the moments of $\sum_{\alpha \in \Sigma, j \in [m]} v_h(\alpha, j)^2$ by using \Cref{lem:simple-sum-squares}.
    We get that for all $q \ge 2$,
    \begin{align}\begin{split}\label{eq:mixed-concentration-squares}
        \pnorm{\sum_{\alpha \in \Sigma, j \in [m]} v_h(\alpha, j)^2}{p}
            &= \pnorm{ \sum_{\alpha \in \Sigma, j \in [m]} \left( \sum_{x \in \Sigma^c} v'(x, (j \xor h_1(x), \alpha \xor h_2(x))) \right)^2}{p}
            \\&\le  L^{(2)}_c \max\!\set{1, \left( \frac{p}{\log\!\left(e \abs{\Sigma} \right)} \right)^c } \sum_{x \in \Sigma^{c}} \norm{v'[x]}{2}^2
            \\&= L^{(2)}_c \gamma_p^c \sum_{x \in \Sigma^{c}} \norm{v[x]}{2}^2
            \\&\le K_c \gamma_p^c \sum_{x \in \Sigma^{c}} \norm{v[x]}{2}^2
        \; .
    \end{split}\end{align}
    Where $L^{(2)}_c$ is constant depending on $c$ as given by \Cref{lem:simple-sum-squares}.
  
    Now combining \cref{eq:mixed-psi-norm-bound}, \cref{eq:mixed-concentration-max}, and \cref{eq:mixed-concentration-squares}, and we get that
    \begin{align}\begin{split}\label{eq:mixed-concentration-psi-bound}
        &\pnorm{\Psi_p\left(\max_{\alpha \in \Sigma, j \in [m]} \abs{v_h(\alpha, j)}, \frac{\sum_{\alpha \in \Sigma, j \in [m]} v_h(\alpha, j)^2}{m} \right)}{p}
            \\&\qquad\qquad\qquad\le e \Psi_p\left(e M_1\Psi_{\bar{p}}\left(2L^{(1)}_c \gamma_p^{c - 1},  2 L^{(1)}_c \gamma_p^{c - 1} \frac{\sigma^2}{\abs{\Sigma}} \right), K_c \gamma_p^c \sigma^2 \right)
        \; .
    \end{split}\end{align}

    We will now consider three different cases.

    \paragraph*{Case 1. $\sigma^2 \le (2 e^{-3} L_c^{(1)}) \bar{p} \gamma_p^{c - 1} \left(e \abs{\Sigma} \right)^{1 - \tfrac{1}{4(2c + 1)}}$.}
    In this case we will show that
    \begin{align}\label{eq:mixed-concentration-case-one}
        e M_1\Psi_{\bar{p}}\left(2L^{(1)}_c \gamma_p^{c - 1},  2 L^{(1)}_c \gamma_p^{c - 1} \frac{\sigma^2}{\abs{\Sigma}} \right)
            \le K_c \gamma_p^c
        \; .
    \end{align}
    We first see that
    \begin{align}\begin{split}\label{eq:mixed-concentration-case-one-calc}
        \frac{2 L^{(1)}_c \gamma_p^{c - 1} \frac{\sigma^2}{\abs{\Sigma}}}{\left( 2L^{(1)}_c \gamma_p^{c - 1} \right)^2}
            &= \frac{\sigma^2}{2L^{(1)}_c \gamma_p^{c - 1} \abs{\Sigma}}
            \\&\le \frac{(2 e^{-3} L_c^{(1)}) \bar{p} \gamma_p^{c - 1} \left(e \abs{\Sigma} \right)^{1 - \tfrac{1}{4(2c + 1)}}}{2L^{(1)}_c \gamma_p^{c - 1} \abs{\Sigma}}
            \\&\le e^{-2} \bar{p} \abs{\Sigma}^{- \tfrac{1}{4(2c + 1)}}
            \\&\le e^{-2} \bar{p}
        \; .
    \end{split}\end{align}
    
    By \cref{eq:psi-far-out} we get that
    \begin{align*}
        e M_1\Psi_{\bar{p}}\left(2L^{(1)}_c \gamma_p^{c - 1},  2 L^{(1)}_c \gamma_p^{c - 1} \frac{\sigma^2}{\abs{\Sigma}} \right)
            &\le e M_1 \frac{\bar{p}}{\log\!\left( \frac{\bar{p} \left(2L^{(1)}_c \gamma_p^{c - 1}\right)^2}{2 L^{(1)}_c \gamma_p^{c - 1} \frac{\sigma^2}{\abs{\Sigma} } } \right) } 2L^{(1)}_c \gamma_p^{c - 1}
            \\&\le 2 e M_1 L_c^{(1)} \frac{\bar{p}}{\log\!\left( e^2 \abs{\Sigma}^{\tfrac{1}{4(2c + 1)}}  \right)} \gamma_p^{c - 1}
            \\&\le 8 e M_1 L_c^{(1)} (2c + 1) \gamma_p^{c}
            \\&\le 24 e c M_1 L_c^{(1)} \gamma_p^{c}
            \\&\le K_c \gamma_p^{c}
        \; .
    \end{align*}
    Here we have used that $\frac{\bar{p} \left(2L^{(1)}_c \gamma_p^{c - 1}\right)^2}{2 L^{(1)}_c \gamma_p^{c - 1} \frac{\sigma^2}{\abs{\Sigma} } } \ge e^2 \abs{\Sigma}^{\tfrac{1}{4(2c + 1)}}$ which follows from \cref{eq:mixed-concentration-case-one-calc} and that $K_c \ge 24 e c M_1 L_c^{(1)}$.

    Now combining \cref{eq:mixed-split}, \cref{eq:mixed-concentration-psi-bound}, and \cref{eq:mixed-concentration-case-one} we get the result.

    \paragraph*{Case 2. $\sigma^2 > (2 e^{-3} L_c^{(1)}) \bar{p} \gamma_p^{c - 1} \left(e \abs{\Sigma} \right)^{1 - \tfrac{1}{4(2c + 1)}}$ and $p \le \abs{\Sigma}^{1 - \tfrac{1}{2(2c + 1)}}$.}
    We will show that
    \begin{align}\label{eq:mixed-concentration-case-two}
        \Psi_p\left(e M_1\Psi_{\bar{p}}\left(2L^{(1)}_c \gamma_p^{c - 1},  2 L^{(1)}_c \gamma_p^{c - 1} \frac{\sigma^2}{\abs{\Sigma}} \right), K_c \gamma_p^c \sigma^2 \right)
            &\le \sqrt{p} \sqrt{K_c \gamma_p^c \sigma^2}
        \; .
    \end{align}

    We use \cref{eq:psi-bernstein} to get that
    \begin{align*}
        \Psi_{\bar{p}}\left(2L^{(1)}_c \gamma_p^{c - 1},  2 L^{(1)}_c \gamma_p^{c - 1} \frac{\sigma^2}{\abs{\Sigma}} \right)
            &\le \max\!\set{\tfrac{1}{2} \sqrt{\bar{p} 2 L^{(1)}_c \gamma_p^{c - 1} \frac{\sigma^2}{\abs{\Sigma}}},
                \tfrac{L^{(1)}_c}{e} \bar{p} \gamma_p^{c - 1} }
        \; .
    \end{align*}
    We again apply \cref{eq:psi-bernstein} to obtain that
    \begin{align*}
        &\Psi_p\left(e M_1\Psi_{\bar{p}}\left(2L^{(1)}_c \gamma_p^{c - 1},  2 L^{(1)}_c \gamma_p^{c - 1} \frac{\sigma^2}{\abs{\Sigma}} \right), K_c \gamma_p^c \sigma^2 \right)
            \\&\qquad\qquad\le \max\!\set{\tfrac{1}{2} \sqrt{p} \sqrt{ K_c \gamma_p^c \sigma^2}, \tfrac{M_1}{2} \Psi_{\bar{p}}\left(2L^{(1)}_c \gamma_p^{c - 1},  2 L^{(1)}_c \gamma_p^{c - 1} \frac{\sigma^2}{\abs{\Sigma}} \right)}
        \; .
    \end{align*}
    Combining the two estimates give us that
    \begin{align*}
        &\Psi_p\left(e M_1\Psi_{\bar{p}}\left(2L^{(1)}_c \gamma_p^{c - 1},  2 L^{(1)}_c \gamma_p^{c - 1} \frac{\sigma^2}{\abs{\Sigma}} \right), K_c \gamma_p^c \sigma^2 \right)
            \\&\qquad\qquad\le \max\!\set{\tfrac{1}{2} \sqrt{p} \sqrt{ K_c \gamma_p^c \sigma^2}, 
            \tfrac{M_1 \sqrt{L^{(1)}_c}}{2\sqrt{2}} p \sqrt{\bar{p} \gamma_p^{c - 1} \frac{\sigma^2}{\abs{\Sigma}}},
            \tfrac{M_1 L^{(1)}_c}{2e} p \bar{p} \gamma_p^{c - 1}}
    \end{align*}
    We will show that the max is equal to $\tfrac{1}{2} \sqrt{p} \sqrt{ K_c \gamma_p^c \sigma^2}$ which will show \cref{eq:mixed-concentration-case-two}.

    First we show that $\tfrac{1}{2} \sqrt{p} \sqrt{ K_c \gamma_p^c \sigma^2} \ge \tfrac{M_1 \sqrt{L^{(1)}_c}}{2\sqrt{2}} p \sqrt{\bar{p} \gamma_p^{c - 1} \frac{\sigma^2}{\abs{\Sigma}}}$.
    We note that this equivalent with showing that $\frac{\tfrac{K_c}{4} p \gamma_p^{c} \sigma^2}{\tfrac{M_1^2 L^{(1)}_c}{8} p^2 \bar{p} \gamma_p^{c - 1} \frac{\sigma^2}{\abs{\Sigma}}} \ge 1$,
    \begin{align*}
        \frac{\tfrac{K_c}{4} p \gamma_p^{c} \sigma^2}{\tfrac{M_1^2 L^{(1)}_c}{8} p^2 \bar{p} \gamma_p^{c - 1} \frac{\sigma^2}{\abs{\Sigma}}}
            &= \frac{2 K_c}{M_1^2 L^{(1)}_c} \cdot \frac{\gamma_p\abs{\Sigma}}{p \bar{p}}
            = \frac{2 K_c}{M_1^2 L^{(1)}_c} \cdot \frac{\abs{\Sigma}}{p \log\!\left(e\abs{\Sigma}\right)}
            \ge \frac{2 K_c}{M_1^2 L^{(1)}_c} \cdot \frac{\abs{\Sigma}^{\tfrac{1}{2(2c + 1)}}}{\log\!\left(e\abs{\Sigma}\right)}
        \; .
    \end{align*}
    Now by using \Cref{lem:log-vs-poly} we get that
    \begin{align*}
        \frac{2 K_c}{M_1^2 L^{(1)}_c} \cdot \frac{\abs{\Sigma}^{\tfrac{1}{2(2c + 1)}}}{\log\!\left(e\abs{\Sigma}\right)}
            &\ge \frac{2 K_c}{e M_1^2 L^{(1)}_c} \cdot \frac{\left(e \abs{\Sigma} \right)^{\tfrac{1}{2(2c + 1)}}}{\log\!\left(e\abs{\Sigma}\right)}
            \ge \frac{2 K_c}{e M_1^2 L^{(1)}_c} \left( \frac{e}{2(2c + 1)} \right)
            = \frac{K_c}{3 M_1^2 L^{(1)}_c c}
        \; .
    \end{align*}
    Now $K_c = \left( L_2 c \right)^c$ and $L^{(1)}_c = \left( T_1 c \right)^c$ where $T_1$ is universal constant determined by \Cref{thm:simple-tab-moments}.
    Now choosing $L_2$ large enough we get that so we get that
    \begin{align*}
        \frac{K_c}{3 M_1^2 L^{(1)}_c c}
            = \frac{L_2^c}{3 M_1^2 T_1^{c - 1}}
            \ge 1
        \; .
    \end{align*}
    Combining it all we get that
    \begin{align*}
        \frac{\tfrac{K_c}{4} p \gamma_p^{c} \sigma^2}{\tfrac{M_1^2 L^{(1)}_c}{8} p^2 \bar{p} \gamma_p^{c - 1} \frac{\sigma^2}{\abs{\Sigma}}}
            \ge 1
        \; .
    \end{align*}
    This implies that $\tfrac{1}{2} \sqrt{p} \sqrt{ K_c \gamma_p^c \sigma^2} \ge \tfrac{M_1 \sqrt{L^{(1)}_c}}{2\sqrt{2}} p \sqrt{\bar{p} \gamma_p^{c - 1} \frac{\sigma^2}{\abs{\Sigma}}}$ as we wanted.

    Next we show that $\tfrac{1}{2} \sqrt{p} \sqrt{ K_c \gamma_p^c \sigma^2} \ge \tfrac{M_1 L^{(1)}_c}{2e} p \bar{p} \gamma_p^{c - 1}$.
    Again we note that this equivalent with showing that $\frac{\tfrac{K_c}{4} p \gamma_p^{c} \sigma^2}{\tfrac{M_1^2 \left(L_c^{(1)} \right)^2}{4 e^2} p^2 \bar{p}^2 \gamma_p^{2c - 2}} \ge 1$,
    \begin{align*}
        \frac{\tfrac{K_c}{4} p \gamma_p^{c} \sigma^2}{\tfrac{M_1^2 \left(L_c^{(1)} \right)^2}{4 e^2} p^2 \bar{p}^2 \gamma_p^{2c - 2}}
            &= \frac{e^2 K_c}{\left( M_1 L_c^{(1)} \right)^2} \frac{\sigma^2}{p \bar{p}^2 \gamma_p^{c - 2}}
            \\&\ge \frac{e^2 K_c}{\left( M_1 L_c^{(1)} \right)^2} \frac{(2 e^{-3} L_c^{(1)}) \bar{p} \gamma_p^{c - 1} \left( e\abs{\Sigma} \right)^{1 - \tfrac{1}{4(2c + 1)}}}{p \bar{p}^2 \gamma_p^{c - 2}}
            \\&\ge \frac{2 K_c}{e M_1^2 L_c^{(1)}} \frac{\left( e\abs{\Sigma} \right)^{1 - \tfrac{1}{4(2c + 1)}}}{p \log\!\left(e\abs{\Sigma}\right)}
            \\&\ge \frac{2 K_c}{e M_1^2 L_c^{(1)}} \frac{\left( e\abs{\Sigma} \right)^{\tfrac{1}{4(2c + 1)}}}{\log\!\left(e\abs{\Sigma}\right)}
        \; .
    \end{align*}
    Again we use \Cref{lem:log-vs-poly} to obtain that
    \begin{align*}
        \frac{2 K_c}{e M_1^2 L_c^{(1)}} \frac{\left( e\abs{\Sigma} \right)^{\tfrac{1}{4(2c + 1)}}}{\log\!\left(e\abs{\Sigma}\right)}
            \ge \frac{2 K_c}{e M_1^2 L_c^{(1)}} \left( \frac{e}{4(2c + 1)} \right)
            \ge \frac{K_c}{6 e M_1^2 L_c^{(1)} c}
        \; .    
    \end{align*}
    Now $K_c = \left( L_2 c \right)^c$ and $L^{(1)}_c = \left( T_1 c \right)^c$ where $T_1$ is universal constant determined by \Cref{thm:simple-tab-moments}.
    Now choosing $L_2$ large enough we get that so we get that
    \begin{align*}
        \frac{K_c}{6 e M_1^2 L_c^{(1)} c}
            = \frac{L_2^c}{6 e M_1^2 T_1^{c - 1}}
            \ge 1
        \; .
    \end{align*}
    Combining it all we get that
    \begin{align*}
        \frac{\tfrac{K_c}{4} p \gamma_p^{c} \sigma^2}{\tfrac{M_1^2 \left(L_c^{(1)} \right)^2}{4 e^2} p^2 \bar{p}^2 \gamma_p^{2c - 2}}
            \ge 1
        \; .
    \end{align*}
    This implies that $\tfrac{1}{2} \sqrt{p} \sqrt{ K_c \gamma_p^c \sigma^2} \ge \tfrac{M_1 L^{(1)}_c}{2e} p \bar{p} \gamma_p^{c - 1}$.

    This proves \cref{eq:mixed-concentration-case-two} and combining this with \cref{eq:psi-lower-bound} we get that
    \begin{align}\begin{split}\label{eq:mixed-concentration-case-two-end}
        \Psi_p\left(e M_1\Psi_{\bar{p}}\left(2L^{(1)}_c \gamma_p^{c - 1},  2 L^{(1)}_c \gamma_p^{c - 1} \frac{\sigma^2}{\abs{\Sigma}} \right), K_c \gamma_p^c \sigma^2 \right)
            \le \Psi_p\left(K_c \gamma_p^{c},  K_c \gamma_p^{c} \sigma^2\right)
        \; .
    \end{split}\end{align}

    Now combining \cref{eq:mixed-split}, \cref{eq:mixed-concentration-psi-bound}, and \cref{eq:mixed-concentration-case-two-end} we get the result.

    \paragraph*{Case 3. $\sigma^2 > (2 e^{-3} L_c^{(1)}) \bar{p} \gamma_p^{c - 1} \left( e\abs{\Sigma} \right)^{1 - \tfrac{1}{4(2c + 1)}}$ and $p > \left(e\abs{\Sigma}\right)^{1 - \tfrac{1}{2(2c + 1)}}$.}
    In this case we will exploit that $\abs{V^{\text{mixed}}_v} \le \abs{\Sigma}^c$.
    This implies that $\pnorm{V^{\text{mixed}}_v}{p} \le \abs{\Sigma}^c$, so if we can prove that
    \begin{align}\label{eq:mixed-concentration-case-three}
        L_1 \Psi_p\left(K_c \gamma_p^{c},  K_c \gamma_p^{c} \sigma^2 \right)
            \ge \abs{\Sigma}^c
        \; ,
    \end{align}
    then the result follows.

    We use \cref{eq:psi-lower-bound} to get that
    \begin{align*}
        \Psi_p\left(K_c \gamma_p^{c},  K_c \gamma_p^{c} \sigma^2 \right)
            &\ge \sqrt{p} \sqrt{K_c \gamma_p^{c} \sigma^2}
            \\&\ge \sqrt{p^2 K_c (2 e^{-3} L_c^{(1)}) \gamma_p^{2c - 1} \left(e \abs{\Sigma} \right)^{1 - \tfrac{1}{4(2c + 1)}}}
            \\&\ge \sqrt{2 e^{-3} K_c L_c^{(1)}} \frac{\left(e \abs{\Sigma}\right)^{\tfrac{2c + 1}{2} \cdot \left(1 - \tfrac{1}{2(2c + 1)}\right)}}{\log\!\left(e\abs{\Sigma}\right)^{\tfrac{2c - 1}{2}}}
            \\&\ge \sqrt{2 e^{-3} K_c L_c^{(1)}} \frac{\left(e \abs{\Sigma}\right)^{c + \tfrac{1}{4}}}{\log\!\left(e\abs{\Sigma}\right)^{\tfrac{2c - 1}{2}}}
    \end{align*}
    Now we use \Cref{lem:log-vs-poly} to get that
    \[
        \frac{\left(e \abs{\Sigma}\right)^{1/4}}{\log\!\left(e\abs{\Sigma}\right)^{\tfrac{2c - 1}{2}}}
            \ge \left( \frac{e}{2(2c - 1)} \right)^{\tfrac{2c - 1}{2}}
            \ge \left( \frac{e}{4c} \right)^{\tfrac{2c - 1}{2}}
        \; .
    \]
    Combining this we get that
    \begin{align*}
        \Psi_p\left(K_c \gamma_p^{c},  K_c \gamma_p^{c} \sigma^2 \right)
            &\ge \sqrt{2 e^{-3} K_c L_c^{(1)}} \left( e \abs{\Sigma} \right)^c \left( \frac{e}{4c} \right)^{\tfrac{2c - 1}{2}}
            \ge \sqrt{2 e^{-3}} \abs{\Sigma}^c \sqrt{K_c L_c^{(1)} \left( \frac{e}{4c} \right)^{2c - 1}}
        \; .
    \end{align*}
    We have that $K_c L_c^{(1)} = c^{2c - 1} T_1^c T_2^{c - 1}$ where $T_1$ is a universal constant and $T_2$ is a universal determined by \Cref{thm:simple-tab-moments}.
    Choosing $T_1$ large enough we get that $\sqrt{K_c L_c^{(1)} \left( \frac{e}{4c} \right)^{2c - 1}} \ge \tfrac{1}{L_1}$.
    Thus \cref{eq:mixed-concentration-case-three} follows and the result is proven. 
\end{proof}

%% file: concentration/lower-bound.tex

Similarly, as in the proof of \Cref{lem:psi-relation-poisson}, we will use the following result by Latała~\cite{Latala1997} that gives a tight bound for $p$-norms of sums of independent and identically distributed symmetric variables.
\begin{lemma}[Latała~\cite{Latala1997}]\label{lem:Latala}
    If $X, (X_i)_{i \in [n]}$ are independent and identically distributed symmetric variables and $p \ge 2$ then,
    \begin{align*}
        \pnorm{\sum_{i \in [n]} X_i}{p}
            \le K_1 \sup\setbuilder{\frac{p}{s}\left( \frac{n}{p} \right)^{1/s} \pnorm{X}{s}}{\max\!\set{2, \tfrac{p}{n}} \le s \le p}
        \; ,
    \end{align*}
    and 
    \begin{align*}
        \pnorm{\sum_{i \in [n]} X_i}{p}
            \ge K_2 \sup\setbuilder{\frac{p}{s}\left( \frac{n}{p} \right)^{1/s} \pnorm{X}{s}}{\max\!\set{2, \tfrac{p}{n}} \le s \le p}
        \; .
    \end{align*}
    Here $K_1$ and $K_2$ are universal constants.
\end{lemma}

We start by proving the lower bound for uniformly random hash functions which will serve as a base for the other lower bounds.
We show the following lemma.

\begin{lemma}
    Let $h \colon U \to [m]$ be a uniformly random function and let $v \colon U \times [m] \to \R$ be a value function defined by,
    \begin{align*}
        v(x, j) = \begin{cases}
            1 - \tfrac{1}{m} &\text{if $j = 0$} \\
            - \tfrac{1}{m}   &\text{otherwise}
        \end{cases}
        \; ,
    \end{align*}
    for all $x \in U$.
    Then for all $2 \le p \le L_1 \abs{U} \log(m)$,
    \begin{align*}
        \pnorm{\sum_{x \in U} v(x, h(x)) }{p}
            \ge L_2 \Psi_p\left(M_v, \sigma_v^2 \right)
        \; ,
    \end{align*}
    where $L_1$ and $L_2$ is a universal constant.
\end{lemma}
It is easy to check that this implies \Cref{thm:simple-tab-moments-lower-bound}.
\begin{proof}
    We will consider the value function, $v \colon U \times [m] \to \R$, defined by,
    \begin{align*}
        v(x, j) = \begin{cases}
            1 - \tfrac{1}{m} &\text{if $j = 0$} \\
            - \tfrac{1}{m}   &\text{otherwise}
        \end{cases}
        \; ,
    \end{align*}
    We then have that $\sigma_v^2 = \abs{U}\tfrac{1}{m}\left(1 - \tfrac{1}{m}\right)$ and $M_v = 1$.
    The goal is then to prove that
    \begin{align}\label{eq:sampling-moments-lower-bound-goal}
        \pnorm{\sum_{x \in U} v(x, h(x))}{p}
            \ge L_2 \Psi_p(1, \sigma_v^2)
        \; .
    \end{align}

    We then use \Cref{lem:sampling-symmetrize} to get that
    \begin{align*}
        \pnorm{\sum_{x \in U} v(x, h(x))}{p}
            \ge \tfrac{1}{2}\pnorm{\sum_{x \in U} \eps(x) v(x, h(x))}{p}
        \; ,
    \end{align*}
    where $\eps \colon U \to \set{-1, 1}$ is uniformly random sign function.
    We can now use \Cref{lem:Latala},
    \begin{align*}
        \pnorm{\sum_{x \in U} \eps(x) v(x, h(x))}{p}
            \ge K_2 \sup\setbuilder{\frac{p}{s}\left( \frac{\abs{U}}{p} \right)^{1/s} \pnorm{v(x, h(x))}{s}}{\max\!\set{2, \tfrac{p}{\abs{U}}} \le s \le p}
        \; .
    \end{align*}
    We will argue that $\pnorm{v(x, h(x))}{s} \ge \tfrac{1}{2}\left(\tfrac{1}{m}\left(1 - \tfrac{1}{m}\right)\right)^{1/s}$ for all $s \ge 2$,
    \begin{align*}
        \pnorm{v(x, h(x))}{s}
            &= \left( \tfrac{1}{m} \left(1 - \tfrac{1}{m}\right)^{s} + \left(1 - \tfrac{1}{m}\right) \left(\tfrac{1}{m}\right)^s \right)^{1/s}
            \\&= \left( \tfrac{1}{m}\left(1 - \tfrac{1}{m}\right) \right)^{1/s} \left( \left(1 - \tfrac{1}{m}\right)^{s - 1} + \left(\tfrac{1}{m}\right)^{s - 1} \right)^{1/s}
            \\&\ge \left( \tfrac{1}{m}\left(1 - \tfrac{1}{m}\right) \right)^{1/s} \left( \max\!\set{\left(1 - \tfrac{1}{m}\right)^{s - 1}, \left(\tfrac{1}{m}\right)^{s - 1}} \right)^{1/s}
            \\&\ge \tfrac{1}{2}\left( \tfrac{1}{m}\left(1 - \tfrac{1}{m}\right) \right)^{1/s}
        \; .
    \end{align*}
    This implies that
    \begin{align*}
        \pnorm{\sum_{x \in U} \eps(x) v(x, h(x))}{p}
            &\ge \frac{K_2}{2} \sup\setbuilder{\frac{p}{s}\left( \frac{\abs{U} \tfrac{1}{m}\left(1 - \tfrac{1}{m}\right)}{p} \right)^{1/s}}{\max\!\set{2, \tfrac{p}{\abs{U}}} \le s \le p}
            \\&= \frac{K_2}{2} \sup\setbuilder{\frac{p}{s}\left( \frac{\sigma_v^2}{p} \right)^{1/s}}{\max\!\set{2, \tfrac{p}{\abs{U}}} \le s \le p}
        \; .
    \end{align*}
    It is easy to see that the function $s \mapsto \frac{p}{s}\left( \frac{\sigma_v^2}{p} \right)^{1/s}$ is maximized at $s^* = \log\!\left( \frac{p}{\sigma_v^2} \right)$.
    It is easy to check that $\log\!\left( \frac{p}{\sigma_v^2} \right) \ge \frac{p}{\abs{U}}$ when $p \le L_1 \abs{U} \log(m)$ for a sufficiently large universal constant $L_1$.
    We then get that
    \begin{align*}
        \pnorm{\sum_{x \in U} \eps(x) v(x, h(x))}{p}
            &\ge \frac{K_2}{2} \sup\setbuilder{\frac{p}{s}\left( \frac{\sigma_v^2}{p} \right)^{1/s}}{\max\!\set{2, \tfrac{p}{\abs{U}}} \le s \le p}
            \\&= \frac{K_2}{2} \sup\setbuilder{\frac{p}{s}\left( \frac{\sigma_v^2}{p} \right)^{1/s}}{2 \le s \le p}
            \\&= \frac{K_2}{2} \Psi_p(1, \sigma_v^2)
        \; .        
    \end{align*}
    The last equality follows by \cref{eq:psi-relation}.
    This proves \cref{eq:sampling-moments-lower-bound-goal} and finishes the proof.
\end{proof}

Now we are ready to prove the lower bound for simple tabulation hashing.

\restateThmSimpleTabMomentsLowerBound
\begin{proof}
    We will define a value function, $v \colon \Sigma^c \times [m] \to \R$, for which $\min_{x \in \Sigma^c} \frac{\norm{v[x]}{2}^2}{\norm{v[x]}{1}^2} = \frac{1}{4\left(1 - \tfrac{1}{m} \right)}$.
    We then have that
    \[
        \gamma_p = \max\!\set{1, \frac{p}{\log\!\left(\frac{e^2 m}{4\left(1 - \tfrac{1}{m} \right)}\right)}}
            \; .
    \]
    If $\gamma_p = 1$ then the result follow by \Cref{thm:sampling-moments-lower-bound}, so we assume that $\gamma_p > 1$.
    We let $S = [1 + \floor{\gamma_p}]^{c - 1} \times \Sigma$ and formally define the value function, $v \colon \Sigma^c \times [m] \to \R$, by,
    \begin{align*}
        v(x, j) = \begin{cases}
            0                &\text{if $x \not\in S$} \\
            1 - \tfrac{1}{m} &\text{if $x \in S$ and $j = 0$} \\
            - \tfrac{1}{m}   &\text{otherwise}
        \end{cases} \; .
    \end{align*}
    For every $i \in [c - 1]$, we define $A_i$ to be the event that $T(i, \alpha) = T(i, 0)$ for all $\alpha \in [1 + \floor{\gamma_p}]$.
    We note that
    \begin{align*}
        \prb{A_i} = \left(\frac{1}{m}\right)^{\floor{\gamma_p}} \ge \left(\frac{1}{m}\right)^{\gamma_p}
            = \exp\!\left(-p \frac{\log(m)}{\log\!\left(\frac{e^2 m}{4\left(1 - \tfrac{1}{m} \right)}\right)}\right)
            \ge e^{-p}
        \; .
    \end{align*}
    We then get that $\prb{\bigwedge_{i \in [c - 1]} A_i} \ge e^{-p(c - 1)}$.

    If $\bigwedge_{i \in [c - 1]} A_i$ happens then we can set $j^* = \bigxor_{i \in [c - 1]} T(i, 0)$.
    Now we define the value function $v' \colon \Sigma \times [m] \to \R$, by
    \begin{align*}
        v'(\alpha, j) = \begin{cases}
            1 - \tfrac{1}{m} &\text{if $j = j^*$} \\
            - \tfrac{1}{m}   &\text{otherwise}
        \end{cases} \; .
    \end{align*}
    We then get that
    \begin{align*}
        \sum_{x \in \Sigma^c} v(x, h(x))
            = \sum_{x \in S} v(x, h(x))
            = \sum_{\alpha \in \Sigma} \left(1 + \floor{\gamma_p}\right)^{c - 1} v'(\alpha, T(c-1, \alpha))
        \; .
    \end{align*}
    This implies that
    \begin{align*}
        \pnorm{\sum_{x \in \Sigma^c} v(x, h(x))}{p}
            &= \pnorm{\sum_{x \in S} v(x, h(x))}{p}
            \\&= \left( \ep{\indicator{\bigwedge_{i \in [c - 1]} A_i}\left(\sum_{x \in S} v(x, h(x))\right)^p} + \ep{\indicator{\left(\bigwedge_{i \in [c - 1]} A_i\right)^{c}}\left(\sum_{x \in S} v(x, h(x))\right)^p} \right)^{1/p}
            \\&\ge \pnorm{\indicator{\bigwedge_{i \in [c - 1]} A_i}\sum_{x \in S} v(x, h(x))}{p}
            \\&= \prb{\bigwedge_{i \in [c - 1]} A_i}^{1/p} \epcond{\left( \sum_{\alpha \in \Sigma} \left(1 + \floor{\gamma_p}\right)^{c - 1} v'(\alpha, T(c-1, \alpha)) \right)^{p}}{\bigwedge_{i \in [c - 1]} A_i}^{1/p}
            \\&\ge e^{-c} \left(1 + \floor{\gamma_p}\right)^{c - 1} \epcond{\left(\sum_{\alpha \in \Sigma} v'(\alpha, T(c-1, \alpha))\right)^{p}}{\bigwedge_{i \in [c - 1]} A_i}^{1/p}
        \; .
    \end{align*}
    Now we use \Cref{thm:sampling-moments-lower-bound} to get that
    \begin{align*}
        \pnormcond{\sum_{\alpha \in \Sigma} v'(\alpha, T(c-1, \alpha))}{p}{\bigwedge_{i \in [c - 1]} A_i}
            \ge L \Psi_p(1, \sigma_{v'}^2) 
        \; ,
    \end{align*}
    where $\sigma_{v'}^2 = \frac{1}{m}\left(1 - \frac{1}{m}\right)\abs{\Sigma}$.
    Combining it all we have that
    \begin{align*}
        \pnorm{\sum_{x \in \Sigma^c} v(x, h(x))}{p}
            &\ge e^{-c} \left(1 + \floor{\gamma_p}\right)^{c - 1} L \Psi_p\left(1, \frac{1}{m}\left(1 - \frac{1}{m}\right)\abs{\Sigma}\right)
            \\&= e^{-c} L \Psi_p\left(\left(1 + \floor{\gamma_p}\right)^{c - 1}, \left(1 + \floor{\gamma_p}\right)^{2(c - 1)} \frac{1}{m}\left(1 - \frac{1}{m}\right)\abs{\Sigma}\right)
            \\&\ge e^{-c} L \Psi_p\left(\gamma_p^{c - 1}, \gamma_p^{c - 1} \frac{1}{m}\left(1 - \frac{1}{m}\right)\abs{\Sigma} \left(1 + \floor{\gamma_p}\right)^{c - 1}\right)
        \; .
    \end{align*}
    The equality follows since $\Psi_p(\lambda M, \lambda^2 \sigma^2) = \lambda \Psi_p(M, \sigma^2)$.
    We have that $\frac{1}{m}\left(1 - \frac{1}{m}\right)\abs{\Sigma} \left(1 + \floor{\gamma_p}\right)^{c - 1} = \sigma_v^2$ so this finishes the proof.
\end{proof}

%% file: concentration/query-element.tex
The goal of the section is to prove that the result holds even when you condition on a query element.
We start by proving the result for simple tabulation.

\restateThmSimpleTabQuery
\begin{proof}
    We start by defining a partition of $\Sigma^c \setminus \set{q}$.
    Write $q = (\alpha_0, \ldots, \alpha_{c - 1})$ and for $I \subsetneq [c]$, we define $G_I \subseteq \Sigma^c$ by
    \begin{align*}
        G_I = \setbuilder{x \in \Sigma^c \setminus \set{q}}{\forall i \in I : (i, \alpha_i) \in x \wedge \forall i \not\in I : (i, \alpha_i) \not\in x}
    \end{align*}
    This clearly gives a partition of $\Sigma^c \setminus \set{q}$.
    Using this we obtain,
    \begin{align*}
        \pnormcond{V^{\text{\normalfont simple}}_{v, q}}{p}{h(q)}
            &= \pnormcond{\sum_{x \in \Sigma^c \setminus \set{q}} v(x, h(x), h(q))}{p}{h(q)}
            \\&= \pnormcond{\pnormcond{\sum_{x \in \Sigma^c \setminus \set{q}} v(x, h(x), h(q))}{p}{(h((i, \alpha_i)))_{i \in [c]} }}{p}{h(q)}
            \\&\le \pnormcond{ \sum_{I \subsetneq [c]} \pnormcond{\sum_{x \in G_I} v(x, h(x), h(q))}{p}{(h((i, \alpha_i)))_{i \in [c]} }}{p}{h(q)}
        \; .
    \end{align*}
    Now using \Cref{thm:simple-tab-moments} we get that,
    \begin{align*}
        &\pnormcond{ \sum_{I \subsetneq [c]} \pnormcond{\sum_{x \in G_I} v(x, h(x), h(q))}{p}{(h((i, \alpha_i)))_{i \in [c]} }}{p}{h(q)}
            \\&\le \pnormcond{ \sum_{I \subsetneq [c]} \Psi_p\left( K_{c - \abs{I}} \gamma_p^{c - 1 - \abs{I}} M_{v, q}, K_{c - \abs{I}} \gamma_p^{c - 1 - \abs{I}} \tfrac{1}{m} \sum_{x \in G_I} \sum_{j \in [m]} v(x, j, h(q))  \right) }{p}{h(q)}
            \\&= \sum_{I \subsetneq [c]} \Psi_p\left( K_{c - \abs{I}} \gamma_p^{c - 1 - \abs{I}} M_{v, q}, K_{c - \abs{I}} \gamma_p^{c - 1 - \abs{I}} \tfrac{1}{m} \sum_{x \in G_I} \sum_{j \in [m]} v(x, j, h(q))  \right)
            \\&\le \Psi_p\left( 2^{2c} K_{c} \gamma_p^{c - 1} M_{v, q}, 2^{2c} K_c \gamma_p^{c - 1} \sigma^2_{v, q}  \right)
        \; .
    \end{align*}
    The last bound follows by using \cref{eq:psi-reverse-growth} from \Cref{lem:psi-properties}.
\end{proof}

We now prove the result for mixed tabulation when conditioning on a query element.
\restateThmMixedTabQuery

\begin{proof}
    We start by defining $\mathcal{H} = \sigma((h_1((i, \alpha_i)), h_2((i, \alpha_i)))_{i \in [c]})$.
    Now we get that,
    \begin{align*}
        &\pnormcond{V^{\text{\normalfont mixed}}_{v, q}}{p}{h(q)}
            \\&= \pnormcond{\sum_{x \in \Sigma^c \setminus \set{q}} v(x, h(x), h(q))}{p}{h(q)}
            \\&= \pnormcond{\pnormcond{\sum_{x \in \Sigma^c \setminus \set{q}} v(x, h(x), h(q))}{p}{\mathcal{H} }}{p}{h(q)}
            \\&\le \pnormcond{\pnormcond{\sum_{\alpha \in \Sigma \setminus \set{h_2(q)}} \eps_2(\alpha) \sum_{x \in \Sigma^c} \eps_1(x) v(x, h_1(x) \xor h_3(\alpha), h(q)) \indicator{h_2(x) = \alpha} }{p}{\mathcal{H} }}{p}{h(q)}
                \\&+ \pnormcond{\pnormcond{\sum_{x \in \Sigma^c \setminus \set{q}} \eps_2(h_2(q)) \eps_1(x) v(x, h_1(x) \xor h_3(h_2(q)), h(q)) \indicator{h_2(x) = h_2(q)} }{p}{\mathcal{H} }}{p}{h(q)}
    \end{align*}
    For the second term we use \Cref{thm:simple-tab-query} to get that,
    \begin{align*}
        &\pnormcond{\pnormcond{\sum_{x \in \Sigma^c \setminus \set{q}} \eps_2(h_2(q)) \eps_1(x) v(x, h_1(x) \xor h_3(h_2(q)), h(q)) \indicator{h_2(x) = h_2(q)} }{p}{\mathcal{H} }}{p}{h(q)}
            \\&\qquad\qquad\qquad\le \Psi_p\left( 2^{2c} K_{c} \gamma_p^{c - 1} M_{v, q}, 2^{2c} K_c \gamma_p^{c - 1} \sigma^2_{v, q}  \right)
            \\&\qquad\qquad\qquad\le \Psi_p\left( 2^{2c} K_{c} \gamma_p^{c} M_{v, q}, 2^{2c} K_c \gamma_p^{c} \sigma^2_{v, q}  \right)
    \end{align*}
    For the first term we use the same proof technique as in the proof of \Cref{thm:simple-tab-query}.
    Write $q = (\alpha_0, \ldots, \alpha_{c - 1})$ and for $I \subsetneq [c]$, we define $G_I \subseteq \Sigma^c$ by
    \begin{align*}
        G_I = \setbuilder{x \in \Sigma^c \setminus \set{q}}{\forall i \in I : (i, \alpha_i) \in x \wedge \forall i \not\in I : (i, \alpha_i) \not\in x}
    \end{align*}
    This clearly gives a partition of $\Sigma^c \setminus \set{q}$.
    Using this we obtain,
    \begin{align*}
        &\pnormcond{\sum_{\alpha \in \Sigma \setminus \set{h_2(q)}} \eps_2(\alpha) \sum_{x \in \Sigma^c} \eps_1(x) v(x, h_1(x) \xor h_3(\alpha), h(q)) \indicator{h_2(x) = \alpha} }{p}{\mathcal{H} }
            \\&\le \sum_{I \subsetneq [c]} \pnormcond{\sum_{\alpha \in \Sigma \setminus \set{h_2(q)}} \eps_2(\alpha) \sum_{x \in G_I} \eps_1(x) v(x, h_1(x) \xor h_3(\alpha), h(q)) \indicator{h_2(x) = \alpha} }{p}{\mathcal{H} }
            \\&= \sum_{I \subsetneq [c]} \pnormcond{\sum_{\alpha \in \Sigma \setminus \set{h_2(q)}} \eps_2(\alpha) \sum_{x \in G_I} \eps_1(x) v(x, h_1(x) \xor h_3(\alpha), h'(q)) \indicator{h_2(x) = \alpha} }{p}{\mathcal{H} }
            \\&\le \sum_{I \subsetneq [c]} \pnormcond{\sum_{\alpha \in \Sigma} \eps_2(\alpha) \sum_{x \in G_I} \eps_1(x) v(x, h_1(x) \xor h_3(\alpha), h'(q)) \indicator{h_2(x) = \alpha} }{p}{\mathcal{H} }
        \; .
    \end{align*}
    We then use \Cref{thm:mixed-tab-moments} to get that,
    \begin{align*}
        &\sum_{I \subsetneq [c]} \pnormcond{\sum_{\alpha \in \Sigma} \eps_2(\alpha) \sum_{x \in G_I} \eps_1(x) v(x, h_1(x) \xor h_3(\alpha), h'(q)) \indicator{h_2(x) = \alpha} }{p}{\mathcal{H} }
            \\&\qquad\qquad\qquad\le 2^c \Psi_p\left( K_{c} \gamma_p^{c} M_{v, q}, K_c \gamma_p^{c} \sigma^2_{v, q}  \right)
            \\&\qquad\qquad\qquad\le \Psi_p\left( 2^{2c} K_{c} \gamma_p^{c} M_{v, q}, 2^{2c} K_c \gamma_p^{c} \sigma^2_{v, q}  \right)
        \; .
    \end{align*}
    Now combining everything we get that,
    \begin{align*}
        \pnormcond{V^{\text{\normalfont mixed}}_{v, q}}{p}{h(q)}
            \le 2 \Psi_p\left( 2^{2c} K_{c} \gamma_p^{c} M_{v, q}, 2^{2c} K_c \gamma_p^{c} \sigma^2_{v, q}  \right)
            \le \Psi_p\left( 2^{2c + 2} K_{c} \gamma_p^{c} M_{v, q}, 2^{2c + 2} K_c \gamma_p^{c} \sigma^2_{v, q}  \right)
        \; ,
    \end{align*}
    which finishes the proof.
\end{proof}

%% file: appendix/mixed.tex
\section{Statistics over \texorpdfstring{$k$}{k}-partitions using mixed tabulation}\label{sec:k-partitions}
We will now describe how mixed tabulation hashing is used for statistics over
$k$-partitions. The description of the
application is largely copied from~\cite{Dahlgaard2015} but now we can use our new strong
concentration bounds for mixed tabulation hashing to obtain stronger results.  

We consider the generic approach where a hash function is used to
$k$-partition a set into $k$ bins. Statistics are computed on each bin,
and then all these statistics are combined so as to get good
concentration bounds. This approach was introduced by Flajolet and
Martin \cite{Flajolet85counting} under the name \emph{stochastic
  averaging} to estimate the number of distinct elements in a data
stream. Today, a more popular estimator of this quantity is the
HyperLogLog counter, which is also based on $k$-partitioning
\cite{Flajolet07hyperloglog,Heule13hyperloglog}. These types of
counters have found many applications, e.g., to estimate the
neighborhood function of a graph with all-distance sketches
\cite{boldi11hyperanf,Cohen14ads}. Later $k$-partitions were used for
set similarity in large-scale machine learning by Li et al.
\cite{li12oneperm,Shrivastava14oneperm,Shrivastava14densify}.
Under the name one-permutation hashing, Li et al.~used $k$-partitions
to gain a factor $k$ in speed within the
classic MinHash framework of Broder et al. \cite{broder97onthe,broder98minwise}.

We will use MinHash for frequency estimation as an
example to illustrate how mixed tabulation yields good statistics over $k$-partitions:
suppose we have a fully random hash function
applied to a set $X$ of red and blue balls.  We want to estimate the
fraction $f$ of red balls. The idea of the MinHash algorithm is to
sample the ball with the smallest hash value. With a fully-random hash
function, this is a uniformly random sample from $X$, and it is red with
probability $f$.
For better concentration, we may use {\em $k$ independent repetitions}: we
repeat the experiment
$k$ times with $k$ independent hash functions. This yields a multiset $S$
of $k$ samples with replacement from $X$.  The fraction of red balls
in $S$ concentrates around $f$ and the error probability falls
exponentially in $k$.

Consider now the alternative experiment based on {\em
  $k$-partitioning}, assuming that $k$ is a power of two.
We use a single hash function, where the first
$\log k$ bits of the hash value partitions $X$ into $k$ bins, and
then the remaining bits are used as a {\em local hash value} within
the bin. We pick the ball with the smallest (local) hash value in each
bin. This is a sample $S$ from $X$ without replacement, and again, the
fraction of red balls in the non-empty bins is concentrated around $f$
with exponential concentration bounds. We note that there are some
differences. We do get the advantage that the samples are without
replacement, which means better concentration. On the other hand, we
may end up with fewer samples if some bins are empty.

The big difference between the two schemes is that the second one runs
$\Omega(k)$ times faster. In the first experiment, each ball
participated in $k$ independent experiments, but in the second one with
$k$-partitioning, each ball picks its bin, and then only
participates in the local experiment for that bin. Thus, time-wise, we get
$k$ experiments for the price of one. Handling each ball, or key, in constant
time is important in applications of high volume streams.

The above approach, however, requires a very powerful hash function.
The main issue is the overall $k$-partitioning distribution between
bins. It could be that if we get a lot of red balls in one bin, then
this would be part of a general clustering of the red balls on a few
bins (examples showing how such systematic clustering can happen with
simple hash functions are given in \cite{patrascu10kwise-lb}). This
clustering would disfavor the red balls in the overall average even if the
sampling in each bin was uniform and independent. This is an issue
of non-linearity, e.g., if there are already
more red than blue balls in a bin, then doubling their number only
increases their frequency by at most $3/2$. We note that $k$-independence
does not by itself suffices to give any guarantees in this situation.

\paragraph*{Using mixed tabulation hashing}
We will now show mixed tabulation can be applied in the above application
using selective full randomness from Theorem \ref{thm:mix-indep}
in tandem with the concentration from \req{eq:mix-tab-Y}. Combined they
give:

\begin{quote}
  Let $h:\Sigma^c\to \{0,1\}^\ell$ be a mixed tabulation hash function
  using $d$ derived characters. Let $M$ be an $\ell$-bit mask with
  don't cares. For a given key set $X\subseteq \Sigma^c$, let $Y$ be
  the set of keys from $X$ with hash values matching $M$. If
  $\Omega(|\Sigma|)=\E[Y]=|\Sigma|(1-\Omega(1))$, then with probability $1 - O(|\Sigma|^{1-\floor{d/2}})$,
  the free bits of the hash values in $Y$ are fully random. Moreover, with this probability,
$|Y|=\E[Y](1\pm O(\sqrt{(\log |\Sigma|)/|\Sigma|}))$.
\end{quote}
For the $k$-partitioning, we need $k$ to be bounded relative to the
alphabet size as $k\leq \frac{|\Sigma|}{4d\ln|\Sigma|}$. Recall
here our hash tables need space $O(|\Sigma|)$ which then
has to be a logarithmic factor bigger than $k$. This
may motivate 16-bit characters rather than 8-bit characters, but
on modern computers this still fits in fast cache.

We a set $X$ of $n$ red and blue balls, and for simplicity in our analysis, we assume that the ratio between
them is at most a factor 2. If $n\leq |\Sigma|/2$, Theorem
\ref{thm:mix-indep} with all bits free states that the balls hash fully-randomly with
high probability. This means that \emph{any analysis} based on
fully-random hashing applies directly. Otherwise let
$q=\ceil{\log(n /(2|\Sigma|/3)}$. We study the
set of balls $Y$ that has $q$ leading zeros in their
local hash value. These $q$ bits are the fixed bits in the mask Theorem \ref{thm:mix-indep}.
The expected size of $Y$ is
between $|\Sigma|/3$ and $2|\Sigma|/3$, so by Theorem \ref{thm:mix-indep}, w.h.p., all
other bits in the hash values of $Y$ are fully random, including
both the global index of $\log k$ bits and the tail of bits in
the local hash value after the $q$ leading zeros. Moreover,
\req{eq:mix-tab-Y}, imply that the size of $Y$ is at most
a factor
\[1\pm O(\sqrt{(\log|\Sigma|)/|\Sigma|})=1\pm O(\sqrt{1/k})\]
from
its mean, and the same holds for the ratio between red and blue balls
in our simple case where there are more than $n/4$ of each.

We claim that, w.h.p, all bins get a ball from $Y$. We
know that $\E[|Y|]\geq |\Sigma|/3\geq 2kd \ln|\Sigma|/3$, so w.h.p.,
$|Y|\geq kd \ln|\Sigma|/2$. Consider now any bin $i$. Since the
bin indices are fully random, the probability
that $i$ gets no ball from $Y$ is at most $(1-1/k)^{kd(\ln|\Sigma|/2)}<1/|\Sigma|^{d/2}$. A union bound now implies that all bins get a ball from
$Y$ with probability $1-k/|\Sigma|^{d/2}\geq 1-|\Sigma|^{1-d/2}$.

Since every bin contains a ball from $Y$ and these are the balls with
$q$ leading zeros in their local hash, the MinHash ball from $X$ in
each bin is from $Y$. Since all keys from $Y$ have at least $q$ leading
zeros in their hash values and all other bits are fully random, this
means that, w.h.p., the result of the experiment is exactly the same as
if we applied it to $Y$ using a fully random hash function, and
our concentration bounds imply that the ratio between red and blue
balls is well-preserved from $X$ to $Y$.

The important point in the above analysis is that it is only the analysis
that needs to know anything about $n$. This is important in more
complicated contexts, e.g., streaming where $n$ is not known in advance,
or in set similarity where red balls represent the intersection of sets
while the blue balls represent their symmetric difference, and where we
use the $k$-partitions to compare sets of many different sizes. We know
that mixed tabulation hashing is almost as good fully random hashing as long
as $k\leq \frac{|\Sigma|}{4d\ln|\Sigma|}$.

The above description is very similar to the one of Dahlgaard et
al.~\cite{Dahlgaard2015}. The difference is that we now have the right concentration bounds
for twisted tabulation, e.g., Dahlgaard et al., had the additional
constraint that $k\leq \frac{|\Sigma|}{\log
  |\Sigma|(\log\!\log|\Sigma|)^2}$, which is completely unnecessary.